\documentclass{lmcs}
\usepackage[utf8]{inputenc}
\pdfoutput=1

\usepackage{lastpage}
\lmcsdoi{18}{1}{12}
\lmcsheading{}{\pageref{LastPage}}{}{}%
{Dec.~01,~2016}{Jan.~19,~2022}{}

\keywords{%
  constraint satisfaction problem, universal algebra,
  dichotomy conjecture, Mal'tsev condition, Taylor term, absorbing subalgebra.
}

\ACMCCS{%
  [{\bf Theory of computation}]:
  Algebraic complexity theory,
  Complexity theory and logic}

\amsclass{
Primary 08A70;    
Secondary 03C05,  
          08A30,  
          08A40}  

\usepackage{hyperref}
\hypersetup{
    unicode=true,           
    pdftoolbar=true,        
    pdfmenubar=true,        
    pdffitwindow=false,     
    pdftitle={Universal Algebraic Methods for CSP},
    pdfauthor={Clifford Bergman and William DeMeo},
    pdfsubject={},
    pdfcreator={pdflatex with hyperref},
    pdfproducer={},
    pdfkeywords= {constraint satisfaction problem} {universal algebra}
                 {dichotomy conjecture} {Mal'tsev condition} {Taylor term}
                 {absorbing subalgebra} {universal algebra},
}
\usepackage{graphicx}
\usepackage{comment}

\usepackage{amsmath}
\usepackage{amssymb}
\usepackage{amsthm}
\usepackage{amscd}
\usepackage{graphicx}
\usepackage{mathtools}
\usepackage{bm}
\usepackage{latexsym,stmaryrd,mathrsfs,scalefnt,ifthen}
\usepackage[mathscr]{euscript}
\usepackage{url}
\usepackage{scalefnt}
\usepackage{tikz}
\usepackage{color}
\usepackage{booktabs}

\usepackage[smaller]{acronym}  
\usepackage{xspace}

\acrodef{ec}[EC]{edge-by-chain}
\acrodef{lics}[LICS]{Logic in Computer Science}
\acrodef{sat}[SAT]{satisfiability}
\acrodef{nae}[NAE]{not-all-equal}
\acrodef{ctb}[CTB]{cube term blocker}
\acrodef{tct}[TCT]{tame congruence theory}
\acrodef{wnu}[WNU]{weak near-unanimity}
\acrodef{CSP}[CSP]{constraint satisfaction problem}
\acrodef{MAS}[MAS]{minimal absorbing subuniverse}
\acrodef{MA}[MA]{minimal absorbing}
\acrodef{cib}[CIB]{commutative idempotent binar}
\acrodef{sd}[SD]{semidistributive}
\acrodef{NP}[NP]{nondeterministic polynomial time}
\acrodef{P}[P]{polynomial time}
\acrodef{PeqNP}[P $ = $ NP]{P is NP}
\acrodef{PneqNP}[P $ \neq $ NP]{P is not NP}

   \newboolean{todos}
   \setboolean{todos}{true}  

   \newboolean{arxiv}
   \setboolean{arxiv}{true}  
   \setboolean{arxiv}{false}  

   \newboolean{extralong}
   \setboolean{extralong}{true}  
   \setboolean{extralong}{false}  

   \newboolean{footnotes}
   \setboolean{footnotes}{true}  
   \setboolean{footnotes}{false}  

   \newboolean{draft}
   \setboolean{draft}{true}  
   \setboolean{draft}{false}

   \newcommand{\todo}[1]{\ifthenelse{\boolean{todos}}{%
       \noindent {\bf To do:} #1}{}}

\usepackage{accents}
\newcommand{\ubar}[1]{\ensuremath{\underaccent{\bar}{#1}}}
\newcommand{\overbar}[1]{\ensuremath{\bar{#1}}}

\newcommand{\etaR}{\ensuremath{\rho}}

\newcommand{\ar}{\ensuremath{\operatorname{ar}}}
\newcommand{\CSP}{\ensuremath{\operatorname{CSP}}}

\newcommand{\im}{\ensuremath{\operatorname{im}}}
\newcommand{\slt}{\ensuremath{\mathbf S_2}\xspace}

\newcommand{\ctb}{\acs{ctb}\xspace}
\newcommand{\lics}{\acs{lics}\xspace}
\newcommand{\ec}{\acs{ec}\xspace}
\newcommand{\csp}{\acs{CSP}\xspace}
\newcommand{\mas}{mass\xspace}
\newcommand{\masses}{masses\xspace}

\newcommand{\csps}{\acsp{CSP}\xspace}
\newcommand{\sd}{\acs{sd}\xspace}
\newcommand{\cib}{\acs{cib}\xspace}
\newcommand{\cibs}{\acsp{cib}\xspace}

\newcommand{\NP}{\acs{NP}\xspace}
\renewcommand{\P}{\acs{P}\xspace}
\newcommand{\mfA}{\ensuremath{\mathfrak{A}}}

\newcommand{\Inj}{\ensuremath{\operatorname{Inj}}}
\newcommand{\Proj}{\ensuremath{\operatorname{Proj}}}

\newcommand{\alg}[1]{\ensuremath{\mathbf{#1}}}
\newcommand{\var}[1]{\ensuremath{\mathcal{#1}}}

\newcommand{\Sub}{\ensuremath{\operatorname{Sub}}}
\newcommand{\Con}{\ensuremath{\operatorname{Con}}}
\newcommand{\Cg}{\ensuremath{\operatorname{Cg}}}

\newcommand{\sdp}{\ensuremath{\leq_{\mathrm{sd}}}}
\newcommand{\defn}[1]{\textit{#1}}
\newcommand{\N}{\ensuremath{\mathbb{N}}}

\newcommand{\<}{\ensuremath{\langle}}
\renewcommand{\>}{\ensuremath{\rangle}}
\newcommand{\bR}{\ensuremath{\mathbf{R}}}

\newcommand{\sansA}{\ensuremath{\mathsf{A}}}

\newcommand{\sansC}{\ensuremath{\mathsf{C}}}
\newcommand{\sansH}{\ensuremath{\mathsf{H}}}
\newcommand{\sansO}{\ensuremath{\mathsf{O}}} 

\newcommand{\sansClo}{\ensuremath{\operatorname{\mathsf{Clo}}}}
\newcommand{\sansS}{\ensuremath{\mathsf{S}}}

\newcommand\ubA{\ensuremath{\ubar{\mathbf{A}}}}
\newcommand\uA{\ensuremath{\ubar{A}}}

\newcommand{\bs}{\ensuremath{\mathbf{s}}}
\newcommand{\sA}{\ensuremath{\mathscr{A}}}
\newcommand{\sB}{\ensuremath{\mathscr{B}}}
\newcommand{\sC}{\ensuremath{\mathscr{C}}}
\newcommand{\sF}{\ensuremath{\mathscr{F}}}
\newcommand{\sH}{\ensuremath{\mathscr{H}}}
\newcommand{\sI}{\ensuremath{\mathscr{I}}}

\newcommand{\sP}{\ensuremath{\mathscr{P}}}
\newcommand{\sR}{\ensuremath{\mathscr{R}}}
\newcommand{\sS}{\ensuremath{\mathscr{S}}}
\newcommand{\sT}{\ensuremath{\mathscr{T}}}
\newcommand{\sV}{\ensuremath{\mathscr{V}}}

\newcommand{\bA}{\ensuremath{\mathbf{A}}}
\newcommand{\ba}{\ensuremath{\mathbf{a}}}
\newcommand{\bB}{\ensuremath{\mathbf{B}}}
\newcommand{\bb}{\ensuremath{\mathbf{b}}}
\newcommand{\bC}{\ensuremath{\mathbf{C}}}

\newcommand{\bD}{\ensuremath{\mathbf{D}}}

\newcommand{\br}{\ensuremath{\mathbf{r}}}
\newcommand{\bS}{\ensuremath{\mathbf{S}}}
\newcommand{\bT}{\ensuremath{\mathbf{T}}}

\newcommand{\bu}{\ensuremath{\mathbf{u}}}
\newcommand{\bv}{\ensuremath{\mathbf{v}}}

\newcommand{\bx}{\ensuremath{\mathbf{x}}}
\newcommand{\vy}{\ensuremath{\mathbf{y}}}
\newcommand{\bz}{\ensuremath{\mathbf{z}}}
\newcommand{\bzero}{\ensuremath{\mathbf{0}}}
\newcommand{\Sol}{\ensuremath{\operatorname{Sol}}}

\newcommand{\nn}{\ensuremath{\ubar{n}}}
\newcommand{\mm}{\ensuremath{\ubar{m}}}
\newcommand{\pp}{\ensuremath{\ubar{p}}}
\newcommand{\kk}{\ensuremath{\ubar{k}}}

\newcommand{\uzero}{\ensuremath{\ubar{0}}}

\newcommand{\kplus}{\ensuremath{\underline{k+1}}}

\newcommand{\AAn}{\ensuremath{A^{(A^n)}}}

\newcommand{\meet}{\ensuremath{\wedge}}
\newcommand{\Meet}{\ensuremath{\bigwedge}}
\renewcommand{\Join}{\ensuremath{\bigvee}}
\newcommand{\join}{\ensuremath{\vee}}
\newcommand{\onto}{\ensuremath{\twoheadrightarrow}}

\newcommand{\absorbing}{\ensuremath{\mathrel{\triangleleft}}}
\newcommand{\minabsorbing}{\ensuremath{\mathrel{\triangleleft\triangleleft}}}
\newcommand{\btheta}{\ensuremath{\bm{\theta}}}
\newcommand{\malcev}{Mal'tsev\xspace}
\renewcommand{\leq}{\ensuremath{\leqslant}}
\renewcommand{\geq}{\ensuremath{\geqslant}}
\newcommand{\CC}[3]{\ensuremath{\sansC(#1, #2; #3)}} 
\makeatletter
\@ifundefined{C}{%
\newcommand{\C}[2]{\ensuremath{\sansC(#1, #2)}}      
}{%
\renewcommand{\C}[2]{\ensuremath{\sansC(#1, #2)}}      
}
\newcommand{\myprod}{\ensuremath{\prod}}

\newcommand\restr[2]{{
  \left.\kern-\nulldelimiterspace
  #1 
  \vphantom{\big|} 
  \right|_{#2} 
  }}

\usepackage{pifont}

\let\:\colon%
\newcommand{\card}[1]{|#1|}
\DeclareMathOperator{\size}{size}

\newcommand{\Sl}{\ensuremath{Sl}} 
\newcommand{\Sq}[1]{\ensuremath{\mathbf{Sq}_{#1}}}
\newcommand{\sdm}{\sd-\meet\xspace}
\newcommand{\casespec}[1]{\medskip\noindent\textbf{#1}.\enspace}

\newcommand\kcapsigma{\ensuremath{\kk \, \cap\, \im \sigma}}

\usepackage{algorithm}  
\usepackage[noend]{algpseudocode}

\begin{document}

\title[Universal Algebraic Methods for CSP]{Universal Algebraic Methods for\texorpdfstring{\\}{}
  Constraint Satisfaction Problems}

\author[C.~Bergman]{Clifford Bergman\rsuper{a}}
\address{Department of Mathematics, Iowa State University, Ames, IA, USA}
\email{cbergman@iastate.edu}
\author[W.~DeMeo]{William DeMeo\rsuper{b}}
\address{Thmpr Laboratory, New York, NY, USA}
\email{williamdemeo@gmail.com}

\thanks{Research of both authors partially supported by the National Science Foundation, Grant No.~1500218.}

\begin{abstract}
\noindent After substantial progress over the last 15 years, the ``algebraic CSP-dichotomy conjecture'' reduces to the following: every local constraint satisfaction problem (CSP) associated with a finite idempotent algebra is tractable if and only if the algebra has a Taylor term operation. Despite the tremendous achievements in this area (including recently announce proofs of the general conjecture),
there remain examples of small algebras with just a single binary operation whose CSP resists direct classification as either tractable or NP-complete using known methods.  In this paper we present some new methods for approaching such problems, with particular focus on those techniques that help us attack the class of finite algebras known as ``commutative idempotent binars'' (CIBs). We demonstrate the utility of these methods by using them to prove that every CIB of cardinality at most 4 yields a tractable CSP\@.
\end{abstract}

\maketitle

\section{Introduction}\label{sec:introduction}
The ``\csp-dichotomy conjecture'' of Tom{\'a}s Feder and Moshe Vardi~\cite{MR1630445} asserts that every constraint satisfaction
problem (\csp) over a fixed finite constraint language is either \NP-complete or tractable. (The current status of the conjecture is discussed in the note at the end of this introduction.)

A discovery of Jeavons, Cohen and Gyssens in~\cite{MR1481313}---later refined by Bulatov, Jeavons and Krokhin in~\cite{MR2137072}---was the ability to transfer the question of the complexity of the \csp over a set of relations to a question of algebra.  Specifically, these authors showed that the complexity of any particular \csp depends solely on the \emph{polymorphisms} of the constraint relations, that is, the functions preserving all the constraints. The transfer to universal algebra was made complete by Bulatov, Jeavons, and Krokhin in recognizing that with any set $\sR$ of constraint relations one can associate an algebra $\bA(\sR)$ whose operations consist of the polymorphisms of $\sR$.

Following this, the \csp-dichotomy conjecture of Feder and Vardi was recast as a universal algebra problem once it was recognized that the conjectured sharp dividing line between those \csps that are \NP-complete and those that are tractable was seen to depend upon universal algebraic properties of the associated algebra. One such property is the existence of a Taylor term (defined in Section~\ref{ssec:term-ops}). Roughly speaking, the ``algebraic \csp-dichotomy conjecture'' is the following: The \csp associated with a finite idempotent algebra is tractable if and only if the algebra has a Taylor term operation in its clone. We state this more precisely (postponing technical definitions) as follows:

\begin{conj}\label{dichotomy-conjecture}
If $\bA$ is an algebra with a Taylor term in its clone and if $\sR$ is a finite set of relations compatible with $\bA$, then $\CSP(\sR)$ is tractable.  Conversely, if $\bA$ is an idempotent algebra with no Taylor term in its clone, then there exists a finite set $\sR$ of relations such that $\CSP(\sR)$ is \NP-complete.
\end{conj}

The second sentence of the conjecture was already established
in~\cite{MR2137072}. The proof of the assertion in the first sentence was
recently announced (independently) by Andrei Bulatov~\cite{Bulatov2017} and
Dmitriy Zhuk~\cite{Zhuk2017,MR4152441}, as discussed in the note at the end of this section.  The purpose of the present work is to provide new methods for establishing the truth of first sentence of the conjecture for certain special classes of algebras. Another goal is to give concrete examples demonstrating how these methods work.

Algebraists have identified two quite different techniques for proving that a finite idempotent algebra is tractable. One, often called the ``local consistency algorithm,'' works for any finite algebra lying in an idempotent, congruence-meet-semidistributive variety (\sdm for short). See~\cite{MR2648455} or~\cite{MR2893395} for details. The other, informally called the ``few subpowers technique,'' applies to any finite idempotent algebra possessing an edge term~\cite{MR2678065}. Definitions of these terms appear in Section~\ref{ssec:edge-sdm}.  While these two algorithms cover a wide class of interesting algebras, they were not enough to resolve the full dichotomy conjecture above.  Several researchers attempted to combine the two existing approaches in a way that captures the outstanding cases. In this paper, we do as well.

For example, suppose that $\bA$ is a finite, idempotent algebra possessing a congruence, $\theta$, such that $\bA/\theta$ lies in an \sdm variety and every congruence class of $\theta$ has an edge term. It is not hard to show that $\bA$ has a Taylor term. Can local consistency be combined with few subpowers to prove that $\bA$ is tractable?

We can formalize this idea as follows. Let $\var{V}$ and $\var{W}$ be idempotent varieties. The \defn{\malcev product} of $\var{V}$ and $\var{W}$ is the class
\begin{equation*}
\var{V} \circ \var{W} = \{\,\bA : (\exists\; \theta\in \Con(\bA))\;\;\bA/\theta \in \var{W} \mathrel{\&} (\forall a\in A)\;a/\theta \in \var{V}\,\}.
\end{equation*}
The \malcev product $\var{V}\circ \var{W}$ is always an idempotent quasivariety, but is generally not closed under homomorphic images. A worthwhile goal for future work would be to prove that if all finite members of $\var{V}$ and $\var{W}$ are tractable, then the same will hold for all finite members of $\var{V}\circ \var{W}$. Tellingly, Freese and McKenzie show in~\cite{FreeseMcKenzie2016} that a number of important properties are preserved by \malcev product.

\begin{thmC}[\cite{FreeseMcKenzie2016}]%
\label{thm:robust}
Let $\var{V}$ and $\var{W}$ be idempotent varieties. For each of the following properties, $P$\!, if both $\var{V}$ and $\var{W}$ have $P$\!, then so does $\sansH(\var{V}\circ \var{W})$:
\begin{enumerate}
\item being idempotent;
\item having a Taylor term;
\item being \sdm;
\item having an edge term.
\end{enumerate}
\end{thmC}

\noindent
It follows from Theorem~\ref{thm:robust} that if both $\var{V}$ and $\var{W}$ are \sdm, or both have an edge term, then every finite member of $\sansH(\var{V}\circ \var{W})$ is tractable. So the next step would be to consider $\var{V}$ with one of these properties and $\var{W}$ with the other. But even if we could prove that tractability is preserved by \malcev product, that would not be enough. In particular, simple algebras would remain problematic.

Interestingly, Barto and Kozik developed an approach that is particularly well-suited to dealing with simple algebras in this context.  In \lics'10~(\cite{MR2953899}), these researchers proved a powerful ``Absorption Theorem'' for products of two ``absorption-free'' algebras in a Taylor variety. At a more recent workshop~\cite{Barto-shanks}, Barto announced further joint work with Kozik on a general ``Rectangularity Theorem'' that gives conditions under which a subdirect product of simple nonabelian algebras contains a full product of ``minimal absorbing'' subalgebras. (Terms in quotes and other technicalities are defined below.)

In Section~\ref{sec:tayl-vari-rect} of the present paper we state and prove a version of Barto and Kozik's Rectangularity Theorem.  In Section~\ref{sec:csps-comm-idemp}, we apply this tool, together with some techniques involving \malcev products and other new decomposition strategies,
to prove that every commutative, idempotent binar of cardinality at most~4 is tractable.

\noindent\textbf{Note}. Subsequent to the initial submission of this manuscript,
Andrei Bulatov~\cite{Bulatov2017} and Dmitriy Zhuk~\cite{Zhuk2017} independently
announced a resolution of Conjecture~\ref{dichotomy-conjecture}. Even more
recently, Zhuk's proof appeared in~\cite{MR4152441}. The proofs are ingenious, lengthy, and
exceedingly complex, but overall follow a strategy similar to the one outlined above.
Nonetheless, we feel the methods and results presented here are useful for tackling specific problems, and doing so more directly than would be possible using the approaches offered by proofs of the general conjecture.

\section{Definitions and Notations}
\subsection{Notation for projections, scopes, and kernels}\label{sec:proj-scop-kern}
Here we collect definitions and notations used in the sequel.  Some of these will be reintroduced more carefully later, as needed.

An operation $f \colon A^n \rightarrow A$ is called \defn{idempotent} provided $f(a, a, \dots, a) = a$ for all $a \in A$. Examples of idempotent operations are the projection functions and these play an important role in later sections, so we start by introducing a flexible notation for them.

We define natural numbers as usual and denote them by $\uzero := \emptyset$, $\nn := \{0, 1, 2, \dots, n-1\}$.
Given sets $A_0$, $A_1$, $\dots$, $A_{n-1}$, An element $\ba \in \myprod_{\nn} A_i$ of their Cartesian product may be viewed as an ordered $n$-tuple by simply listing its values, $\ba = (\ba(0), \ba(1), \dots, \ba(n-1))$, or  as a function defined on index set $\nn$; the latter view is emphasized symbolically as follows:
\[
\ba \colon \nn \to \bigcup_{i\in \nn} A_i; \;\; i\mapsto \ba(i) \in A_i.
\]

If $\sigma\colon \kk \to \nn$ is a $k$-tuple of numbers in $\nn$, then we can compose an $n$-tuple $\ba$ in $\myprod_{i\in \nn} A_i$
with $\sigma$ yielding the $k$-tuple $\ba\circ \sigma$ in
$\myprod_{\kk}A_{\sigma(i)}$. Generally speaking, we will try to avoid
nonstandard notational conventions; two exceptions are the following:
\[
\uA := \prod_{i\in \nn} A_i \quad \text{ and } \quad \uA_{\sigma} := \prod_{i\in \kk} A_{\sigma(i)}.
\]

Now, if the $k$-tuple $\sigma\colon \kk \to \nn$ happens to be one-one, and if we let $p_\sigma$ denote the map $\ba \mapsto \ba\circ \sigma$, then $p_\sigma$ is the usual \emph{projection function} from $\uA$ onto $\uA_{\sigma}$. Thus, $p_{\sigma}(\ba)$ is a $k$-tuple whose $i$-th component is $(\ba\circ \sigma)(i) = \ba(\sigma(i))$. We will make frequent use of such projections, as well as their images under the covariant and contravariant powerset functors $\sP$ and $\overbar{\sP}$.   Indeed, we let $\Proj_\sigma\colon \sP(\uA) \to \sP(\uA_\sigma)$ denote the \emph{projection set function} defined for each $R \subseteq \uA$ by $\Proj_\sigma R = \sP(p_\sigma)(R) = \{p_\sigma(\bx) \mid \bx \in R\} = \{\bx\circ \sigma \mid \bx \in R\}$,
and we let $\Inj_\sigma\colon \overbar{\sP}(\uA_\sigma) \to \overbar{\sP}(\uA)$ denote the \emph{injection set function} defined for each
$S\subseteq \uA_\sigma$ by
\begin{equation}\label{eq:19}
  \Inj_\sigma S = \overbar{\sP}(p_\sigma)(S) = \{\bx \mid p_\sigma(\bx) \in S\}
   = \{\bx \in \uA \mid \bx\circ \sigma \in S\}.
\end{equation}
Of course, $\Inj_\sigma S$ is nothing more than the inverse image of the set $S$ with respect to the projection function $p_\sigma$. We sometimes use the shorthand $R_\sigma = \Proj_\sigma R$ and $S^{\overleftarrow{\sigma}} = \Inj_\sigma S$ for the projection and injection set functions, respectively.

A one-one function $\sigma \colon \kk \to \nn$ is called a \emph{scope} or \emph{scope function}. If
$R \subseteq \uA_{\sigma}$
is a subset of the projection of $\uA$ onto coordinates $(\sigma(0), \sigma(1), \dots, \sigma(k-1))$, then we call $R$ a \emph{relation on $\uA$ with scope $\sigma$}. The pair $(\sigma, R)$ is called a \emph{constraint}, and $R^{\overleftarrow{\sigma}}$ is the set of tuples in $\uA$ that \emph{satisfy} $(\sigma, R)$.

By $\eta_\sigma$ we denote the \emph{kernel} of the projection function $p_\sigma$, which is the following equivalence relation on $\uA$:
\begin{equation}\label{eq:60}
  \eta_\sigma = \{(\ba,\ba') \in \uA^2 \mid p_\sigma(\ba) = p_\sigma(\ba') \} = \{(\ba,\ba') \in \uA^2 \mid \ba \circ \sigma = \ba' \circ \sigma \}.
\end{equation}

More generally, if $\theta$ is an equivalence relation on the set $\myprod_{j\in \kk} A_{\sigma(j)}$, then we  define the equivalence relation $\theta_\sigma$ on the set $\uA$ 
by $\theta_\sigma = \{(\ba, \ba') \in \uA^2 \mid (\ba\circ \sigma) \mathrel{\theta} (\ba' \circ \sigma)\}$.
Thus,  $\theta_\sigma$ consists of pairs in $\uA^2$ that land in $\theta$ when projected onto coordinates in the scope $\sigma$. If $0_{\uA}$ is the least equivalence on {\uA}, then $\eta_\sigma$ is shorthand for $(0_{\uA})_\sigma$.

If the domain of $\sigma$ is a singleton, $\kk = \{0\}$, then $\sigma$ is just a one-element list, say, $\sigma= (j)$, and we write $p_j$ instead of $p_{(j)}$. Similarly, we write $\Proj_j$ instead of $\Proj_{(j)}$, $\eta_j$ instead of $\eta_{(j)}$, etc.
Some obvious consequences of these definitions are the following: $\bigvee_{j\in \nn}\eta_j =\uA^2$, $\eta_\sigma= \bigwedge_{j\in \sigma}\eta_j$, $\eta_{\nn} = \bigwedge_{j\in \nn}\eta_j = 0_{\uA}$, $\theta_\sigma = \bigwedge_{j\in \kk}\theta_{\sigma(j)}$.


\subsection{Notation for algebraic structures}
This section introduces the notation we use for algebras and related concepts. The reader should consult~\cite{MR2839398} for more details and background on general (universal) algebras and the varieties they inhabit.

\subsubsection{Product algebras}
Fix $n\in \N$, let $F$ be a set of operation symbols, and for each $i\in \nn$ let $\bA_i = \<A_i, F\>$ be an algebra of type $F$. Let $\ubA = \bA_0 \times \bA_1 \times \cdots \times \bA_{n-1}= \myprod_{\nn} \bA_i$ denote the product algebra. If $k \leq n$ and if $\sigma \colon \kk\to \nn$ is a one-one function, then we denote by $p_\sigma \colon \ubA \onto \myprod_{\kk}\bA_{\sigma(i)}$ the \defn{projection} of $\ubA$ onto the
``$\sigma$-factors'' of $\ubA$, which is an algebra homomorphism; thus the kernel $\eta_\sigma$ defined in~(\ref{eq:60}) is a congruence relation of $\ubA$.

\subsubsection{Term operations}\label{ssec:term-ops}
For a nonempty set $A$, we let $\sansO_A$ denote the set of all finitary operations on $A$. That is, $\sansO_A = \bigcup_{n\in \N}\AAn$. A \defn{clone of operations} on $A$ is a subset of $\sansO_A$ that contains all projection operations and is closed under general composition.  If $\bA = \<A, F^\bA\>$ denotes the algebra with universe $A$ and set of basic operations $F$, then $\sansClo (\bA)$ denotes the clone generated by $F$, which is also known as the \defn{clone of term operations} of $\bA$. We write $\sansClo_n (\bA)$ to denote the $n$-ary members of $\sansClo (\bA)$ (despite the fact that $\sansClo_n (\bA)$ is not a clone).

A \defn{strong \malcev condition} is a primitive positive sentence in the language of clones.  A \defn{\malcev condition} is a countably infinite disjunction of strong \malcev conditions $\{\varphi_i\}$ that get weaker as $i$ increases (that is, $\varphi_i \vdash \varphi_{i+1}$). A \defn{nontrivial \malcev condition} is one that fails in some variety (equivalently, fails in the variety of sets). An \defn{idempotent \malcev condition} is one that satisfies the following: for each term $f$ appearing in the condition, the identity $f(x, \dots, x) \approx x$ is a consequence of the identities in the \malcev condition. (See, e.g.,~\cite{MR2839398} or~\cite{MR3076179} for more details.)

Walter Taylor proved in~\cite{MR0434928} that a variety $\var{V}$ satisfies some nontrivial idempotent \malcev condition if and only if it satisfies one of the following form: for some~$n$, $\var{V}$ has an idempotent $n$-ary term  $t$ such that for each $i\in \nn$ there is an identity of the form
\[ t(\ast, \ldots, \ast, x, \ast, \ldots, \ast) \approx t(\ast, \ldots, \ast, y, \ast, \ldots, \ast) \]
true in $\var{V}$ where distinct variables $x$ and $y$ appear in the $i$-th position on either side of the identity.  Such a term $t$ is now commonly called a \defn{Taylor term}.

Throughout this paper we assume all algebras are finite (though some results may apply more generally). Starting in Section~\ref{sec:tayl-vari-rect}, we make the additional assumption that the algebras in question come from a single \defn{Taylor variety} $\var{V}$, by which we mean that $\var{V}$ has a Taylor term and every term of $\var{V}$ is idempotent.

\subsubsection{Subdirect products}
If $k, n \in \N$, if $\sA = (A_0, A_1, \dots, A_{n-1})$ is a list of sets, and if $\sigma \colon \kk \to \nn$ is a $k$-tuple of indices, then a relation $R$ over $\sA$ with scope $\sigma$ is a subset of the Cartesian product $A_{\sigma(0)} \times A_{\sigma(1)} \times \cdots \times A_{\sigma(k-1)}$.
Let $F$ be a set of operation symbols and for each $i\in \nn$ let $\bA_i = \<A_i, F\>$ be an algebra of type $F$. If $\ubA = \myprod_{\nn}\bA_i$ is the product of these algebras, then a relation $R$ over $\sA$ with scope $\sigma$ is called \defn{compatible with $\ubA$} if it is closed under the basic operations in $F$. In other words, $R$ is compatible if the induced algebra $\bR= \<R,F\>$  is a subalgebra of $\myprod_{\kk} \bA_{\sigma(j)}$. If $R$ is compatible with the product algebra and if the projection of $R$ onto each factor is surjective, then $\bR$ is called a \defn{subdirect product} of the algebras in the list $(\bA_{\sigma(0)}, \bA_{\sigma(1)}, \dots, \bA_{\sigma(k-1)})$; we denote this situation by writing $\bR \sdp \myprod_{\kk} \bA_{\sigma(j)}$.

\section{Abelian Algebras}
In later sections nonabelian algebras will play the following role: some of the theorems will begin with the assumption that a particular algebra $\bA$ is nonabelian and then proceed to show that if the result to be proved were false, then $\bA$ would have to be abelian.  To prepare the way for such arguments, we review some basic facts about abelian algebras.

\subsection{Definitions}
Let $\bA = \<A, F^{\bA}\>$ be an algebra. A reflexive, symmetric, compatible binary relation $T\subseteq A^2$ is called a \defn{tolerance of $\bA$}.  Given a pair $(\bu, \bv) \in A^m\times A^m$ of $m$-tuples of $A$, we write $\bu \mathrel{\bT} \bv$ just in case $\bu(i) \mathrel{T} \bv(i)$ for all $i\in \mm$.  We state a number of definitions in this section using tolerance relations, but the definitions don't change when the tolerance in question happens to be a congruence relation (i.e., a transitive tolerance).

Suppose $S$ and $T$ are tolerances on $\bA$.  An \defn{$S,T$-matrix} is a $2\times 2$ array of the form
\[
\begin{bmatrix*}[r] t(\ba,\bu) & t(\ba,\bv)\\ t(\bb,\bu)&t(\bb,\bv)\end{bmatrix*},
\]
where $t$, $\ba$, $\bb$, $\bu$, $\bv$ have the following properties:
\begin{enumerate}[label={(\roman*)}]
\item $t\in \sansClo_{\ell + m}(\bA)$,
\item $(\ba, \bb)\in A^\ell\times A^\ell$ and $\ba \mathrel{\bS} \bb$,
\item $(\bu, \bv)\in A^m\times A^m$ and $\bu \mathrel{\bT} \bv$.
\end{enumerate}
Let $\delta$ be a congruence relation of $\bA$.
If the entries of every $S,T$-matrix satisfy
\begin{equation}%
  \label{eq:22}
t(\ba,\bu) \mathrel{\delta} t(\ba,\bv)\quad \iff \quad t(\bb,\bu) \mathrel{\delta} t(\bb,\bv),
\end{equation}
then we say that $S$ \defn{centralizes $T$ modulo} $\delta$ and we write
$\CC{S}{T}{\delta}$.
That is, $\CC{S}{T}{\delta}$  means that
(\ref{eq:22}) holds \emph{for all}
$\ell$, $m$, $t$, $\ba$, $\bb$, $\bu$, $\bv$ satisfying properties (i)--(iii). 

The \defn{commutator} of $S$ and $T$, denoted by $[S, T]$,
is the least congruence $\delta$ such that $\CC{S}{T}{\delta}$
holds.
Note that $\CC{S}{T}{0_A}$ is equivalent to $[S,T] = 0_A$, and this
is sometimes called the \defn{$S, T$-term condition};
when it holds we say  that
$S$ \defn{centralizes} $T$, and write $\C{S}{T}$.
A tolerance $T$ is called \defn{abelian} if
$\C{T}{T}$ (i.e., $[T, T] = 0_A$).
An algebra $\bA$ is called \defn{abelian} if $1_A$ is abelian
(i.e., $\C{1_A}{1_A}$).



\subsection{Facts about centralizers and abelian congruences}
We now collect some useful facts about centralizers of congruence relations that are needed in Section~\ref{sec:tayl-vari-rect}.
The facts listed in the first lemma below are well-known and easy to prove.
(For examples, see~\cite[Prop~3.4]{HM:1988} and~\cite[Thm~2.19]{MR3076179}.)
\begin{lem}%
\label{lem:centralizers}
Let $\bA$ be an algebra and suppose
$\bB$ is a subalgebra of $\bA$.
Let $\alpha$, $\beta$, $\gamma$, $\delta$, $\alpha_i$
$\beta_j$, $\gamma_k$
be congruences of $\bA$, for some
$i \in I$, $j\in J$, $k \in K$. Then the following hold:
\begin{enumerate}
\item\label{centralizing_over_meet}
  $\CC{\alpha}{\beta}{\alpha \meet \beta}$;
\item\label{centralizing_over_meet2}
  if $\CC{\alpha}{ \beta}{ \gamma_k}$ for all $k \in K$, then
  $\CC{\alpha}{ \beta}{ \Meet_{K}\gamma_k}$;
\item\label{centralizing_over_join1}
  if $\CC{\alpha_i}{ \beta}{ \gamma}$ for all $i\in I$, then
  $\CC{\Join_{I}\alpha_i}{ \beta}{\gamma}$;
\item\label{monotone_centralizers1}
  if $\CC{\alpha}{ \beta}{ \gamma}$ and $\alpha' \leq \alpha$, then
  $\CC{\alpha'}{ \beta}{ \gamma}$;
\item\label{monotone_centralizers2}
  if $\CC{\alpha}{ \beta}{ \gamma}$ and $\beta' \leq \beta$, then
  $\CC{\alpha}{ \beta'}{ \gamma}$;
\item\label{centralizing_over_subalg}
  if $\CC{\alpha}{ \beta}{ \gamma}$ in $\bA$,
  then $\CC{\alpha\cap B^2}{ \beta\cap B^2}{\gamma\cap B^2}$ in $\bB$;
\item\label{centralizing_factors}
  if $\gamma \leq \delta$, then $\CC{\alpha}{ \beta}{ \delta}$
  in $\bA$ if and only if $\CC{\alpha/\gamma}{ \beta/\gamma}{ \delta/\gamma}$
  in $\bA/\gamma$.
\end{enumerate}
\end{lem}

\ifthenelse{\boolean{extralong}}{
\begin{rem}
By (\ref{centralizing_over_meet}),
if $\alpha \meet \beta = 0_{A}$,
then $\C{\beta}{\alpha}$ and $\C{\alpha}{\beta}$.
\end{rem}
}{}

\noindent
The next two lemmas are essential in a number proofs below.
The first lemma identifies special conditions
under which certain quotient congruences are abelian.
The second gives fairly general conditions under which
quotients of abelian congruences are abelian.
\begin{lem}%
  \label{lem:common-meets}
  Let $\alpha_0$, $\alpha_1$, $\beta$ be congruences of $\bA$ and suppose
  $\alpha_0 \meet \beta = \delta = \alpha_1 \meet \beta$.
  Then $\CC{\alpha_0 \join \alpha_1}{ \beta}{ \delta}$.  If, in addition,
  $\beta \leq \alpha_0 \join \alpha_1$, then
  $\CC{\beta}{ \beta}{ \delta}$, so $\beta/\delta$ is an
  abelian congruence of $\bA/\delta$.
\end{lem}
Lemma~\ref{lem:common-meets}
is an easy consequence
of items (\ref{centralizing_over_meet}), (\ref{centralizing_over_join1}),
(\ref{monotone_centralizers1}), and (\ref{centralizing_factors}) of
  Lemma~\ref{lem:centralizers}.
\begin{lem}%
  \label{lem:abelian-quotients}
  Let $\var{V}$ be a locally finite variety with a Taylor term and let $\bA\in \var{V}$.
  Then $\CC{\beta}{\beta}{\gamma}$ for all $[\beta, \beta] \leq \gamma$.
\end{lem}
Lemma~\ref{lem:abelian-quotients} can be proved  by combining
the next result, of David Hobby and Ralph McKenzie,
with a result of Keith Kearnes and Emil Kiss.
\begin{lemC}[\protect{\cite[Thm~7.12]{HM:1988}}]%
  \label{lem:HM-thm-7-12}
  A locally finite variety $\var{V}$ has a Taylor term if and only if it has a
  so called \defn{weak difference term}; that is, a term $d(x,y,z)$ satisfying
  the following conditions for all $\bA \in \var{V}$, all $a, b \in A$, and all
  $\beta \in \Con (\bA)$:
  $d^{\bA}(a,a,b) \mathrel{[\beta, \beta]} b \mathrel{[\beta, \beta]} d^{\bA}(b,a,a)$,
  where $\beta = \Cg^{\bA}(a,b)$.
\end{lemC}

\begin{lemC}[\protect{\cite[Lem~6.8]{MR3076179}}]%
  \label{lem:KK-lem-6-8}
  If $\bA$ belongs to a variety with a
  weak difference term and if $\beta$ and $\gamma$ are congruences of $\bA$
  satisfying $[\beta, \beta] \leq \gamma$, then $\CC{\beta}{\beta}{\gamma}$.
\end{lemC}
\begin{rem}%
  \label{rem:abelian-quotients}
  It follows immediately from Lemma~\ref{lem:abelian-quotients} that in a locally
  finite Taylor variety, $\var{V}$, quotients of abelian algebras are abelian, so the
  abelian members of $\var{V}$ form a subvariety.
  But this can also be derived from Lemma~\ref{lem:HM-thm-7-12},
  since $[\beta, \beta] = 0_A$ implies $d^{\bA}$ is a \malcev term operation on
  the blocks of $\beta$, so if $\bA$ is abelian---i.e., if
  $\CC{1_A}{1_A}{0_A}$---then Lemma~\ref{lem:HM-thm-7-12},
  implies that $\bA$ has a \malcev term operation.
  (This will be recorded below in~Theorem~\ref{thm:type2cp}.)
  It then follows that homomorphic images of $\bA$ are
  abelian. (See~\cite[Cor~7.28]{MR2839398} for more details).
\end{rem}

\ifthenelse{\boolean{draft}}{\newpage}{}

\section{Absorption Theory}%
\label{sec:absorption}
In this section we survey some of the theory related to an important concept called \emph{absorption}, which was invented by Libor Barto and Marcin Kozik. (In~\cite{Barto:2017}, these authors provide their own survey, to which we refer the reader for more details.) After introducing the concept, we discuss some of the properties that make absorption so useful.  The main results in this section are not new. The only possibly novel contributions are some straightforward observations in Section~\ref{sec:linking-lemmas} (though these have likely been observed by others). Our intention here is merely to collect and present these results in a way that makes them easy to apply in the sequel.

Let $\bA$ be an algebra and $t\in \sansClo(\bA)$ a $k$-ary term operation. A subalgebra  $\bB \leq \bA$ is said to be \defn{absorbing in $\bA$ with respect to $t$} if for all $1\leq j\leq k$ and for all
\[
    (b_1, \dots, b_{j-1}, a, b_{j+1}, \dots, b_k)\in B^{j-1}\times A \times B^{k-j}
\]
we have $t(b_1, \dots, b_{j-1}, a, b_{j+1}, \dots, b_k)\in B$. In other terms, $t[B^{j-1}\times A \times B^{k-j}] \subseteq B$, for all $1\leq j \leq k$, where $t[D]$ denotes the set $\{ t(x) \mid x\in D \}$.  The notation $\bB \absorbing \bA$ and the phrase ``$\bB$ \defn{absorbs} $\bA$'' are used just in case $\bB$ is an absorbing subalgebra of $\bA$ with respect to some term. To be explicit about the term, one writes $\bB \absorbing_t \bA$, and $t$ is sometimes called an \defn{absorbing term} for $\bB$ in this case. If $B$ is a subuniverse of $\bA$, then $B \absorbing \bA$ means that the subalgebra $\bB = \langle B, \dots \rangle$, with universe $B$, absorbs $\bA$. The notation $\bB \minabsorbing \bA$ means that $\bB$ is a minimal absorbing subalgebra of $\bA$, that is, $B$ is minimal (with respect to set inclusion) among the absorbing subuniverses of $\bA$.  An algebra is said to be \defn{absorption-free} if it has no proper absorbing subalgebras.

\subsection{Absorption theorem}%
\label{sec:absorption-thm}
In later sections, we will make frequent use of a powerful tool of Barto and Kozik called the ``Absorption Theorem'' (Thm~2.3 of~\cite{MR2893395}). This result concerns the special class of ``linked'' subsets of products. A subset $R \subseteq A_0 \times A_1$ is said to be \defn{linked} if it satisfies the following: for $a, a' \in \Proj_0 R$ there exist elements $c_0, c_2, \dots, c_{2n} \in A_0$ and $c_1, c_3, \dots, c_{2n+1} \in A_1$ such that $c_0 = a$, $c_{2n} = a'$, and for all $0\leq i<n$, $(c_{2i},c_{2i+1})\in R$ and $(c_{2i+2},c_{2i+1})\in R$. A subalgebra $\bR \leq \bA_0 \times \bA_1$ linked just in case its universe is.

Here is an easily proved fact that provides some equivalent ways to define ``linked.''
\begin{fact}
Let $\bR  = \<R, \dots\> \sdp \bA_0 \times \bA_1$, let $\etaR_i = \ker(R \onto A_i)$ denote the kernel of the projection of $R$ onto its $i$-th coordinate, and let $R^{-1} = \{(y,x) \in A_1 \times A_0 \mid (x,y) \in R\}$. Then the following are equivalent:
\begin{enumerate}
\item $\bR$ is linked;
\item $\etaR_0\join \etaR_1 = 1_R$;
\item if $a, a' \in \Proj_0 R$, then $(a,a')$ is in the transitive closure of $R\circ R^{-1}$.
\end{enumerate}
\end{fact}

\begin{thm}[Absorption Theorem~\protect{\cite[Thm~2.3]{MR2893395}}]%
\label{thm:absorption}
If $\var{V}$ is an idempotent locally finite variety, then the following are equivalent:
\begin{itemize}
\item $\var{V}$ is a Taylor variety;
\item if $\bA_0, \bA_1 \in \var{V}$ are finite, absorption-free and $\bR \sdp \bA_0 \times \bA_1$ is linked, then $\bR = \bA_0 \times \bA_1$.
\end{itemize}
\end{thm}

Before moving on to the next subsection, we need another definition.  If $f\colon A^\ell\to A$ and $g\colon A^m\to A$ are  $\ell$-ary and $m$-ary operations on $A$, then by $f \star g$ we mean the $\ell m$-ary operation that maps $\ba = (a_{1 1},\dots, a_{1 m}, a_{2 1}, \dots,  a_{2 m},\dots, a_{\ell m}) \in A^{\ell m}$ to
\[
(f\star g)(\ba) = f(g(a_{1 1}, \dots, a_{1 m}), g(a_{2 1}, \dots, a_{2 m}), \dots,  g(a_{\ell 1}, \dots, a_{\ell m})).
\]

\subsection{Properties of absorption}
In this section we prove some easy facts about algebras that hold quite generally.  In particular, we don't assume the presence of Taylor terms or
other \malcev conditions in this section.  However, for the sake of simplicity, we do assume that all algebras are finite and belong to a single idempotent variety, though some of the results of this section hold for infinite algebras as well.   Almost all of the results in this section are well known. The only exceptions are mild generalizations of facts that have already appeared in print. We have restated or reworked some of the known results in a way that makes them easier for us to apply.

\newcommand\sseq{\ensuremath{\subseteq}}
\begin{lem}\label{lem:fact1}
Let $\bA$ be an algebra with term operations $f$ and $g$. If $\bB\absorbing \bA$ with respect to either $f$ or $g$, then $\bB\absorbing \bA$ with respect to $f\star g$.
\end{lem}
The straightforward proof appears in Section~\ref{sec:fact1}. The next corollary is easy to prove by induction on~$p$.

\begin{cor}%
  \label{cor:fact1gen}
  Fix $p <\omega$, a variety $\var{V}$, algebras $\bA$, $\bB \in \var{V}$, and terms $t_0, t_1, \dots, t_{p-1}$ in the language of $\var{V}$. If $\bB\absorbing_{t_i} \bA$ for some $0\leq i < p$, then $\bB \absorbing_s \bA$, where $s = t_0 \star t_1 \star \cdots \star t_{p-1}$.
\end{cor}

\begin{lemC}[\protect{\cite[Prop~2.4]{MR2893395}}]%
\label{lem:bk-prop-2-4}
  Let $\bA$ be an algebra.
  \begin{itemize}
  \item If $\bC \absorbing \bB \absorbing \bA$, then $\bC \absorbing \bA$.
  \item If $\bB \absorbing_f \bA$, $\bC \absorbing_g \bA$, $B \cap C\neq \emptyset$, and $t = f\star g$, then $B \cap C \absorbing_t\bA$.
  \end{itemize}
\end{lemC}

\noindent
The following is an easy corollary of Lemma~\ref{lem:bk-prop-2-4} (for proof see Sec.~\ref{sec:fact2}) and so is the generalization to finite products and minimal absorbing subalgebras given in Lemma~\ref{lem:min-abs-prod} (and proved in Sec.~\ref{sec:proof-cor-min-abs-prod}).

\begin{cor}%
\label{cor:fact2}
Suppose $\bB_i \absorbing_{f_i} \bA_i$ for $i \in\{0,1\}$ and $t = f_0\star f_1$. Then $\bB_0\times \bB_1 \absorbing_t \bA_0\times \bA_1$.
\end{cor}

\begin{lem}\label{lem:min-abs-prod}
  Let $\bB_i \leq \bA_i$ $(0\leq i < n)$ be algebras and define $\bB := \myprod_{\nn} \bB_i$ and $\bA := \myprod_{\nn} \bA_i$.
  If $\bB_i\absorbing_{f_i}\bA_i$ (resp., $\bB_i\minabsorbing_{f_i}\bA_i$) for each $0\leq i < n$, then $\bB \absorbing_t \bA$ (resp., $\bB \minabsorbing_t \bA$) where $t:= f_0\star f_1 \star \cdots \star f_{n-1}$.
\end{lem}
Here is another simple observation with a trivial proof that is nonetheless quite useful for proving certain algebras absorption-free.
\begin{lem}%
\label{lem:restriction}
Let $\bB$ and $\bC$ be subalgebras of a finite idempotent algebra $\bA$. Suppose $\bB \absorbing_t \bA$, and suppose $D = B\cap C \neq \emptyset$. Then the restriction of $t$ to $C$ is an absorbing term for $D$ in $C$, whence $D\absorbing C$.
\end{lem}
Lemma~\ref{lem:min-abs-prod} and Lemma~\ref{lem:restriction} yield the following observation that we will need in Section~\ref{sec:applications}:
\begin{cor}\label{cor:gen-abs1}
Let
$\bA_0$, $\bA_1$, $\dots$, $\bA_{n-1}$
be finite algebras with $\bB_i\absorbing \bA_i$ for each $i$.
Suppose $\bR \leq \myprod_{i}\bA_i$ and $R \cap \myprod_i B_i \neq \emptyset$.
Then $R \cap \myprod_i B_i$ is an absorbing subuniverse of $\bR$.
\end{cor}

For $0<j\leq k$, we say that a $k$-ary term operation $t$ of an algebra $\bA$ \emph{depends on its $j$-th argument} provided there exist $a_0, a_1, \dots, a_{k-1}$ such that $p(x) := t(a_0, \dots, a_{j-2}, x, a_{j}, \dots, a_{k-1})$ is a nonconstant polynomial of $\bA$.

Let \bA\ be an algebra, and $s \in C \subseteq A$.  We call $s$ a \emph{sink for} $C$ provided the following holds: for every term $t \in \sansClo_k(\bA)$ and all $0< j \leq k$, if $t$ depends on its $j$-th argument, then $t(c_0, c_1, \dots, c_{j-2}, s, c_{j}, \dots, c_{k-1})=s$ for all $c_i \in C$.  If $B$ is an absorbing subuniverse that intersects nontrivially with a set containing a sink, then $B$ must also contain the sink.
\begin{lem}%
\label{lem:sink}
  If $B \absorbing A$, if $s$ is a sink for $C\subseteq A$, and
  if $B\cap C \neq \emptyset$, then $s\in B$.
\end{lem}

We conclude this subsection with three more properties of absorption that will be useful below.  The first of these is proved in Section~\ref{sec:proof-lemma-sdp-general} below.

\begin{lem}\label{lem:sdp-general}
Let $\bA_0$, $\bA_1$, $\dots$, $\bA_{n-1}$ be finite idempotent algebras of the same type, and suppose $\bB_i \minabsorbing \bA_i$ for $0\leq i< n$. Let $\bR \sdp \bA_0 \times \bA_1 \times \cdots \times \bA_{n-1}$,  and let $R' = R \cap (B_0 \times B_1 \times \cdots \times B_{n-1})$. If $R'\neq \emptyset$, then $\bR' \sdp \bB_0 \times \bB_1 \times \cdots \times \bB_{n-1}$.
\end{lem}

\begin{lemC}[\protect{\cite[cf.~Prop~2.15]{MR2893395}}]%
\label{lem:linked-absorber}
Let $\bA_0$ and $\bA_1$ be finite idempotent algebras of the same type, let
$\bR \sdp \bA_0 \times \bA_1$ and assume $\bR$ is linked.
If $\bS \absorbing \bR$, then $\bS$ is also linked.
\end{lemC}

\begin{lemC}[\protect{\cite[Lem~4.1]{MR3374664}}]\label{lem:abelian-AF}
  Finite idempotent abelian algebras are absorption-free.
\end{lemC}

\subsection{Linking is easy, sometimes}\label{sec:linking-lemmas}
We will apply the Absorption Theorem frequently below, so we pause here to consider one of the hypotheses of the theorem that might seem less familiar to some readers.  Specifically, we consider when we should expect a subdirect product to be \emph{linked}, as required of $\bR \sdp \bA_0\times \bA_1$ in the statement of the Absorption Theorem. We show that in some special cases the hypothesis comes essentially for free, courtesy of the following elementary property of subdirect products:
  \begin{lem}\label{lem:basic}
    Let $\bA_0$ and $\bA_1$ be algebras.  Suppose
    $\bR \sdp \bA_0 \times \bA_1$ and let $\etaR_i = \ker(\bR \onto \bA_i)$
    denote the kernel of the $i$-th projection of $\bR$ onto $\bA_i$.
    \begin{enumerate}
    \item If $\bA_0$ is simple, then either $\etaR_0 \join \etaR_1 = 1_R$ or $\etaR_0 \geq \etaR_1$.
    \item If $\bA_0$ and $\bA_1$ are both simple, then either $\etaR_0 \join \etaR_1 = 1_R$ or $\etaR_0 = 0_R = \etaR_1$.
    \end{enumerate}
  \end{lem}
  \begin{proof} $\bR/\etaR_0 \cong \bA_0$ and $\bA_0$ is simple, so $\etaR_0$ is a coatom (i.e., maximal) in the congruence lattice $\Con \bR$. Thus, $\etaR_0 \join \etaR_1 < 1_R$ would imply $\etaR_0 =\etaR_0 \join \etaR_1$, proving (1). If $\bA_1$ is also simple, then the same argument, with the roles of $\etaR_0$ and $\etaR_1$ swapped, shows that $\etaR_0 \leq \etaR_1$, so $\etaR_0 = \etaR_1$.  Since $\etaR_0 \meet  \etaR_1 = 0_R$, both projection kernels are $0_R$ in this case.
  \end{proof}
  \begin{cor}\label{cor:rect-two_factors-pre}
    Suppose $\bA_0$, $\bA_1$ are simple and $\bR \sdp \bA_0 \times \bA_1$.
    If $\etaR_0 \neq \etaR_1$ or $\bA_0 \ncong \bA_1$, then $\bR$ is linked.
  \end{cor}

In a locally finite Taylor variety, when one factor of $\bA_0 \times \bA_1$ is abelian and the other is simple and nonabelian, linking comes for free; that is, every subdirect product of such $\bA_0 \times \bA_1$ is linked, as we now prove.
\begin{cor}\label{cor:S-NA-AF_A-pre}
  Suppose $\bA_0$ and $\bA_1$ are algebras in a locally finite Taylor variety.  If $\bA_0$ is abelian and $\bA_1$ is simple and nonabelian, then every subdirect product of $\bA_0 \times \bA_1$ is linked.     Moreover, if $\bR \sdp \bA_0 \times \bA_1$ and if $\bB_1 \minabsorbing \bA_1$, then $\bR$ intersects $\bA_0\times \bB_1$ nontrivially and this intersection forms a linked subdirect product of $\bA_0 \times \bB_1$.
\end{cor}
\begin{proof}
Suppose  $\bR \sdp \bA_0 \times \bA_1$ is not linked. Since $\bA_1$ is simple, $\etaR_0 \leq \etaR_1$ by Lemma~\ref{lem:basic}. Thus, $\etaR_0 = \etaR_0 \meet \etaR_1 = 0_R$, so $\bR \cong \bR/\etaR_0 \cong \bA_0$ is abelian, while $\bR/\etaR_1 \cong \bA_1$ is not. This is a contradiction since $\bR/\theta$ is abelian for all $\theta \in \Con \bR$ (by Remark~\ref{rem:abelian-quotients}). For the second part, since $\bR$ is subdirect, it follows that $R \cap (A_0\times B_1)$ is nonempty and subdirect by Lemmas~\ref{lem:linked-absorber} and~\ref{lem:abelian-AF}. By Corollary~\ref{cor:gen-abs1}, $\bR \cap (\bA_0\times \bB_1) \absorbing \bR$, so Lemma~\ref{lem:linked-absorber} implies that the intersection is also linked.
\end{proof}
We can extend the previous result to multiple abelian factors by collecting such factors into one. Let $\nn := \{0,1,\dots, n-1\}$ and $k':=\nn-\{k\}$ (in particular, $0' :=  \{1,2 ,\dots, n-1\}$), and let $\bR_{0'} := \Proj_{0'} \bR$.
\begin{cor}\label{cor:Link-1}
    Let $\bA_0$, $\bA_1$, $\dots$, $\bA_{n-1}$ be algebras in a locally finite Taylor variety. Suppose $\bA_0$ is simple, nonabelian, and $\bA_1, \bA_2, \dots, \bA_{n-1}$ are abelian. If $\bR \sdp \myprod_{\nn} \bA_i$, then $\bR \sdp \bA_0 \times \bR_{0'}$ is linked.
\end{cor}
\begin{proof}
    Suppose $\bR \sdp \bA_0 \times \bR_{0'}$ is not linked: $\etaR_0 \join \etaR_{0'} < 1_R$. Since $\bA_0$ is simple, $\etaR_0$ is a coatom, so $\etaR_0 \geq \etaR_{0'}$. Therefore, $ \etaR_{0'} = \etaR_0 \meet \etaR_{0'} = 0_R$, so $\bR \cong \bR/\etaR_{0'} \cong \bR_{0'} \leq \myprod_{0'}\bA_i$. This proves that $\bR$ is abelian, yet $\bR/\etaR_0 \cong \bA_0$ is not, which contradicts Remark~\ref{rem:abelian-quotients}.
\end{proof}

The Absorption Theorem, Lemmas~\ref{lem:sdp-general}--\ref{lem:abelian-AF}, and Corollaries~\ref{cor:gen-abs1},~\ref{cor:rect-two_factors-pre},~\ref{cor:S-NA-AF_A-pre} yield

\begin{lem}\label{lem:rect-two_factors}
  Let $\bA_0$ and $\bA_1$ be finite  algebras in a Taylor variety with minimal absorbing subalgebras $\bB_i\minabsorbing \bA_i$ ($i =0,1$), and suppose $\bR \sdp \bA_0 \times \bA_1$, where $\etaR_0 \neq \etaR_1$.
    \begin{enumerate}[label={(\roman*)}]
    \item If $\bA_0$ and $\bA_1$ are simple and $R\cap (B_0 \times B_1) \neq \emptyset$, then $\bB_0 \times \bB_1\leq \bR$.
    \item  If $\bA_0$ is simple nonabelian and $\bA_1$ is abelian, then $\bB_0 \times \bA_1 \leq \bR$.
    \end{enumerate}
\end{lem}
\begin{proof}\
  \begin{enumerate}[label={(\roman*)}]
   \item First note that $\bR$ is linked, by Corollary~\ref{cor:rect-two_factors-pre}. Let $\bR':= \bR \cap (\bB_0 \times \bB_1)$. Then by Lemma~\ref{lem:sdp-general} $\bR' \sdp \bB_0 \times \bB_1$, and by Corollary~\ref{cor:gen-abs1} $\bR'\absorbing \bR$, so $\bR'$ is also linked, by Lemma~\ref{lem:linked-absorber}.  The hypotheses of the Absorption Theorem---with $\bR'$ in place of $\bR$ and $\bB_i$ in place of $\bA_i$---are now satisfied.  Therefore, $\bR' = \bB_0 \times \bB_1$.
   \item This follows from Corollary~\ref{cor:S-NA-AF_A-pre},
     Lemma~\ref{lem:abelian-AF}, and the Absorption Theorem.
\qedhere
  \end{enumerate}
\end{proof}

\noindent
Again, collecting multiple abelian factors into one, we obtain
  \begin{cor}%
  \label{cor:S-NA-AF_A-multi}
  Let  $\bA_0$, $\bA_1$, $\dots$, $\bA_{n-1}$  be finite Taylor algebras, where $\bA_0$ is simple nonabelian, $\bB_0 \minabsorbing \bA_0$,
  and the remaining $\bA_i$ are abelian, 
  and suppose $\bR \sdp \myprod_{\nn}\bA_i$. Then $\bB_0 \times \bR_{0'} \leq \bR$.
  \end{cor}
  \begin{proof}
    Obviously, $\bR \sdp \bA_0 \times \bR_{0'}$ and $R \cap (B_0 \times R_{0'}) \neq \emptyset$.  Also, $\bR_{0'}:=\Proj_{0'} \bR$ is abelian, so we can  apply Lemma~\ref{lem:rect-two_factors} (ii), with $\bR_{0'}$ in place of $\bA_1$.
  \end{proof}

\section{The Rectangularity Theorem}\label{sec:tayl-vari-rect}
\subsection{Some history}
In late May of 2015 we attended a workshop on ``Open Problems in Universal Algebra,'' at which Libor Barto announced a new theorem that he
and Marcin Kozik proved called the ``Rectangularity Theorem.''  At this meeting Barto gave a detailed overview of the proof~\cite{Barto-shanks}.
The authors of the present paper then made a concerted effort over many months to fill in the details and produce a complete proof.  Unfortunately, each attempt uncovered a gap that we were unable to fill. Nonetheless, we found a slightly different route to the theorem. Our argument is similar to the one presented by Barto in most of its key aspects. In particular, we make heavy use of absorption and the Absorption Theorem plays a key role.  However, we were unable to complete the proof without a new ``Linking Lemma'' (Lem.~\ref{lem:Link-2}) that we proved using a general result of Kearnes and Kiss. Thus, our argument is similar in spirit to the original, but is (at the time of this writing, to our knowledge) the first complete proof of Barto and Kozik's Rectangularity Theorem to appear in print. We hope our presentation sheds addition light on the result.

In the next subsection, we prove the Rectangularity Theorem and some corollaries that we will apply in Section~\ref{sec:applications}, where we show how to use these results to determine the complexity of certain constraint satisfaction problems.

\subsection{Preliminaries}
We are mainly concerned with finite algebras belonging to a single Taylor variety $\var{V}$ (though some results below hold more generally).
As mentioned, our goal in this subsection is a proof of the Rectangularity Theorem of Barto and Kozik.  The theorem can be described informally as follows: Start with finite algebras $\bA_0$, $\bA_1$, $\dots$, $\bA_{n-1}$ in a Taylor variety with minimal absorbing subalgebras $\bB_i \minabsorbing \bA_i$. Suppose at most one of these algebras is abelian and the nonabelian algebras are simple.  Let $\bR$ be a subdirect product of $\myprod_{\nn} \bA_i$. Then (under a few more assumptions) if $R$ intersects nontrivially with the product $\myprod_{\nn}B_i$, then $R$ contains the entire product $\myprod_{\nn}B_i$.

Our proof of this result will depend on a number of useful lemmas that we now prove.  

\noindent {\bf Notation}.
\begin{itemize}
\item $\nn :=\{0, 1, 2, \dots, n-1\}$.
\item For $\sigma \subseteq \nn $, $A = A_0 \times A_1 \times \cdots \times A_{n-1}$, and $\ba = (a_0, a_1\dots, a_{n-1})$, let
  \[ \eta_\sigma := \ker(A \onto \Pi_\sigma A_i) = \{(\ba, \ba') \in A^2 \mid a_i = a_i' \text{ for all } i \in \sigma\} \]
  be the kernel of the projection of $A$ onto coordinates $\sigma$.
\item For $\bR \sdp \bA_0 \times \bA_1 \times \cdots \times \bA_{n-1}$, let
  \[
  \etaR_\sigma := \ker(R \onto \Pi_\sigma A_i) = \{(\br, \br') \in R^2 \mid r_i = r_i' \text{ for all } i \in \sigma\}
  \]
  be the kernel of the projection of $R$ onto the coordinates $\sigma$. Thus, $\etaR_\sigma = \eta_\sigma \cap R^2$.
\item For $\sigma\subseteq \nn$, let $\sigma' := \nn -\sigma$, and take $\bR \sdp \myprod_\sigma \bA_i \times \myprod_{\sigma'}\bA_i$ to meanthe
  \ifthenelse{\boolean{footnotes}}{%
    following:\footnote{Note that the expression
    $\bR \sdp \myprod_\sigma \bA_i \times \myprod_{\sigma'}\bA_i$ does not mean
    $\bR$ is a subalgebra of $\myprod_\sigma \bA_i \times \myprod_{\sigma'}\bA_i$.
    Generally speaking, such an interpretation would require permuting the coordinates
    of elements of $\bR$, which is feasible but unnecessary if we use the definition
    given in the text above.}
  }{following:}
  \begin{enumerate}
  \item $\bR$ is a subalgebra of $\myprod_{\nn} \bA_i$;
  \item $\Proj_\sigma\bR = \myprod_\sigma \bA_i$;
  \item $\Proj_{\sigma'}\bR = \myprod_{\sigma'} \bA_i$;
  \end{enumerate}
  we say that $\bR$ is a \defn{subdirect product of}
  $\myprod_\sigma \bA_i$ \emph{and} $\myprod_{\sigma'}\bA_i$ in this case.
\item
  The subdirect product $\bR \sdp \myprod_\sigma \bA_i \times
  \myprod_{\sigma'}\bA_i$
  is said to be \defn{linked} if $\etaR_\sigma \join \etaR_{\sigma'} = 1_R$.
\item We sometimes use $\bR_\sigma$ as shorthand for $\Proj_\sigma\bR$, the projection
  of $\bR$ onto coordinates $\sigma$.

\end{itemize}

\subsection{Rectangularity theorem}%
\label{sec:rect-theor}
With a few more preliminary results, we will be ready to prove the general Rectangularity Theorem. The first of these comes from combining Lemma~\ref{lem:min-abs-prod}, transitivity of absorption, Corollary~\ref{cor:gen-abs1}, and Lemma~\ref{lem:linked-absorber}.
\begin{lem}\label{lem:general-linked}
 Let $\bA_0$, $\bA_1$, $\dots$, $\bA_{n-1}$ be finite algebras in a Taylor variety, let $\bB_i\minabsorbing \bA_i$ for each $i\in \nn $, and let $\nn  = \sigma \cup {\sigma'}$ be a partition of $\{0, 1, \dots, n-1\}$ into two nonempty disjoint subsets. Assume $\bR$ is a \emph{linked} subdirect product of $\myprod_{\sigma} \bA_i$ and $\myprod_{\sigma'}\bA_i$, and suppose $R' := R \cap \myprod_i B_i \neq \emptyset$. Then $\myprod_i B_i\subseteq R$, so $\bR' = \myprod_i \bB_i$.
\end{lem}
\begin{proof}
  By Lemma~\ref{lem:sdp-general}, $\bR' \sdp \myprod_\sigma \bB_i \times \myprod_{\sigma'} \bB_i$, and by Corollary~\ref{cor:gen-abs1}, $\bR' \absorbing \bR$.  Therefore, Lemma~\ref{lem:linked-absorber} implies $\bR'$ is linked. By Lemma~\ref{lem:min-abs-prod} and transitivity of absorption, it follows that $\myprod_{\sigma} \bB_i$ and $\myprod_{\sigma'} \bB_i$ are both absorption-free, so the Absorption Theorem implies that $\bR' = \myprod_{\sigma}\bB_i \times \myprod_{\sigma'}\bB_i$.
\end{proof}

Next, we recall a theorem of Keith Kearnes and Emil Kiss, proved in~\cite{MR3076179}, about algebras in a Taylor variety.

\begin{thmC}[\protect{\cite[Thm~3.27]{MR3076179}}]\label{thm:kearnes-kiss-3.27}
  Suppose $\alpha$ and $\beta$ are congruences of a Taylor algebra. Then $\CC{\alpha}{ \alpha}{ \alpha \meet \beta}$ if and only if $\CC{\alpha \join \beta}{ \alpha \join \beta}{ \beta}$.
\end{thmC}
The context in which we apply Theorem~\ref{thm:kearnes-kiss-3.27} is that of the following corollary. 
\begin{cor}\label{cor:common-meets}
  Let $\alpha_0$, $\alpha_1$, $\beta$, $\delta$ be congruences of a Taylor algebra $\bA$, and suppose that $\alpha_0 \meet \beta = \delta = \alpha_1 \meet \beta$ and $\alpha_0 \join \alpha_1 = \alpha_0 \join \beta =\alpha_1 \join\beta =1_A$.  Then, $\CC{1_A}{ 1_A}{ \alpha_0}$ and $\CC{1_A}{ 1_A}{ \alpha_1}$, so $\bA/\alpha_0$ and $\bA/\alpha_1$ are abelian algebras.
\end{cor}
\begin{proof}
  The hypotheses of Lemma~\ref{lem:common-meets} hold, so $\CC{\beta}{\beta}{\delta}$ and $\beta/\delta$ is an abelian congruence of $\bA/\delta$.  Now, since $\delta = \alpha_i \meet \beta$, we have $\CC{\beta}{ \beta}{ \alpha_i \meet \beta}$, so Theorem~\ref{thm:kearnes-kiss-3.27} implies $\CC{\alpha_i\join \beta}{ \alpha_i\join \beta}{ \alpha_i}$. This yields $\CC{1_A}{ 1_A}{ \alpha_i}$, since $\alpha_i \join \beta =1_A$.  By Lemma~\ref{lem:centralizers}~(\ref{centralizing_factors}) then, $\CC{1_A/\alpha_i}{ 1_A/\alpha_i}{ \alpha_i/\alpha_i}$. Equivalently, $\CC{1_{A/\alpha_i}}{ 1_{A/\alpha_i}}{ 0_{A/\alpha_i}}$. That is, $\bA/\alpha_i$ is abelian.
\end{proof}

Before proceeding, let us remark that the statements of many results to follow will include the hypothesis that subdirect products do not involve redundant factors (since such factors can be combined into one without changing the solution set of the corresponding \csp instance).  This is expressed by assuming that kernels of projections onto distinct coordinates are distinct, that is, $\etaR_i \neq \etaR_j$ for all $i\neq j$.

\begin{lem}[Linking Lemma]%
  \label{lem:Link-2}
  Let $n\geq 2$ and $k' := \nn-\{k\}$. Let $\bA_0$, $\bA_1$, $\dots$, $\bA_{n-1}$ be finite algebras in a Taylor variety, and let $\bB_i \minabsorbing \bA_i$ for all $0\leq i < n$. Suppose
  \begin{itemize}
  \item at most one $\bA_i$ is abelian,
  \item all nonabelian factors are simple,
  \item $\bR \sdp \bA_0 \times \bA_1 \times \cdots \times \bA_{n-1}$, and
  \item $\etaR_i \neq \etaR_j$ for all $i\neq j$.
  \end{itemize}
  Then there exists $k$ such that $\bR\sdp \bA_k \times \bR_{k'}$ is linked.
\end{lem}
\ifthenelse{\boolean{extralong}}{%
  \begin{rem}
    To be pedantic, the expression  $\bR\sdp \bA_k \times \bR_{k'}$ that appears in the conclusion of the lemma only makes sense if we swap the $0$-th and $k$-th coordinates of all elements of $\bR$ so that the ``$0$-th projection'' is projection onto $\bA_k$. This is merely a defect of the notation.
  \end{rem}
}{}

\begin{proof}
 By way of contradiction, suppose $\etaR_k \join \etaR_{k'} < 1_R$ for all $0\leq k< n$.

 \medskip\noindent
 \textit{Case 1}. Every $\bA_i$ is nonabelian.

 \smallskip\noindent
    If $n=2$, the result holds by Corollary~\ref{cor:rect-two_factors-pre}. Assume $n>2$. Since every factor is simple, each $\etaR_k$ is a coatom, so $\etaR_k \join \etaR_{k'} < 1_R$ implies $\etaR_k \geq \etaR_{k'}$.  Therefore,
    \begin{equation}\label{eq:1111}
      \etaR_{k'}:= \Meet_{i \neq k}\etaR_i = \bigl(\Meet_{i \neq k}\etaR_i\bigr) \meet \etaR_k = \etaR_{\nn} = 0_R.
    \end{equation}
    Note that (\ref{eq:1111}) holds for all $k\in \nn$. Now, let $\tau \subseteq \nn$ be a subset that is maximal among those satisfying $\etaR_\tau > 0_R$. By what we just showed, $\tau$ omits at least two indices, say, $j$ and $\ell$. By maximality, $\etaR_\tau \meet \etaR_j = 0_R = \etaR_\tau \meet \etaR_\ell$,, and $\etaR_\tau \nleq \etaR_j$, and $\etaR_\tau \nleq \etaR_\ell$. Since all factors are simple, $\etaR_j$ and $\etaR_\ell$ are (distinct) coatoms, so $\etaR_\tau \join \etaR_j  = \etaR_\tau \join \etaR_\ell = \etaR_j\join \etaR_\ell = 1_R$. Therefore, by Corollary~\ref{cor:common-meets}, both $\bA_j \cong \bR/\etaR_j$ and $\bA_\ell\cong \bR/\etaR_\ell$ are abelian, which contradicts the assumption that there are no abelian factors.

\medskip\noindent
 \textit{Case 2}. Exactly one $\bA_i$ is abelian.

\smallskip\noindent
    Assume without loss of generality that  $\bA_0$ is the abelian factor.  Then, for $k > 0$, $\bA_k$ is simple and nonabelian, so $\etaR_k$ is a coatom of $\Con\bR$. Therefore,  $\etaR_k \join \etaR_{k'} < 1_R$ implies $\etaR_{k'}\leq \etaR_k$, so
    \begin{equation}\label{eq:110}
      \etaR_{k'}= \etaR_{k'} \meet \etaR_k = \Meet_{i\in \nn}\etaR_i = 0_R.
    \end{equation}
    (Note that this holds for every $0<k<n$.)

    Now, let $\theta = \etaR_{0'} := \Meet_{i \neq 0} \etaR_i$.  Let $\tau$ be a maximal subset of $0'$ such that $\etaR_\tau > \theta$.  Obviously $\tau$ is a proper subset of $0'$. Moreover, by~(\ref{eq:110}) there exists $j\in 0'$ such that $\etaR_\tau \nleq \etaR_j$. Since $\etaR_j$ is a coatom, $\etaR_\tau \join \etaR_j = 1_R$, and by maximality of $\tau$, we have $\etaR_\tau \meet \etaR_j = \theta$.

    Theorem~\ref{thm:kearnes-kiss-3.27} yields $\CC{\etaR_\tau \join \etaR_j}{ \etaR_\tau \join \etaR_j}{ \etaR_j}$ if and only if $\CC{\etaR_\tau}{ \etaR_\tau}{  \etaR_\tau \meet \etaR_j}$; equivalently,
   \[\CC{1_R}{ 1_R}{ \etaR_j} \quad \iff \quad \CC{\etaR_\tau}{ \etaR_\tau}{  \theta}.\]
    However, $\CC{1_R}{ 1_R}{ \etaR_j}$ implies $\bA_j$ is abelian, but $\bA_0$ is the only abelian factor, so
    \begin{equation}\label{imp:12}
      \CC{\etaR_\tau}{ \etaR_\tau}{  \theta} \; \text{{\it does not hold}}.
    \end{equation}
    Since $\bA_0$ is abelian,  $\CC{ 1_R}{ 1_R}{ \etaR_0}$; {\it a fortiori},  $\CC{\etaR_\tau \join \etaR_0}{ \etaR_\tau \join \etaR_0}{ \etaR_0}$, which holds if and only if  $\CC{\etaR_\tau}{ \etaR_\tau}{  \etaR_\tau \meet \etaR_0}$ by Theorem~\ref{thm:kearnes-kiss-3.27}. Now, if $\etaR_0 \meet \etaR_\tau \leq \theta$, then Lemma~\ref{lem:abelian-quotients} implies $\CC{\etaR_\tau}{ \etaR_\tau}{ \theta}$, which is false by (\ref{imp:12}). Thus, $\etaR_0 \meet \etaR_\tau \nleq \theta$. It follows that $\etaR_0 \meet \etaR_\tau \neq 0_R$. By~(\ref{eq:110}), then, there are at least two distinct indices, say, $j$ and $\ell$, in $\nn$ that do not belong to $\tau$.   By maximality of $\tau$, $\etaR_\tau \meet \etaR_j = \theta = \etaR_\tau \meet \etaR_\ell$.  By Corollary~\ref{cor:common-meets} (with $\alpha_0 = \etaR_j$, $\alpha_1 = \etaR_\ell$, $\beta = \etaR_\tau$, $\delta = \theta$), it follows that $\bR/\etaR_j \cong \bA_j$ and $\bR/\etaR_\ell \cong \bA_\ell$ are both abelian---a contradiction.
\end{proof}
We have assembled all the tools we will use to accomplish the main goal of this subsection, which is to prove the following:
\begin{thm}[Rectangularity Theorem, cf.~\protect{\cite[Theorem 38]{Barto:2017}}]\label{thm:rectangularity}
  Let $\bA_0$, $\bA_1$, $\dots$, $\bA_{n-1}$ be finite algebras in a Taylor variety, and let $\bB_i \minabsorbing \bA_i$ for all $0 \leq i < n$. Suppose
\begin{itemize}
 \item at most one $\bA_i$ is abelian,
 \item all nonabelian factors are simple,
 \item $\bR \sdp \bA_0 \times \bA_1 \times \cdots \times \bA_{n-1}$,
 \item $\etaR_i \neq \etaR_j$ for all $i\neq j$, 
 \item $R':= R \cap (B_0 \times B_1 \times \cdots \times B_{n-1}) \neq \emptyset$.
\end{itemize}
Then, $\bR' = \bB_0 \times \bB_1 \times \cdots \times \bB_{n-1}$. 
\end{thm}

\begin{proof}
  It suffices to prove $B_0 \times B_1 \times \cdots \times B_{n-1} \subseteq R$,
  which we do by induction on the number of factors in the product $\bA_0 \times \bA_1 \times \cdots \times \bA_{n-1}$.

  For $n=2$ the result holds by Lemma~\ref{lem:rect-two_factors}. Fix $n>2$ and assume for all $2 \leq k < n$ that the result holds for subdirect products of $k$ factors. We prove the result for subdirect products of $n$ factors.

  Let $\sigma$ be a nonempty proper subset of $\nn$ (so $1\leq |\sigma| < n$). As above $\sigma'$ is the complement of $\sigma$ in $\nn$, and $\bR_\sigma$ is the projection of $\bR$ onto coordinates in $\sigma$.  Each of the projections $\bR_\sigma$ and  $\bR_{\sigma'}$ satisfies the assumptions of the theorem. (Indeed, $\bR_{\sigma} \sdp \myprod_{\sigma}\bA_i$ and for all $i\neq j$ in~$\sigma$, $\ker(R_{\sigma} \onto A_i) \neq \ker(R_{\sigma}\onto A_j)$; similarly for $\bR_{\sigma'}$.) Therefore, the induction hypothesis implies $\myprod_{\sigma}\bB_i \leq \bR_{\sigma}$ and $\myprod_{\sigma'}\bB_i \leq \bR_{\sigma'}$. By Lemma~\ref{lem:min-abs-prod}, $\myprod_{\sigma}\bB_i  \minabsorbing \myprod_{\sigma} \bA_i$, and
  from $\myprod_{\sigma}\bB_i \leq \bR_{\sigma}\leq \myprod_{\sigma} \bA_i$ follows $\myprod_{\sigma}\bB_i \absorbing \bR_{\sigma}$. In fact, $\myprod_{\sigma}\bB_i \minabsorbing \bR_{\sigma}$ as well, by minimality of $\myprod_{\sigma}\bB_i \minabsorbing \myprod_{\sigma}\bA_{i}$, and transitivity of absorption. To summarize, for every $\emptyset \subsetneqq \sigma\subsetneqq \nn$,
  \begin{equation}\label{eq:20}
  \bR \sdp \bR_{\sigma} \times \bR_{\sigma'}, \quad \myprod_\sigma \bB_i \minabsorbing \bR_{\sigma}, \quad \myprod_{\sigma'} \bB_i \minabsorbing \bR_{\sigma'}.
  \end{equation}
  By the Linking Lemma, $\bR \sdp \bR_{k} \times \bR_{k'}$ is linked for some $k$. Thus, by Lemma~\ref{lem:general-linked}, the proof is complete.
\end{proof}

Theorem~\ref{thm:rectangularity} can be adapted to handle cases in which there are multiple abelian factors.
\begin{cor}\label{cor:tayl-vari-abel-fact}
  Let $\bA_0$, $\bA_1$, $\dots$, $\bA_{n-1}$ be finite algebras in a Taylor variety, and let $\bB_i \minabsorbing \bA_i$ for all $0 \leq i < n$. Suppose $\alpha \subseteq \nn$ and
\begin{itemize}
  \item $\bA_i$ is abelian for each $i \in \alpha$,
  \item $\bA_i$ is nonabelian and simple for each $i \in \alpha'$,
  \item $\bR \sdp \bA_0 \times \bA_1 \times \cdots \times \bA_{n-1}$,
  \item $\etaR_i \neq \etaR_j$, for all $i,j\in \alpha'$ with $i\neq j$, and
  \item $R':= R \cap (B_0 \times B_1 \times \cdots \times B_{n-1}) \neq \emptyset$.
\end{itemize}
Then $\bR'= \bR_\alpha  \times \myprod_{\alpha'}\bB_i$.
\end{cor}
\begin{proof}
  If $\alpha= \emptyset$ then the product has no abelian factors and the result follows immediately from the Rectangularity Theorem.  If $\alpha\neq \emptyset$, let $\alpha' = \{i_0, i_1, \dots, i_{m-1}\}$.  Clearly,
  $\bR \sdp \bR_\alpha \times \bA_{i_0} \times \bA_{i_1} \times \cdots \times \bA_{i_{m-1}}$, and this product has a single abelian factor,
  $\bR_\alpha \leq \myprod_{\alpha} \bA_i$. If $\theta$ denotes the kernel of the projection of $R$ onto $R_\alpha$, then for all $j<m$ we have
  $\theta \neq \etaR_{i_j}$, since $\bR_\alpha$ is abelian while $\bA_{i_j}$ is not. We can now apply the Rectangularity Theorem to this subdirect
  representation since the projection congruences are pairwise distinct.
\end{proof}

We conclude this subsection with two observations that facilitate application of the results above to \csp problems.
\begin{cor}\label{cor:RT-cor}
  Let $\bA_0$, $\bA_1$, $\dots$, $\bA_{n-1}$ be finite algebras in a Taylor variety with $\bB_i \minabsorbing \bA_i$ for all $0\leq i < n$, and let $\alpha \subseteq \nn$. Suppose $\bR$ and $\bS$ are subdirect products of $\myprod_{\nn} \bA_i$ and
\begin{enumerate}
  \item $\bA_i$ is abelian for each $i \in \alpha$,
  \item $\bA_i$ is nonabelian and simple for each $i \notin \alpha$,
  \item $\etaR_i^R \neq \etaR_j^R$ and $\etaR_i^S \neq \etaR_j^S$, for all $i,j \in \alpha'$ with $i\neq j$,
  \item $R$ and $S$ intersect $\myprod_{\nn} B_i$ nontrivially, and
  \item there exists $\bx \in R_\alpha \cap S_\alpha$.
\end{enumerate}
Then $R \cap S \neq \emptyset$.
\end{cor}
\begin{proof}
  By Corollary~\ref{cor:tayl-vari-abel-fact}, $\bR'= \bR_\alpha   \times \myprod_{\alpha'}\bB_i$ and $\bS'= \bS_\alpha   \times \myprod_{\alpha'}\bB_i$.  Therefore, $\bx \in R_\alpha \cap S_\alpha$ implies $\{\bx\} \times \myprod_{\alpha'}\bB_i \subseteq R \cap S$.
\end{proof}
The generalization to more than two subdirect products is trivial. Nonetheless, for future reference, we record the \textit{general formulation of the Rectangularity Theorem} in
\begin{cor}%
  \label{cor:RT-cor-gen}
  Let $\bA_0$, $\bA_1$, $\dots$, $\bA_{n-1}$ be finite algebras in a Taylor variety with $\bB_i \minabsorbing \bA_i$ for all $0\leq i < n$, and let $\alpha \subseteq \nn$. Suppose $\bR_\ell \sdp \myprod_{\nn} \bA_i$ ($0\leq \ell < m$) and
  \begin{enumerate}
  \item $\bA_i$ is abelian for each $i \in \alpha$,
  \item $\bA_i$ is nonabelian and simple for each $i \notin \alpha$,
  \item $\forall \ell \in \mm$, $\forall i, j \in \alpha'$, $i\neq j \implies
    \etaR^\ell_i \neq \etaR^\ell_j$ (where $\etaR^\ell_i := \ker(\bR_\ell \onto \bA_i)$),
  \item\label{item:RT-cor-gen-4} each $R_\ell$ intersects $\myprod B_i$ nontrivially,
  \item\label{item:RT-cor-gen-5} there exists $\bx \in \bigcap \Proj_\alpha R_\ell$.
  \end{enumerate}
  Then $\bigcap R_\ell \neq \emptyset$. In fact $\{\bx\} \times \prod\limits_{i\in \alpha'} \bB_i
  \subseteq \bigcap\limits_{\ell <m} \bR_\ell$.
\end{cor}

\ifthenelse{\boolean{draft}}{\newpage}{}

\section{CSP Applications}%
\label{sec:applications}
In this section we give a precise definition of what we mean by a ``constraint satisfaction problem,'' and what it means for such a problem to be ``tractable.'' We then give examples demonstrating how to use the algebraic tools developed above to prove tractability.

\subsection{Definition of a constraint satisfaction problem}%
\label{sec:defin-constr-satisf}
We give a definition of ``constraint satisfaction problem'' that is convenient for our purposes. This is not the most general definition possible, but let us postpone consideration of the scope of our setup.

Let $\bA = \<A, \sF\>$ be a finite idempotent algebra, and let $\Sub(\bA)$ and $\sansS(\bA)$ denote the set of subuniverses and subalgebras
of $\bA$, respectively.
\begin{defi}\label{def:csp}
  Let $\mathfrak{A}$ be a collection of algebras of the same similarity type. We define $\CSP(\mathfrak{A})$ to be the following decision problem:
\begin{itemize}
\item  An  \defn{$n$-variable instance} of $\CSP(\mathfrak{A})$ is a triple $\<\sV, \sA, \sC\>$, where
  \begin{itemize}
  \item $\sV$ is a set of $n$ \emph{variables}; often we take $\sV$ to be $\nn= \{0, 1, \dots, n-1\}$;
  \item $\sA$ is a list $(\bA_0, \bA_1, \dots, \bA_{n-1})\in \mfA^n$ of algebras from $\mathfrak{A}$, one for each variable;
  \item $\sC$ is a list $((\sigma_0, R_0), \dots, (\sigma_{J-1}, R_{J-1}))$ of \emph{constraints}; each $\sigma_j$ is a \emph{scope function} with arity $\ar(\sigma_j) = m_j$; each $R_j$ is a \emph{constraint relation}, which we take to be the universe of a subdirect product of the algebras in $\sA$ whose indices belong to $\im \sigma_j$; that is, \[\bR_j \sdp \prod_{0\leq i < m_j}\bA_{\sigma_j(i)}.\]
  \end{itemize}
\item A \defn{solution} to the instance $\<\sV, \sA, \sC\>$ is an assignment $f \colon \sV \to \bigcup_{\nn}A_i$ of values to variables that satisfies all constraint relations.  More precisely, $f\in \myprod_{\nn}A_i$ and $f \circ \sigma_j \in R_j$ holds for all $0\leq j < J$. We denote the set of solutions to $\<\sV, \sA, \sC\>$ by $\Sol(\sC, \nn)$.

\end{itemize}
\end{defi}

\begin{rem}\
  \begin{enumerate}[label={(\roman*)}]
  \item The scope function $\sigma_i \colon \mm_i \to \sV$ picks the $m_i$ variables involved in the constraint relation $R_i$. Thus,
    $((\sigma_0, R_0), \dots, (\sigma_{J-1}, R_{J-1})) \in \myprod_{i< J}\sV^{m_i}$.
  \item
    Frequently we require that arities of scope functions are bounded above by, say, $m \geq \ar(\sigma_j)$ $(j<J)$, which yields the
    \emph{local constraint satisfaction problem} $\CSP(\mathfrak A, m)$, with instances so restricted.
  \item
    If $(\sigma, R) \in \sC$ is a constraint of an $n$-variable instance, and if $\Sol((\sigma, R), \nn)$ denotes the set of tuples in $\prod_{\nn}A_i$
    that satisfy $(\sigma, R)$, then
    \[ \Sol((\sigma, R), \nn) = R^{\overleftarrow{\sigma}} := \bigl\{\bx \in \prod_{i\in \nn}A_i \mid \bx \circ \sigma \in R\bigr\}.\]
    Thus, the instance $\<\sV, \sA, \sC\>$ has solution set $\Sol(\sC, \nn) = \bigcap_{(\sigma, R) \in \sC} R^{\overleftarrow{\sigma}}$.

  \item If $\mathfrak A$ contains a single algebra, we write $\CSP(\bA)$ instead of $\CSP(\{\bA\})$.  It is important to note that, in our definition of an instance of $\CSP(\mathfrak A)$, a constraint relation is a subdirect product of algebras in $\mathfrak A$. Thus, constraint relations of an insxtance of $\CSP(\bA)$ are subdirect powers of $\bA$. In the literature it is conventional to allow instances of $\CSP(\bA)$ with constraint relations that are (not necessarily subdirect) subpowers of $\bA$.  However, there is a simple ``arc-consistency'' algorithm (see, e.g.,~\cite[Sec. 4.2.1]{Barto:2017}) that takes such an instance and either shows it has no solution or transforms it into an instance that has subdirect constraint relations and the same solution set as the original instance.
\end{enumerate}
\end{rem}

\subsection{Instance size and tractability}%
\label{sec:inst-size-tract}
The computational complexity of an algorithm for solving instances of a \csp is measured as a function of input size. Thus, in order to say what it means for $\CSP(\mathfrak A)$ to be ``computationally tractable,'' we need a definition of input size. This amounts to determining the number of bits required to completely specify an instance of the problem. In practice, an upper bound on the size is usually sufficient.

Using the notation in Definition~\ref{def:csp} as a guide, we bound the size of an instance $\sI=\<\sV, \sA, \sC\>$ of $\CSP(\mathfrak A)$.
Let $q=\max(\card{A_0},\, \card{A_1},\dots,\card{A_{n-1}})$, let $r$ be the maximum rank of an operation symbol in the similarity type, and $p$ the number of operation symbols. Then each member of the list $\sA$ requires at most $pq^r\log q$ bits to specify. Thus
\begin{equation*}
\size(\sA) \leq npq^r\log q.
\end{equation*}
Similarly, each constraint scope $\sigma_j\colon \mm_j \to \nn$ can be encoded using $m_j\log n$ bits. Taking $m=\max(m_1,\dots,m_{J-1})$,
we have
\begin{equation*}
\size(\sigma_0, \sigma_1, \dots, \sigma_{J-1}) \leq Jm\log n.
\end{equation*}
Finally, the constraint relation $R_j$ requires at most $q^{m_j}\cdot m_j \cdot \log q$ bits. Thus, 
\begin{equation*}
\size(R_0, R_1, \dots, R_{J-1}) \leq Jq^m\cdot m\log q.
\end{equation*}
Combining these encodings and using the fact that $\log q \leq q$, we deduce
\begin{equation}\label{eqn:size}
\size(\sI) \leq npq^{r+1} + Jmq^{m+1} + Jmn.
\end{equation}
In particular, for the problem $\CSP(\mathfrak A,m)$, the maximum scope arity, $m$, is fixed, as is~$r$. In this case, we can assume $J\leq n^m$. Consequently, $\size(\sI) \in O((nq)^{m+1})$ which yields a polynomial bound (in $nq$) for the size of the instance.

A problem is called \defn{tractable} if there exists a deterministic algorithm that solves all instances of the problem and does so in an amount of time that is bounded above by a polynomial function of the input size.  We can use Definition~\ref{def:csp} above to classify the complexity of an algebra $\bA$, or collection of algebras $\mathfrak A$, according to the complexity of their corresponding constraint satisfaction problems.

An algorithm $\sansA$ is called a \defn{polynomial-time algorithm} for $\CSP(\mathfrak A)$ 
if there exist constants $c$ and $d$ such that, given an instance $\sI$ of $\CSP(\mathfrak A)$ of size $S= \size(\sI)$, $\sansA$ halts in at most $c S^d$ steps and outputs whether or not $\sI$ has a solution.  In this case, we say $\sansA$ ``solves the decision problem $\CSP(\mathfrak A)$ in polynomial time'' and we call the algebras in $\mathfrak A$ ``jointly tractable.'' We say that $\mathfrak A$ is \emph{jointly locally tractable} if, for every
natural number, $m$, there is a polynomial-time algorithm $\sansA_m$ that solves $\CSP(\mathfrak A,m)$.
We call an algebra $\bA$ \emph{tractable} when $\mathfrak A = \sansS(\bA)$ is jointly tractable.

We emphasize that $\CSP(\mathfrak A)$ is a \emph{decision problem,} that is, the algorithm is only required to respond ``yes'' or ``no'' to the question of whether a particular instance has a solution, it does not have to actually produce a solution. However, it turns out that if $\CSP(\mathfrak A)$ is tractable then the corresponding \emph{search problem} is also tractable; thus, one could design the algorithm to find a solution in polynomial time, if a solution exists,  see~\cite[Cor~4.9]{MR2137072}.

\subsection{Sufficient conditions for tractability}%
\label{ssec:edge-sdm}
A \emph{meet semidistributive} lattice is one that satisfies the quasiidentity $x\meet y \approx x\meet z \rightarrow x\meet y \approx x\meet (y\join z)$. A variety is called \sdm if every member of the variety has a meet semidistributive congruence lattice. Idempotent \sdm varieties are known to be Taylor varieties~\cite{HM:1988}. In~\cite{MR2893395}, Barto and Kozik proved the following.

\begin{thmC}[\cite{MR2893395}]\label{thm:sdm-tractable}
A finite idempotent algebra in an \sdm variety is tractable.
\end{thmC}

A second important technique for establishing tractability is the ``few subpowers algorithm'' of~\cite{MR2678065}, which its discoverers describe as a broad generalization of Gaussian elimination.

\begin{defi}\label{defn:edge-term}
Let $\var{V}$ be a variety and $k$ an integer, $k>1$. A $(k+1)$-ary term $t$ is called a \emph{$k$-edge term for $\var{V}$} if the following $k$ identities hold in $\var{V}$:
\begin{align*}
t(y,y,x,x,x,\dots,x) &\approx x\\
t(y,x,y,x,x,\dots,x) &\approx x\\
t(x,x,x,y,x,\dots,x) &\approx x\\
&\vdots\\
t(x,x,x,x,x,\dots,y) &\approx x.
\end{align*}
\end{defi}
Clearly, every \malcev term and every near unanimity term is an edge term, and every edge term is idempotent and Taylor. Combining the main results of~\cite{MR2563736} and~\cite{MR2678065} yields the following theorem.

\begin{thm}\label{thm:edge-tractable}
Let \bA\ be a finite idempotent algebra with an edge term. Then \bA\ is tractable.
\end{thm}

Finally, we note that tractability is largely preserved by familiar algebraic constructions.

\begin{thmC}[\cite{MR2137072}]\label{thm:HSP-tract}
Let \bA\ be a finite, idempotent, tractable algebra. Every subalgebra and finite power of \bA\ is tractable. Every homomorphic image of \bA\ is locally tractable.
\end{thmC}

\subsection{Rectangularity Theorem: obstacles and applications}
The goal of this subsection is to consider aspects of the Rectangularity Theorem that limit its utility as a tool for proving tractability of \csps.
We give a brief overview of the potential obstacles, and then consider each one in more detail.
\begin{enumerate}
\item\label{item:abelian-potatoes-tractable}
  {\bf Abelian factors must have easy partial solutions.} One potential limitation concerns abelian factors in the product algebra associated with a \csp instance. Indeed, application of Corollary~\ref{cor:RT-cor-gen} requires a partial solution $\bx \in \bigcap \Proj_\alpha R_\ell$ to the restricted instance derived from projecting constraints onto abelian factors. Fortunately, Section~\ref{sec:tract-abel-algebr} will show that this obstacle is easily overcome.

\item {\bf Intersecting products of minimal absorbing subalgebras.}
  The Rectangularity Theorem (and corollaries) assumes the universes of the subdirect products in question intersect nontrivially with a single product $\myprod B_i$ of minimal absorbing subuniverses.
  We refer to ``minimal absorbing subuniverses'' quite frequently, so from now on we call them \defn{masses}; that is,
  \begin{itemize}
    \item \defn{mass} := {\bf m}inimal {\bf a}bsorbing {\bf s}ubuniver{\bf s}e; 
    \item \defn{mass product} := product of \masses. 
  \end{itemize}
  Assumption~(\ref{item:RT-cor-gen-4}) of Corollary~\ref{cor:RT-cor-gen} requires all constraint relations intersect nontrivially with a single \mas product. This is a real limitation, as we demonstrate in Section~\ref{sec:mass-products}.

\item {\bf Nonabelian factors must be simple.}
  This is the most obvious limitation of the theorem and at this point we don't have a completely general means of overcoming it.  However, there are some work-arounds that can be useful in special cases; we describe these below.
\end{enumerate}

\noindent
In the next two subsections we address potential limitations (1) and (2). In Section~\ref{sec:var-reduc} we develop alternative methods for proving tractability of nonsimple algebras, and in Section~\ref{sec:csps-comm-idemp} we apply these methods in the special setting of ``commutative idempotent binars.''

\subsubsection{Tractability of abelian algebras}%
\label{sec:tract-abel-algebr}
To address concern~(\ref{item:abelian-potatoes-tractable}) above we observe that finite abelian algebras yield tractable \csps. We show how to use this and the Rectangularity Theorem to find solutions to \csp instances (or prove none exists). To begin, recall the result of tame congruence theory~\cite[Thm~7.12]{HM:1988}  that we reformulated above in Lemma~\ref{lem:HM-thm-7-12}.
As noted in Remark~\ref{rem:abelian-quotients} above, this result has the following important corollary.
(For a proof that avoids tame congruence theory, see~\cite[Thm~5.1]{MR3374664}.)
\begin{thm}%
  \label{thm:type2cp}
Let $\var{V}$ be a locally finite variety with a Taylor term. Every finite abelian member of $\var{V}$ generates a congruence-permutable variety. Consequently, every finite abelian member of $\var{V}$ is tractable.
\end{thm}

Let $\bA$ be a finite algebra in a Taylor variety and fix an instance $\sI = \<\sV, \sA, \sC\>$ of $\CSP(\sansS(\bA))$ with $n = |\sV|$ variables.
Suppose all nonabelian algebras in the list $\sA$ are simple. Let $\alpha \subseteq \nn$ denote the indices of the abelian algebras in $\sA$,
and assume without loss of generality that $\alpha = \{0,1,\dots, q-1\}$. That is, $\bA_0$, $\bA_1$, $\dots$, $\bA_{q-1}$ are finite idempotent abelian algebras. Consider the restricted instance $\sI_\alpha$ obtained by dropping all constraints with scopes that do not intersect $\alpha$, and by restricting the remaining constraint relations to the abelian factors. Since the only algebras involved in $\sI_\alpha$ are abelian, this is an instance of a
tractable \csp.  Therefore, we can check in polynomial-time whether or not $\sI_\alpha$ has a solution. If there is no solution, then the original instance $\sI$ has no solution. On the other hand, suppose $f_\alpha\in \myprod_{j\in \alpha}A_j$ is a solution to $\sI_\alpha$. In Corollary~\ref{cor:RT-cor-gen}, to reach the conclusion that the full instance has a solution, we required a partial solution $\bx \in \bigcap \Proj_\alpha R_\ell$. This is precisely what $f_\alpha$ provides.

To summarize, we look for a partial solution by restricting the instance to abelian factors and, if successful, we use the partial solution for $\bx$ in Corollary~\ref{cor:RT-cor-gen}. Then, if the remaining hypotheses of Corollary~\ref{cor:RT-cor-gen} hold, we conclude that a solution to the original instance exists. If no solution to the restricted instance exists, then the original instance has no solution. Thus, assumption~(\ref{item:RT-cor-gen-5}) of
Corollary~\ref{cor:RT-cor-gen} does not limit the application of the result.

\subsubsection{Mass products}%
\label{sec:mass-products}
This section concerns products of minimal absorbing subalgebras, or ``\mas products.'' Hypothesis~(\ref{item:RT-cor-gen-4}) of Corollary~\ref{cor:RT-cor-gen} assumes that all constraint relations intersect nontrivially with a single \mas product. Instances where this hypothesis does not hold are easy to contrive. For example, take the algebra $\bA = \<\{0,1\}, m\>$, where $m \colon A^3 \to A$ is an idempotent majority operation---that is,
$m(x,x,x)\approx x$ and $m(x,x,y)\approx m(x,y,x)\approx m(y,x,x) \approx x$. Let $\bR = \<R, m\>$ and $\bS=\<S, m\>$ be subdirect products of
$\bA^3$ with universes $R = \{(0,0,0), (0,0,1), (0,1,0), (1,0,0)\}$ and $S = \{(0,1,1), (1,0,1), (1,1,0), (1,1,1)\}$.
Then there are \mas products that intersect nontrivially with either $R$ or $S$, and each of $R$ and $S$ fully contains every \mas product that it intersects. However, there is no single \mas product intersecting nontrivially with both $R$ and $S$.


The example described above is very special. In particular, there is no solution to the instance with constraint relations $R$ and $S$. We might hope that when an instance \emph{does} have a solution, then there must be a solution that passes through a \mas product. As we now demonstrate, this is false. In fact, Example~\ref{ex:mass-products-3} presents subdirect powers that intersect nontrivially (so the instance has a solution), yet their intersection avoids every \mas product (so the Rectangularity Theorem cannot locate a solution nor prove none exists).
\begin{prop}\label{claim:mass-products-2}
There exists an algebra $\bA$ with subdirect powers $\bR$ and $\bS$ such that $R \cap S \neq \emptyset$ and
$R \cap S \cap \myprod B_i = \emptyset$, for every collection $\{B_i\}$ of \masses.
\end{prop}

\noindent We prove this by simply producing an example that meets the stated conditions.

\begin{exa}\label{ex:mass-products-3}
Let $\bA = \<\{0,1,2\}, \circ\>$ be an algebra with binary operation $\circ$ given by
\vskip3mm
 \begin{center}
 \begin{tabular}{c|ccc}
      $\circ$ & 0 & 1 & 2 \\
      \hline
      0 & 0 & 1 & 2\\
      1 & 1 & 1 & 0\\
      2 & 2 & 0 & 2
    \end{tabular}
 \end{center}
\vskip3mm
The proper nonempty subuniverses of $\bA$ are $\{0\}$, $B_1:=\{1\}$, $B_2:=\{2\}$, $\{0,1\}$, and $\{0,2\}$. Note that $B_1$ and $B_2$ are masses with respect to the term $t(x,y,z,w) = (x \circ y) \circ (z \circ w)$. Note also that $\{0\}$ is not an absorbing subuniverse of~$\bA$. For if $\{0\}\absorbing \bA$, then $\{0\}\absorbing \{0,1\}$ which is  false, since $\{0,1\}$ is a semilattice with absorbing element $1$.

Let $\bA_0 \cong \bA \cong \bA_1$, $R = \{(0,0), (1,1), (2,2)\}$, and $S = \{(0,0), (1,2), (2,1)\}$.  Then $\bR$ and $\bS$ are subdirect products of $\bA_0\times \bA_1$ and $R\cap S= \{(0,0)\}$.  There are four masses of $\bA_0 \times \bA_1$, namely,
$B_i \times B_j = \{(i, j)\}$, for $i, j \in \{1, 2\}$. Finally, observe,
\[R\cap (B_i \times B_j) = \begin{cases}  \{(i,j)\}, & i=j,\\  \emptyset, & i\neq j;\end{cases}
\qquad
S\cap (B_i \times B_j) = \begin{cases}  \emptyset, & i=j\\  \{(i,j)\}, & i\neq j.\end{cases}\]
Thus, $R \cap S \cap (B_i \times B_j) = \emptyset$ for all $i, j \in \{1, 2\}$, proving Proposition~\ref{claim:mass-products-2}.
\end{exa}

\subsection{Algorithm synthesis for heterogeneous problems}%
\label{sec:heter-potat}
We now present a new result (Theorem~\ref{thm:fry-pan2}) which, like the Rectangularity Theorem, aims to describe some of the elements that
must belong to certain subdirect products. The conclusions we draw are weaker than those of Theorem~\ref{thm:rectangularity}, but the hypotheses required here are simpler, and the motivation is different.

We have in mind subdirect products of ``heterogeneous''  families of algebras. What we mean by this is described and motivated as follows:
let $\sC_1$, $\dots$, $\sC_m$ be classes of algebras, all of the same signature. Perhaps we are lucky enough to possess a single algorithm that
proves the algebras in $\bigcup\sC_i$ are jointly tractable. Suppose instead that we have no such single algorithm at hand, but at least we do happen to know that the classes $\sC_i$ are \emph{separately tractable}; that is, for each $i$ we have a proof (or program) $\P_i$ certifying the tractability of $\CSP(\sC_i)$.  It is natural to ask whether and under what conditions it might be possible to derive from $\{\P_i \mid 1\leq i \leq m\}$ a single proof (program) certifying that the algebras in $\bigcup \sC_i$ are \emph{jointly tractable}, so that $\CSP(\bigcup \sC_i)$ is tractable.

The results in this section take a step in this direction by considering two special classes of algebras that are known to be separately tractable,
and demonstrating that they are, in fact, jointly tractable. We apply this tool in Section~\ref{sec:block-inst} to prove tractability of a \csp involving algebras that were already known to be tractable, but not previously known to be jointly tractable.

The results here involve special term operations called cube terms and transitive terms. A $k$-ary idempotent term $t$ is a \emph{cube term} if for every coordinate $i \leq k$, $t$ satisfies an identity of the form $t(x_1, \dots, x_k) \approx y$, where $x_1,\dots, x_k \in \{x, y\}$ and $x_i = x$. (See~\cite{MR2563736}.) A $k$-ary operation $f$ on a set $A$ is called \emph{transitive in the $i$-th coordinate} if for every $u, v \in A$, there exist $a_1,\dots, a_k \in A$ such that $a_i = u$ and $f(a_1 ,\dots, a_n) = v$. Operations that are transitive in every coordinate are called \emph{transitive}. (See~\cite{MR3374664}.)

\begin{fact}%
  \label{fact:cubes-are-trans}
  Let $\bA$ be a finite idempotent algebra and suppose $t$ is a cube term operation on $\bA$. Then $t$ is a transitive term operation on $\bA$.
\end{fact}
\begin{proof}
  This follows immediately from the definitions. 
\end{proof}

\begin{fact}\label{fact:pseudovar}
  The class
  \begin{equation}\label{eq:0001}
    \sT = \{\bA \mid \bA \text{ finite and every subalgebra of $\bA$ has a transitive term op}\}
  \end{equation}
  is closed under the taking of homomorphic images, subalgebras, and finite products. That is, $\sT$ is a \emph{pseudovariety}.
\end{fact}

We also require the following obvious fact about nontrivial algebras in a Taylor variety. (We call an algebra \emph{nontrivial} if it has more than one element in its universe.)
\begin{fact}\label{fact:nonconstant-terms-exist}
  If $\bA$ and $\bB$ are nontrivial algebras in a Taylor variety $\var{V}$, then for some $k>1$ there is a $k$-ary term $t$ in $\var{V}$ such that $t^{\bA}$ and $t^{\bB}$ each depends on at least two of its arguments.
\end{fact}

Finally, we are ready for the main result of this section.

\begin{thm}\label{thm:fry-pan2}
  Let $\bA_0, \bA_1, \dots, \bA_{n-1}$ be finite idempotent algebras in a Taylor variety and assume there exists a proper nonempty subset $\alpha \subset \nn$ such that
  \begin{itemize}
  \item $\bA_i$ has a cube term for all $i \in \alpha$,
  \item $\bA_j$ has a sink $s_j \in A_j$ for all $j \in \alpha'$;  let $\bs \in \prod_{\alpha'}A_j$ be a tuple of sinks.
  \end{itemize}
   If $\bR \sdp \prod_{\nn} \bA_i$, then the set $X:=R_\alpha \times \{\bs\}$ is a subuniverse of $\bR$.
\end{thm}
\begin{rem}
To foreshadow applications of Theorem~\ref{thm:fry-pan2}, imagine we have algebras of the type described, and a collection $\sR$ of subdirect products of these algebras. Suppose also that we have somehow determined that the intersection of the $\alpha$-projections of these subdirect
products is nonempty, say, $\bx_\alpha \in \bigcap_{R \in \sR} R_\alpha$. Then the full intersection $\bigcap \sR$ will also be nonempty, since according to the theorem it must contain the tuple that is $\bx_\alpha$ on $\alpha$ and $\bs$ off $\alpha$.
\end{rem}

\begin{proof}
  Fix $\bx \in X$, so $\bx\circ \alpha' = \bs$ and $\bx \circ \alpha \in R_\alpha$. We prove $\bx \in R$.
  Define $\bA_{\alpha} =\prod_{\alpha}\bA_i$ and $\bA_{\alpha'} =\prod_{\alpha'}\bA_i$, so $R_{\alpha}$ and $R$ are subuniverses of $\bA_\alpha$ and $\bA_\alpha \times \bA_{\alpha'}$, respectively.

  By Fact~\ref{fact:cubes-are-trans},  for each $i \in \alpha$ every subalgebra of $\bA_i$ has a transitive term operation.  Moreover, the class $\sT$ defined in~(\ref{eq:0001}) is a pseudovariety, so every subalgebra of $\bA_{\alpha}$ has a transitive term operation.  Suppose there are $J$ subalgebras of $\bA_{\alpha}$ and let $\{t_j\mid 0\leq j < J\}$ denote the collection of transitive term operations, one for each subalgebra. Then it is not hard to prove that $t := t_0 \star  t_1 \star \cdots \star  t_{J-1}$ (defined in Sec.~\ref{sec:absorption-thm}) is a transitive term for every subalgebra of $\bA_\alpha$.  In particular $t$ is transitive for the subalgebra with universe $R_{\alpha}$.

Assume $t$ is $N$-ary. Fix $j\in \alpha'$.  
We may assume without loss of generality that $t^{\bA_j}$ depends on its first argument, at position 0. (It must depend on at least one argument by idempotence.) Now, since $\bR$ is a subdirect product of $\prod_{\nn}\bA_i$, there exists $\br^{(j)} \in R$ such that $\br^{(j)}(j) = s_j$, the sink in $\bA_j$.  Since   $\bx\circ \alpha\in R_\alpha$ and since $t^{\bA_\alpha}$ is transitive over $\bR_\alpha$, there exist $\br_1, \dots, \br_{N-1}$ in $R$ such that
  \begin{align*}
   \vy^{(j)}&:= t^{\bA_{\alpha} \times \bA_{\alpha'}}(\br^{(j)}, \br_1, \dots, \br_{N-1})\\
      &= (t^{\bA_{\alpha}}(\br^{(j)}\circ \alpha, \br_1\circ \alpha, \dots, \br_{N-1}\circ \alpha),
    t^{\bA_{\alpha'}}(\br^{(j)}\circ \alpha', \br_1\circ \alpha', \dots, \br_{N-1}\circ \alpha'))
  \end{align*}
  belongs to $R$  and satisfies $\vy^{(j)}\circ \alpha = \bx \circ \alpha$ and $\vy^{(j)}(j) = s_j$.


  If $\alpha' = \{j\}$, we are done, since $\vy^{(j)} = \bx$ in that case. If $|\alpha'|>1$, then we repeat the foregoing procedure for each $j \in \alpha'$ and obtain a subset $\{\vy^{(j)} \mid j \in \alpha'\}$ of $R$, each member of which agrees with $\bx$ on $\alpha$ and has a sink in some position
  $j\in \alpha'$.

  Next, choose distinct $j, k \in \alpha'$. Suppose $w$ is a Taylor term for $\var{V}$. Then by Fact~\ref{fact:nonconstant-terms-exist} we may assume without loss of generality that $w^{\bA_j}$ depends on its $p$-th argument and $w^{\bA_k}$ depends on its $q$-th argument, for some $p\neq q$. Consider
  \begin{align*}
    \bz:= w^{\Pi \bA_i}(\vy^{(j)}, \dots, \vy^{(j)}, & \, \vy^{(k)}, \vy^{(j)}, \dots, \vy^{(j)})\\
    &\; ^{\widehat{\lfloor}} \, \text{$q$-th argument}
  \end{align*}
  Evidently, $\bz(j) = s_j$, $\bz(k) = s_k$, and $\bz \circ \alpha = \bx \circ \alpha$ by idempotence, since, when restricted to indices in $\alpha$, all the input arguments agree and are equal to $\bx \circ \alpha$.  If $\alpha' = \{j, k\}$, we are done.  Otherwise, choose $\ell \in \alpha'-\{j, k\}$, and again $w^{\bA_\ell}$ depends on at least one of its arguments, say, the $u$-th.  Let
  \begin{align*}
    \bz':= w^{\Pi \bA_i}(\bz, \dots, \bz, & \, \vy^{(\ell)}, \bz \dots, \bz).\\
    &\; ^{\widehat{\lfloor}} \, \text{$u$-th argument}
  \end{align*}
  Then $\bz'$ belongs to $R$, agrees with $\bx$ on $\alpha$, and satisfies $\bz'(j) = s_j$, $\bz'(k) = s_k$, and   $\bz'(\ell) = s_\ell$.
  Continuing in this way until the set $\alpha'$ is exhausted produces an element in $R$ that agrees with $\bx$ everywhere.
  In other words,  $\bx$ itself belongs to $R$.
\end{proof}

In Section~\ref{sec:block-inst} we apply Theorem~\ref{thm:fry-pan2} in the special case where ``has a cube term''  in the first hypothesis
is replaced with ``is abelian.'' Let us be explicit.

\begin{cor}\label{cor:fry-pan}
  Let $\bA_1, \dots, \bA_{n-1}$ be finite idempotent algebras in a locally finite Taylor variety.  Suppose there exists $0< k < n-1$ such that
  \begin{itemize}
  \item $\bA_i$ is abelian for all $i < k$;
  \item $\bA_i$ has a sink $s_i \in A_i$ for all $i \geq k$.
  \end{itemize}
  If $\bR \sdp \prod \bA_i$, then $R_{\kk} \times \{s_k\} \times \{s_{k+1}\} \times \cdots \times \{s_{n-1}\}  \subseteq R$.
\end{cor}
\begin{proof}
  Since $\bA_{\kk} := \prod_{i<k} \bA_i$ is abelian and lives in a locally finite Taylor variety, there exists a term $m$ such that $m^{\bA_{\kk}}$ is a \malcev operation on $\bA_{\kk}$ (Theorem~\ref{thm:type2cp}). Since a \malcev term is a cube term, the result follows from Theorem~\ref{thm:fry-pan2}.
\end{proof}
\ifthenelse{\boolean{draft}}{\newpage}{}

\section{Problem Instance Reductions}\label{sec:var-reduc}
In this section we develop some useful notation for taking an instance of
a \csp and restricting or reducing it in various ways,
either by removing variables or by reducing modulo a sequence of congruence relations.
The utility of these tools will be demonstrated in Section~\ref{sec:csps-comm-idemp}.

Throughout this section, $\bA$ will denote a finite idempotent algebra.
The problem we will focus on is $\CSP(\mathfrak A)$, defined in
Section~\ref{sec:defin-constr-satisf}, and we will be particularly interested
in the special case in which $\mathfrak A = \sansS(\bA)$.

Recall, we denote an $n$-variable instance of $\CSP(\sansS(\bA))$ by $\sI = \< \nn, \sA, \sC\>$, where $\nn = \{0, 1, \dots, n-1\}$
represents the set of variables, $\sA = (\bA_0, \bA_1, \dots, \bA_{n-1})$ is a list of $n$ subalgebras of $\bA$, and
$\sC = ((\sigma_0, R_0), (\sigma_1, R_1), \dots, (\sigma_{J-1}, R_{J-1}))$ is a list of constraints with respective arities $\ar(\sigma_j) = m_j$.
Thus, $R_j\subseteq \prod_{i < m_j} A_{\sigma(i)}$. Much of the discussion below refers to an arbitrary constraint in $\sC$. In such cases it will simplify notation to drop subscripts and denote the constraint by $(\sigma, R)$.

\subsection{Variable reductions}\label{sec:variable-reductions}
\subsubsection{Partial scopes and partial constraints}
Consider the restriction of an $n$-variable \csp instance $\sI$ to the first $k$ of its variables, for some $k\leq n$.
This results in a new \emph{partial scope} given by the function $\restr{\sigma}{\sigma^{-1}(\kcapsigma)}$.
Call this the \emph{$k$-partial scope of} $\sigma$ and, to simplify the notation, let
\[\restr{\sigma}{\overleftarrow{\kk}}=\restr{\sigma}{\sigma^{-1}(\kcapsigma)}.\]
If $\kcapsigma = \emptyset$ then the $k$-partial scope of $\sigma$ is the empty function.
To obtain the \emph{$k$-partial constraint of} $(\sigma, R)$, we take the $k$-partial scope of $\sigma$ as the new scope; for the constraint relation we take the restriction of each tuple in $R$ to its first $p = |\kcapsigma|$ coordinates.  If we let
$\restr{R}{\overleftarrow{\kk}}=\restr{R}{\sigma^{-1}(\kcapsigma)}$, then the \emph{$k$-partial constraint} of $(\sigma, R)$ is given by $(\restr{\sigma}{\overleftarrow{\kk}}, \restr{R}{\overleftarrow{\kk}})$.

For example, let $\sigma$ be a scope consisting of variables $2$, $4$, $7$, so $(\sigma(0), \sigma(1),\sigma(2)) = (2,4,7)$.
To find the 5-partial constraint of $(\sigma, R)$, restrict $(\sigma, R)$ to the first $k= 5$ variables of the instance.
Then, $\sigma^{-1}(\kcapsigma) = \sigma^{-1}\{2,4\} = \{0,1\}$, and $\restr{\sigma}{\overleftarrow{\kk}} = (\sigma(0), \sigma(1)) = (2,4)$, so
$\restr{R}{\overleftarrow{\kk}} = \{(x,y) \mid (x,y,z) \in R\}$.

\subsubsection{Partial instances}\label{sec:partial-instances}
The \emph{$k$-partial instance} $\sI_{\kk}$ of $\sI$ is the restriction of $\sI$ to its first $k$ variables. Thus, $\sI_{\kk}$ is the instance with constraint set $\sC_{\kk}$ equal to the set of all $k$-partial constraints of $\sI$. If we let $\Sol(\sI, \kk)$ denote the set of solutions to $\sI_{\kk}$, then $f \in  \Sol(\sI, \kk)$ if and only if for each $j \in J$, we have $f \circ \restr{\sigma_j}{\overleftarrow{\kk}}\; \in\; \restr{R_j}{\overleftarrow{\kk}}$. We might be tempted to call $\Sol(\sI, \kk)$ the set of ``partial solutions,'' but that's a bit misleading since an $f \in \Sol(\sI, \kk)$ may or may not extend to a solution to the full instance $\sI$.

\subsection{Quotient reductions}
Let $\bA$ be a finite idempotent algebra and $\sI = \< \nn, \sA, \sC\>$ an $n$-variable instance of $\CSP(\sansS(\bA))$.
As above, we assume $\sA = (\bA_0, \bA_1, \dots, \bA_{n-1}) \in \sansS(\bA)^n$, $\sC = ((\sigma_0, R_0), (\sigma_1, R_1), \dots, (\sigma_{J-1}, R_{J-1}))$, and $\bR_j\sdp \prod_{\mm_j} A_{\sigma(i)}$, for $0\leq j < J$.

\subsubsection{Quotient instances}
Suppose $\theta_i \in \Con (\bA_i)$ ($0\leq i < n$) and define $\btheta = (\theta_0, \dots, \theta_{n-1})$. 
For $\bx = (x_0, x_1, \dots, x_{n-1}) \in \prod_{\nn} A_i$, let $\bx/\btheta$ denote the tuple whose $i$-th component is the $\theta_i$-class of $\bA_i$ containing $x_i$, so $\bx/\btheta = (x_0/\theta_0, x_1/\theta_1, \dots, x_{n-1}/\theta_{n-1}) \in\myprod_{i \in \nn} A_i/\theta_i$.

We need a more general notation than this since our tuples may come from $\prod_{\mm} A_{\sigma(i)}$ for some scope function
$\sigma\colon \mm \to \nn$.  Viewing  $\bx \in \prod_{\nn} A_i$ (resp., $\btheta$) as a function from $\nn$ to $\bigcup A_i$ (resp., $\bigcup \Con (\bA_i)$), write $\bx \circ \sigma \in \prod_{\mm} A_{\sigma(i)}$ and $\btheta \circ \sigma \in \prod_{\mm} \Con (\bA_{\sigma(i)})$ for the corresponding scope-restricted tuples, and then define
\[(\bx \circ \sigma)/(\btheta \circ \sigma) = (x_{\sigma(0)}/\theta_{\sigma(0)}, x_{\sigma(1)}/\theta_{\sigma(1)}, \dots, x_{\sigma(m-1)}/\theta_{\sigma(m-1)})
\in  \prod_{i\in \mm} A_{\sigma(i)}/\theta_{\sigma(i)}.\]
Given $\btheta$ and $(\sigma, R) \in \sC$, define the set $\sC/\btheta$ of \emph{quotient constraints} by $\{(\sigma, R/\btheta)\mid (\sigma, R) \in \sC\}$, where
\begin{equation}\label{eq:5}
  R/\btheta := \bigl\{\br/(\btheta\circ \sigma)
  \in \prod_{i \in \mm} A_{\sigma(i)}/\theta_{\sigma(i)} \mid \br \in R\bigr\}.
\end{equation}


Here are a few easily proved facts about quotient constraints.
\begin{fact}\label{fact:quotient-sdp}
  If $\btheta \in \prod_{\nn}\Con (\bA_i)$ and if $\bR \sdp \prod_{\mm} \bA_{\sigma(i)}$, then $R/\btheta$ defined in (\ref{eq:5})
  is a subuniverse of $\prod_{\mm}\bA_{\sigma(i)}/\theta_{\sigma(i)}$ and the corresponding subalgebra is subdirect.
\end{fact}

\begin{fact}\label{fact:quotient-instance}
  If $\btheta \in \prod_{\nn}\Con (\bA_i)$ and if  $\sI$ is an $n$-variable instance of $\CSP(\bA)$, then the constraint set $\sC/\btheta$ defines an $n$-variable instance of $\CSP(\sansH \sansS (\bA))$.
\end{fact}
\noindent The \emph{quotient instance} $\sI/\btheta$ is the instance of $\CSP(\sansH \sansS (\bA))$ described in Fact~\ref{fact:quotient-instance}. 

\begin{fact}%
  \label{fact:soln-quotient}
  If $\bx$ is a solution to $\sI$, then $\bx/\btheta$ is a solution to $\sI/\btheta$.
\end{fact}
\noindent By Fact~\ref{fact:soln-quotient}, if there is a quotient instance with no solution, then $\sI$ has no solution. However, there may be solutions to every proper quotient instance and yet no solution to $\sI$.

\subsubsection{Block instances}%
\label{sec:block-inst}
Let $\mathfrak A$ be a collection of finite idempotent algebras.
Let $\sI = \< \nn, \sA, \sC\>$ be an $n$-variable instance of $\CSP(\mathfrak A)$, where $\sA = (\bA_0, \bA_1, \dots, \bA_{n-1})$ and
$\sC$ is a finite set of constraints. If $\btheta \in \prod_{\nn}\Con (\bA_i)$ and  $\bx \in \prod_{\nn}A_i$, then the tuple $\bx/\btheta$ of $\theta_i$-blocks is actually a list of algebras, by idempotence.

For each constraint $(\sigma, R)\in \sC$, consider restricting the relation $R$ to the $x_i/\theta_i$-classes in its scope $\sigma$.
In other words, replace the constraint $(\sigma, R)$ with the \emph{block constraint} $(\sigma, R \cap \Pi_{\sigma} \bx/\btheta)$, where we have defined $\Pi_{\sigma} \bx/\btheta$ to be the product of the blocks $x_{\sigma(i)}/\theta_{\sigma(i)}$ ($1\leq i \leq m$). Finally, let $\sI_{\bx/\btheta} = \< \nn, \bx/\btheta, \sC_{\bx/\btheta}\>$ denote the problem instance of $\CSP(\sansS(\sA))$ specified by the constraint set
$\sC_{\bx/\btheta} := \{(\sigma, R \cap \Pi_{\sigma}\bx/\btheta)\mid(\sigma, R) \in \sC\}$.
We call $\sI_{\bx/\btheta}$ the \emph{$\bx/\btheta$-block instance of} $\sI$.
Evidently, a solution to $\sI_{\bx/\btheta}$ is also a solution to $\sI$.

\subsection{The Quotient-Block strategy}
The notions ``quotient instance'' and ``block instance'' suggest the following strategy for solving \csps that works in certain special cases (one of which we will see in Example~\ref{ex:quot-inst}): First search for a solution $\bx/\btheta$ to the quotient instance $\sI/\btheta$ for some conveniently chosen $\btheta$. If no quotient solution exists, then the original instance $\sI$ has no solution. Otherwise, if $\bx/\btheta$ is a solution to $\sI/\btheta$, then try to solve the $\bx/\btheta$-block instance of $\sI$. If successful, then the instance $\sI$ has a solution. Otherwise, try again with a different solution to $\sI/\btheta$. If all solutions to the quotient instance are exhausted without finding a solution to a corresponding block instance, then $\sI$ has no solution.

This approach is effective as long as we can find a quotient instance $\sI/\theta$ for which there is a polynomial bound on the number of quotient solutions. Although this is not always possible (consider instances involving only simple algebras!), there are situations in which an appropriate choice of congruences makes it easy to check that every instance has quotient instance with a very small number of solutions.  We present an example.

\begin{exa}\label{ex:quot-inst}
Let $\bA = \<\{0,1,2,3\}, \cdot \>$, be an algebra with a single binary operation given by the table on the left in Figure~\ref{tab:final-4}.
\begin{figure}
\centering
 \begin{tabular}{c|cccc}
      $\cdot $ & 0 & 1 & 2 & 3\\
      \hline
      0 & 0 & 0 & 3& 2\\
      1 & 0 & 1 & 3& 2\\
      2 & 3 & 3 & 2 & 1\\
      3 & 2 & 2 & 1 & 3
 \end{tabular}
 \hskip1cm
  \begin{tabular}{c|ccc}
   $\cdot^{\Sq 3}$&0&1&2\\
  \hline
  0&0&2&1\\
  1&2&1&0\\
  2&1&0&2
  \end{tabular}
 \hskip1cm
  \begin{tabular}{c|cccc}
      $t$ & 0 & 1 & 2 & 3\\
      \hline
      0 & 0 & 0 & 0& 0\\
      1 & 0 & 1 & 1& 1\\
      2 & 2 & 2 & 2 & 2\\
      3 & 3 & 3 & 3 & 3
    \end{tabular}
  \caption{Operation tables of Example~\ref{ex:quot-inst}:
    for the algebra $\bA$ (left);
    for the quotient $\bA/\Theta \cong \Sq 3$ (middle);
    for $t(x,y) = x\cdot (y\cdot (x\cdot y))$ (right).}%
  \label{tab:final-4}
\end{figure}
The proper nonempty subuniverses are $\{0,1\}$, $\{1,2,3\}$, and the singletons. The algebra \bA\ has a single proper nontrivial congruence relation, $\Theta$, with partition $|01|2|3|$. The quotient algebra $\bA/\Theta$ is a 3-element Steiner quasigroup, which happens to be abelian. Denote the latter by \Sq 3 and note the middle table in Figure~\ref{tab:final-4} describing its binary operation. Observe that the subalgebra $\{1,2,3\}$ of $\bA$ is also isomorphic to $\Sq 3$.

As Theorem~\ref{thm:type2cp} predicts, the algebra $\bA/\Theta$ has a \malcev term $q(x,y,z)=y\cdot(x\cdot z)$.
Let $s(x,y) = q(x,y,y) = y\cdot(x\cdot y)$. Then $\bA/\Theta\vDash s(x,y)\approx x$. By iterating the term $s$ on its second variable, we arrive at a term $t(x,y)$ such that, in~\bA, $t(x,t(x,y)) = t(x,y)$. In fact, $t(x,y)=x\cdot(y\cdot(x\cdot y))$. To summarize
\begin{align*}
\bA &\vDash t(x,t(x,y)) \approx t(x,y) \\
\bA/\Theta &\vDash t(x,y) \approx x.
\end{align*}
The table for $t$ appears on the right in Figure~\ref{tab:final-4}.

For the next result, we denote the two-element semilattice by $\slt=\<\{0,1\},\meet\>$.

\begin{lem}\label{lem:sq3s2}
 Let $\mathfrak A = \{\slt,\Sq 3\}$. Then $\CSP(\mathfrak A)$ is tractable.
\end{lem}

\begin{proof}
Let $\sI=\<\nn,\sA, \sC\>$ be an instance of $\CSP(\mathfrak A)$. Recall, the set of solutions to $\sI$ is $\Sol(\sC, \nn) = \bigcap_{\sC} R^{\overleftarrow{\sigma}}$. We shall apply Corollary~\ref{cor:fry-pan} to establish tractability.
We have $\sA = (\bA_0, \bA_1,\dots,\bA_{n-1}) \in \mathfrak A^n$. Let $\alpha=\{\,i : \bA_i \cong \Sq 3\,\}$ and $\alpha'=\{\,i : \bA_i \cong \slt\,\}$, so $\nn = \alpha \cup \alpha'$.  We may assume without loss of generality that, for each constraint $(\sigma, R)$ appearing in $\sC$, the associated algebra $\bR$ is a subdirect product of $\prod_{\mm}\bA_{\sigma(i)}$, where $\bA_{\sigma(i)} \cong \Sq 3$ for all $\sigma(i) \in \alpha$, and  $\bA_{\sigma(i)} \cong \slt$ for all $\sigma(i) \in \alpha'$.

Let $0$ denote the bottom element of each semilattice. For each constraint $(\sigma, R)$, let $\restr{R}{\overleftarrow{\alpha}}$ be an abbreviation for  $\restr{R}{\sigma^{-1}(\alpha \cap \im \sigma)}$ (the projection of $R$ onto factors in its scope with indices in $\alpha$). Let $\bzero$ denote a tuple of $0$'s of length $|\alpha' \cap \im \sigma|$. Then Corollary~\ref{cor:fry-pan} implies that
for each constraint $(\sigma, R) \in \sC$ we have
\begin{equation}\label{eq:fry-pan}
 \restr{R}{\overleftarrow{\alpha}} \times \{\bzero\} \subseteq R.
\end{equation}
Obviously, if $\sI$ has a solution, then the $\alpha$-partial instance $\sI_\alpha$ also has a solution. Conversely, if $f \in \prod_{\alpha}A_i$ is a solution to $\sI_\alpha$, then (\ref{eq:fry-pan}) implies that $g \in \prod_{\nn} A_i$ defined by
\[g(i) =
  \begin{cases}
    f(i), & \text{if $i \in \alpha$,} \\
    0, & \text{if $i \in \alpha'$},
  \end{cases}
\]
is a solution to $\sI$. We conclude that $\sI$ has a solution if and only if the $\alpha$-partial instance has a solution.  Since abelian algebras yield tractable \csps, the proof is complete.
\end{proof}

Now we address the tractability of the example at hand.

\begin{prop}\label{prop:quotient-block}
  If $\langle \bA, \cdot\rangle$ is the four-element algebra with operation table given in Figure~\ref{tab:final-4}, then $ \CSP(\sansS(\bA))$ is tractable.
\end{prop}
\begin{proof}
  Let $\sI = \<\nn, \sA, \sC\>$ be an instance of $\CSP(\sansS(\bA))$ where$\sA= (\bA_0, \bA_1, \dots, \bA_{n-1}) \in \sansS(\bA)^n$.
  For each $i \in \nn$, let $\theta_i = \Theta$ if $A_i = A$, and let $\theta_i = 0_{A_i}$ otherwise. The universes $A_i/\theta_i$ of the quotient algebras satisfy
  \begin{equation*}
  A_i/\theta_i =
  \begin{cases}
    \{\{0,1\}, \{2\}, \{3\}\},  & \text{if $A_i = \{0,1,2,3\}$,}\\
    \{\{1\},\{2\},\{3\}\}, & \text{if $A_i = \{1,2,3\}$,}\\
    \{\{0\},\{1\}\}, & \text{if $A_i = \{0,1\}$.}
  \end{cases}
  \end{equation*}
  In the first two cases $\bA_i/\theta_i$ is isomorphic to $\Sq 3$ and in the third case $\bA_i/\theta_i$ is a 2-element meet semilattice.
  (Singleton factors, $\bA_i/\theta_i =\{a\}$, are safely ignored because all solutions must assign the the value $a$ to the associated variable.)

  Let $\btheta = (\theta_0, \theta_1, \dots, \theta_{n-1})$ and consider the quotient instance $\sI/\btheta$ which, by the observations in the preceding paragraph, is an instance of $\CSP(\mathfrak A)$, where  $\mathfrak A  = \{\Sq 3,\slt\}$. By Lemma~\ref{lem:sq3s2}, $\sI/\btheta$ can be solved in polynomial-time. If $\sI/\btheta$ has no solution, then neither does $\sI$.  Otherwise, let $f  \in \Sol(\sI/\btheta, \nn)$ be a solution to the quotient instance. We will use $f$ to construct a solution $g \in \Sol(\sI, \nn)$ to the original instance. (In so doing, we often map a singleton set to its inhabitant via the usual set union operation; e.g., $\bigcup\{x\} = x$.)

  Assume $\nn = \alpha\cup \beta \cup \gamma$ is a disjoint union where
  \[ A_i =
  \begin{cases}
    \{0,1,2,3\},  & \text{if $i \in \alpha$,}\\
    \{1,2,3\}, & \text{if $i \in \beta$,}\\
    \{0,1\}, & \text{if $i \in \gamma$.}
  \end{cases}
  \]
  \noindent In other terms, $\prod_{\nn} A_i \cong \{0,1,2,3\}^\alpha \times \{1,2,3\}^\beta \times \{0,1\}^\gamma$.  Assume the factors are rearranged if necessary so that this correspondence is on-the-nose equality. By definition,
  $\theta_i$ is $\Theta$ for $i \in \alpha$ and $0_{A_i}$ for $i\notin \alpha$.
  Thus, the $i$-th entry of $\sA/\btheta$ is
  \[
  A_i/\theta_i =
  \begin{cases}
    \{\{0,1\}, \{2\}, \{3\}\},  & \text{if $i \in \alpha$,}\\
    \{\{1\},\{2\},\{3\}\}, & \text{if $i \in \beta$,}\\
    \{\{0\},\{1\}\}, & \text{if $i \in \gamma$.}
  \end{cases}
  \]
  The $i$-th value $f(i)$ of the quotient solution must belong to $A_i/\theta_i$.
  Let $\alpha = \alpha_0 \cup \alpha_1$ be such that $f(i) = \{0,1\}$ when $i \in \alpha_0$ and $f(i) \in\{ \{2\}, \{3\}\}$ when $i \in \alpha_1$. Then,
  \begin{equation}\label{eq:10}
  f \in \{\{0,1\}\}^{\alpha_0} \times \{\{2\}, \{3\}\}^{\alpha_1} \times
  \{\{1\},\{2\},\{3\}\}^{\beta} \times
  \{\{0\},\{1\}\}^{\gamma}\!.
  \end{equation}
  Note that $f(i)$ is a singleton for all $i \notin \alpha_0$, and recall that $\bigcup \{a\}= a$. Define $g \in \prod_{\nn}A_i$ as follows: $g(i) = 0$ if $i \in \alpha_0\cup \gamma$, and $g(i) = \bigcup f(i)$, if $i\in \alpha_1 \cup \beta$.  We will prove $g\in \Sol(\sI, \nn)$.

Fix an arbitrary constraint $(\sigma, R)\in \sC$, let $S := \im \sigma$ be the set of indices in the scope of $R$, and let $S^c := \nn - S$ denote the complement of $S$ with respect to $\nn$. Since we assumed $f$ solves $\sI/\btheta$ and satisfies (\ref{eq:10})  there must exist $\br\in \prod_{\nn}A_i$ satisfying not only $(\sigma, R)$ but also the following: $\br(i) \in \{0,1\}$ for $i \in \alpha_0$ and $\br(i) = \bigcup f(i)$ for $i \notin \alpha_0$. 

Now consider other elements satisfying the (arbitrary) constraint $(\sigma, R)$.  By subdirectness of $R$, for each $\ell \in \alpha_0 \cup \gamma$ there exists $\bx^\ell \in \prod_{\nn}A_i$ satisfying $(\sigma, R)$ and $\bx^\ell(\ell) = 0$. Since $\br$ and $\bx^\ell$ satisfy $(\sigma, R)$---that is $\br \circ \sigma\in R$ and $\bx^\ell \circ \sigma\in R$---we also have $t(\br,  \bx^\ell) \circ \sigma\in R$.
  Consider the entries of $t(\br, \bx^\ell)$, using the partition $\nn = \alpha_0 \cup \alpha_1 \cup \beta \cup \gamma$ defined above.
  Since $\bx^\ell(\ell) = 0$, we have $t(\br, \bx^\ell)(\ell)=0$.  For $i \neq \ell$, observe the following:
  \begin{itemize}
  \item  If $i\in \alpha_0 - \{\ell\}$, then $\br(i) \in \{0,1\}$, 
    so $t(\br, \bx^\ell)(i)\in\{0,1\}$.
  \item  If $i\in \alpha_1$, then $\br(i) = \bigcup f(i)\in \{2,3\}$, 
    so $t(\br, \bx^\ell)(i)=\bigcup f(i)$.
  \item  If $i\in \beta$, then $\br(i)= \bigcup f(i)\in \{1,2,3\}$, 
    so $t(\br, \bx^\ell)(i) = \bigcup f(i)$.
  \item  If $i\in \gamma-\{\ell\}$, then $\br(i)= \bigcup f(i)\in \{0,1\}$,  
    so $t(\br, \bx^\ell)(i) \in \{0,1\}$.
  \end{itemize}
  To summarize, for each $\ell \in \alpha_0 \cup \gamma$, there exists $t(\br, \bx^\ell) \in \prod_{\nn}A_i$ satisfying $(\sigma, R)$ and having values $t(\br, \bx^\ell)(i) \in \{0,1\}$ for $i \in \alpha_0\cup \gamma$ and
  \[ t(\br, \bx^\ell)(i) =
  \begin{cases}
    0, & i = \ell,\\
    \bigcup f(i), & i\in \alpha_1 \cup \beta.
  \end{cases}
  \]
  Finally, take the product of all members of $\{t(\br, \bx^\ell) \mid \ell \in \alpha_0 \cup \gamma\}$ with respect to the binary operation of the original algebra $\<A, \cdot\>$ to obtain $t(\br, \bx^{\ell_1}) \cdot t(\br, \bx^{\ell_2}) \cdot \cdots \cdot t(\br, \bx^{\ell_{L}}) = g$, where $L = |\alpha_0\cup \gamma|$. This proves that $g \in \Sol((\sigma, R), \nn)$. Since the constraint $(\sigma, R)$ was arbitrary, we have $g \in \Sol(\sI, \nn)$, as desired.
\end{proof}
\end{exa}

Example~\ref{ex:quot-inst} and Proposition~\ref{prop:quotient-block} do not fully illustrate the scope of the quotient-block strategy since, as we saw in the proof, every solution $f$ to the quotient instance $\sI/\btheta$ induced a solution $g$ on the block instance $f/\theta$.  In general, we may not be so lucky. Nonetheless, we demonstrated the quotient-block strategy on this example because the analysis carried out in Section~\ref{sec:csps-comm-idemp} requires that the associated \csp is tractable (as asserted in Proposition~\ref{prop:quotient-block}).

\subsection{The least block algorithm}
This section describes an algorithm that was apparently discovered by ``some subset of'' Petar Markovi\'c and Ralph McKenzie.  We first learned about this by reading Markovi\'c's slides~\cite{Markovic:2011} from a 2011 talk in Krakow which gives a one-slide description of the algorithm. Since no proof (or even formal statement of the assumptions) from the originators of this result seems to be forthcoming, we include our own version and proof in this section.

Let \bA\ be a finite idempotent algebra and let $\theta$ be a congruence on \bA. Recall that by idempotence, every $\theta$-class of $A$ will be a subalgebra of \bA. Suppose that, as an algebra, each $\theta$-class comes from a variety, $\var{E}$, possessing a $(k+1)$-ary edge term, $e$. Define $s(x,y)= e(y,y,x,x,\dots,x)$. Then according to Definition~\ref{defn:edge-term},  $\var{E} \vDash s(x,y) \approx x$.
By iterating $s$ on its second variable, we obtain a binary term $t$ such that
\begin{align}
\bA &\vDash t(x,t(x,y)) \approx t(x,y)\label{eqn:ILT} \\
\var{E} &\vDash t(x,y) \approx x.\label{eqn:Lz}
\end{align}
Assume, finally, that the induced algebra $\<A/\theta, t^{\bA/\theta}\>$ is a linearly ordered
meet semilattice. When this occurs, we call \bA\ an \emph{Edge-by-Chain} (\ec) algebra, and we say, ``$\alg A$ is \ec.'' Our objective in this subsection is to prove the following theorem.

\begin{thm}\label{thm:least-block}
Every finite, idempotent \ec algebra is tractable.
\end{thm}

\begin{proof}
Recall, an algebra \bA\ is \emph{tractable} if the problem $\CSP(\sansS(\bA))$ is solvable in polynomial time. (See Subsection~\ref{sec:inst-size-tract}.)

Let  \bA\ be finite, idempotent and \ec. We continue to use the congruence $\theta$ and term $t$ developed above. Let $\sI$ be an instance of the problem $\CSP(\sansS(\bA))$. Thus $\sI$ is a triple $\<\sV,\sA,\sC\>$ in which
$\sA=(\bA_0, \bA_1,\dots,\bA_{n-1})$ is a list of subalgebras of \bA, see Definition~\ref{def:csp}.  The congruence $\theta$ induces a congruence (which we will continue to call $\theta$) on each $\bA_i$, and the term $t$ will induce a linearly ordered semilattice structure on $A_i/\theta$. In other words, each $\bA_i$ is also an \ec algebra.

Therefore, for each $i<n$, the structure of $\<A_i/\theta,t\>$ will be of the form $C_{i,0} < C_{i,
1} < \cdots < C_{i,q_i}$ in which each $C_{i,j}$ is a $\theta$-class of $\bA_i$. We have the
following useful relationship.
\begin{equation}\label{eqn:almost-meet}
i<n,\, j\leq k,\, u\in C_{i,j},\, v\in C_{i,k} \implies   t(u,v) = u.
\end{equation}
To see this, let $w=t(u,v)$. Then $w\in C_{i,j}$, since the quotient structure is a semilattice under $t$. Therefore, by~\eqref{eqn:ILT}, $t(u,w)=t(u,t(u,v)) = t(u,v) = w$. On the other hand, since $u,w \in C_{i,j} \in \var{E}$ we have $t(u,w)=u$ by~\eqref{eqn:Lz}. Thus $u=w$.

Recall that for an index $k\leq n$, the $k$th-partial instance of $\sI$ is the instance, $\sI_{\kk}$, obtained by restricting $\sI$ to the first $k$ variables (Section~\ref{sec:partial-instances}). The set of solutions to $\sI_{\kk}$ is denoted $\Sol(\sI,\kk)$.

For every $k < n$ we recursively define the index $j_k$ by
\begin{equation*}
j_k = \text{least  $j\leq q_k$ such that $\Sol(\sI,\kplus)\cap \prod_{i=1}^k C_{i,j_i} \neq \emptyset$.}
\end{equation*}
If there is a $k$ such that no such $j_k$ exists, then $j_k, j_{k+1}, \dots$ are undefined.

To complete the proof of Theorem~\ref{thm:least-block}, we need to prove two final claims.
\begin{clm}\label{clm:1}
  $\sI$ has a solution if and only if $j_{n-1}$ is defined.
\end{clm}
\begin{clm}\label{clm:2} $j_i$ can be computed in polynomial time for all $i<n$.
\end{clm}


One direction of Claim~\ref{clm:1} is immediate: if $j_{n-1}$ is defined then $\Sol(\sI)= \Sol(\sI,\nn)$ is nonempty. We address the converse. Assume that $\sI$ has solutions. We argue, by induction on~$k$, that every $j_k$ is defined. For the base step, choose any $g\in \Sol(\sI)$.  Since $g(0) \in A_0$, there is some $\ell$ such that $g(0) \in C_{0,\ell}$. Then $j_0$ is
defined and is at most~$\ell$.

Now assume that $j_0, j_1,\dots,j_{k-1}$ are defined, but $j_k$ is undefined. Let $f\in \Sol(\sI,
\kk)\cap \prod_{i=0}^{k-1} C_{i,j_i}$. Since $j_k$ is undefined, $f$ has no extension to a
member of $\Sol(\sI,\kplus)$. We shall derive a contradiction. Since $\sI$ has a solution,
the restricted instance $\sI_{k+1}$ certainly has solutions. Choose any $g\in \Sol(\sI,k+1)$ with $g(k)
\in C_{k,\ell}$ for the smallest possible~$\ell$. Let $\sC= \bigl((\sigma_0,R_0),
(\sigma_1,R_1),\dots,(\sigma_{J-1},R_{J-1})\bigr)$ be the list of constraints of $\sI$.
By assumption, for each $m<J$, $f$ is a solution to $\restr{(\sigma_m,R_m)}{\kk}$ (the
restriction of $R_m$ to its first $k$ variables). Therefore, there is an
extension, $f_m$, of $f$, to $k+1$ variables such that $f_m$ is a solution to
$\restr{(\sigma_m,R_m)}{\kplus}$.

Let $m<J$ and set $h_m = t(g,f_m)$. Since both $g$ and $f_m$ are
solutions to $\restr{(\sigma_m,R_m)}{\kplus}$, and $t$ is a term, $h_m$ is also
a solution to $\restr{(\sigma_m,R_m)}{\kplus}$.
For every $i<k$ we have $f_m(i)= f(i) = f_0(i)$. Thus
$h_m(i) = t\bigl(g(i), f_m(i)\bigr) = t\bigl(g(i),f_0(i)\bigr) = h_0(i)$ for $i<k$, and $h_m(k) = t\bigl(g(k),f_m(k)\bigr) = g(k)$.
The latter relation holds by~\eqref{eqn:almost-meet}
and the choice of $g$. It follows that every $h_m$ coincides with $h_0$. Since $h_m$ satisfies
$R_m$, we have $h_0 \in \Sol(\sI,\kplus)$. This contradicts our assumption that $f$ has no
extension to a member of $\Sol(\sI,\kplus)$.

Finally, we must verify Claim~\ref{clm:2}. Assume we have computed $j_i$, for $i<k$. To
compute $j_k$ we proceed as follows.
\begin{algorithmic}
 \For{$j=0$ to $q_k$}
 \State $\sB \leftarrow (\bC_{0,j_0}, \bC_{1,j_1}, \dots, \bC_{k-1,j_{k-1}},\bC_{k,j})$
 \If{$\sI_{\sB}$ has a solution in $\CSP(\sB)$}
 \Return $j$
 \EndIf
 \EndFor
 \Return \textsc{Failure}
 \end{algorithmic}
 In this algorithm, $\sI_{\sB}$ is the instance of $\CSP(\sB)$ obtained by first restricting $\sI$ to
 its first $k+1$ variables, and then restricting each constraint relation to $\prod\sB$. Since the
members of $\sB$ all lie in the edge-term variety $\var{E}$, it follows from
Theorems~\ref{thm:edge-tractable} and~\ref{thm:HSP-tract} that $\CSP(\sB)$ runs in polynomial-time.
 \end{proof}

\ifthenelse{\boolean{draft}}{\newpage}{}

\section{CSPs of Commutative Idempotent Binars}%
\label{sec:csps-comm-idemp}

Taylor terms were defined in Section~\ref{ssec:term-ops}. It is not hard to see that a binary
term is Taylor if and only if it is idempotent and commutative. This suggests, in light of the
algebraic \csp-dichotomy conjecture, that we study commutative, idempotent binars (\cib's for
short), that is, algebras with a single basic binary operation that is commutative and
idempotent. If the conjecture is true, then every finite \cib should be tractable.

The associative {\cib}s are precisely the semilattices. Finite semilattices have long been known to be tractable,~\cite{MR1481313}. The variety of semilattices is \sdm and is equationally complete. Every nontrivial semilattice
must contain a subalgebra isomorphic to the unique two-element semilattice, $\slt=\<\{0,1\},\meet\>$, and conversely, \slt\ generates the variety of all semilattices.

As we shall see, the algebra \slt plays a central role in the structure of {\cib}s. In particular, the omission of \slt implies tractability.

\begin{thm}%
  \label{thm:omit5-cib}
  If $\bA$ is a finite \cib then the following are equivalent:
  \begin{enumerate}[label={(\arabic*)}]
  \item $\bA$ has an edge term.
  \item $\var{V}(\bA)$ is congruence modular.
  \item $\slt \notin \sansH \sansS (\bA)$
    \end{enumerate}
\end{thm}

\noindent
The proof that (1) implies (2) appears in~\cite{MR2563736} and holds for general
finite idempotent algebras. The contrapositive of (2) implies (3) is easy:
if
$\slt \in \sansH \sansS (\bA)$, then
$\slt^2 \in \var{V}(\bA)$, and the congruence
lattice of $\slt^2$ is not modular.
So it remains to show that (3) implies (1). For this we use the idea of a
``cube-term blocker'' (\cite{MR2926316}).
  A \emph{cube term blocker} (\ctb) for an algebra $\bA$ is a pair $(D, S)$ of subuniverses
  with the following properties: $\emptyset < D < S\leq A$ and for every term
  $t(x_0, x_1, \dots, x_{n-1})$ of $\bA$ there is an index $i \in \nn$ such that,
  for all $\bs = (s_0, s_1, \dots, s_{n-1}) \in S^n$, if $s_i \in D$
  then $t(\bs)\in D$.

\begin{thmC}[\protect{\cite[Thm 2.1]{MR2926316}}]%
\label{thm:ctb}
Let $\bA$ be a finite idempotent algebra. Then $\bA$ has an edge term iff
it possesses no cube-term blockers.
\end{thmC}

(The notions of cube term and edge term both originate in~\cite{MR2563736}. In that paper it is proved that a finite algebra has a cube term if and only if it has an edge term.)

\begin{lem}%
\label{lem:cib-ctb}
  A finite \cib $\bA$ has a \ctb if and only if $\slt \in \sansH \sansS (\bA)$.
\end{lem}
\begin{proof}
  Assume $(D, S)$ is a \ctb for $\bA$. Then there exists $s\in S-D$.  Consider
  $D^+ :=D \cup\{s\}$.  Evidently  $D^+$ is a subuniverse of $\bA$, and
  if $\bD^+$ denotes the corresponding subalgebra, then
  $\theta_D := D^2 \cup \{(s,s)\}$ is a congruence of $\bD^+$. It's easy to see that
  $\bD^+/\theta_D \cong \slt$, so $\slt \in \sansH \sansS (\bA)$.

  Conversely, if $\slt \in \sansH \sansS (\bA)$, then there exists $\bB \leq \bA$ and
  a surjective homomorphism $h\colon \bB \to \slt$. Let $B_0=h^{-1}(0)$. Then $\emptyset \neq B_0 < B\leq A$, and $(B_0, B)$ is a \ctb for $\bA$.
  \end{proof}

Lemma~\ref{lem:cib-ctb}, along with Theorem~\ref{thm:ctb}, completes the proof of
Theorem~\ref{thm:omit5-cib} by showing (1) is false if and only if (3) is false.
In fact, something stronger is true. Kearnes has shown that if $\var{V}$ is any variety of {\cib}s that omits \slt, then $\var{V}$ is congruence permutable~\cite{Kearnes:comm}.

Before proving the next corollary, we need two more definitions.  A variety is called \defn{regular} if it is defined by regular identities. In contrast, a variety is called \defn{strongly irregular} if it satisfies an identity $t(x, y) \approx x$ for some binary term $t$ in which both $x$ and $y$ appear~\cite{MR3350338}. Every strongly irregular variety has an equational base consisting of a set of regular identities and a single strongly irregular identity. Note that most ``interesting'' varieties are strongly irregular---most Maltsev conditions involve a strongly irregular identity. For example, the Maltsev condition for congruence-permutability has a ternary term $q(x, y, z)$ satisfying $q(x, y, y) \approx x$, which is a strongly irregular identity. By contrast, the variety of semilattices is regular.

The proof of the next result will refer to \textit{strongly irregular varieties} and \textit{P\l onka sums}. Our definitions of these concepts are the same as those found in~\cite{MR3350338}, but we include them here for completeness. A variety is called \defn{strongly irregular} if it satisfies an identity $t(x, y) \approx x$ for some binary term $t$ in which both $x$ and $y$ appear. A similarity type of algebras is said to be \defn{plural} if it contains no nullary operation symbols, and at least one non-unary operation symbol.

\newcommand{\cupdot}{\mathbin{\mathaccent\cdot{\bigcup}}}

Let $\langle S, \vee \rangle$ be a semilattice, $\{\alg A_s \mid s \in S\}$ a collection of algebras of plural type $\rho : \mathcal F \to \mathbb N$, and $\{\varphi_{s,t} : \alg A_s \to \alg A_t \mid s \leq_\vee t\}$ a collection of homomorphisms satisfying $\varphi_{s,s} = 1_{A_s}$ and $\varphi_{t,u} \circ \varphi_{s,t} = \varphi_{s,u}$. The \defn{P\l onka sum} of the system $\langle \alg A_s : s \in S; \varphi_{s,t} : s \leq_\vee t\rangle$ is the algebra $\alg A$ of type $\rho$ with universe $A = \cupdot \{\alg A_s | s \in S\}$ and for $f \in \mathcal F$ a basic $n$-ary operation,
\[
f^{\alg A} (x_1, x_2, \dots, x_n) = f^{\alg A_s} (\varphi_{s_1,s} (x_1 ), \varphi_{s_2,s} (x_2), \dots, \varphi_{s_n,s} (x_n))
\]
in which $s = s_1 \vee s_2 \vee \cdots \vee s_n$ and $x_i \in A_{s_i}$ for $1 \le i \le n$.

\begin{cor}\label{cor:edge-prod}
Let $\bA_0$, $\bA_1$,\dots,$\bA_{n-1}$ be finite {\cib}s satisfying the equivalent conditions in Theorem~\ref{thm:omit5-cib}. Then both $\bA_0\times \cdots\times \bA_{n-1}$ and $\bA_0\times\cdots\times \bA_{n-1} \times \slt$ are tractable.
\end{cor}

\begin{proof}
By the theorem, $\var{V}(\bA_0)$ and $\var{V}(\bA_1)$ each have an edge term. By Theorem~\ref{thm:robust}, $\sansH\bigl(\var{V}(\bA_0)\circ \var{V}(\bA_1)\bigr)$ has an edge term as well. Since $\bA_0\times \bA_1$ lies in this variety, it possesses that same edge term. Iterating this process, the algebra $\bB=\bA_0 \times \cdots \times \bA_{n-1}$ has an edge term, hence is tractable by Theorem~\ref{thm:edge-tractable}.

Let $t(x_1,\dots,x_{k+1})$ be an edge term for $\bB$. Then according to the first identity in Definition~\ref{defn:edge-term}, $\var{V}(\bB)$ is a strongly irregular variety. Consequently, the algebra $\bB\times \slt$ is a P\l onka sum of two copies of \bB. As a P\l onka sum of tractable algebras, Theorem~4.1 of~\cite{MR3350338} implies that $\bB\times \slt$ is tractable as well.
\end{proof}

It follows from the corollary and Theorem~\ref{thm:HSP-tract} that every finite member of $\var{V}(\bA_0\times\bA_1\times\cdots\times \bA_{n-1}\times \slt)$ is at least locally tractable.

As a counterpoint to Theorem~\ref{thm:omit5-cib}, we mention the following. The proof requires some basic tame congruence theory, for which we refer the reader to~\cite[Thm~9.10]{HM:1988}.

\begin{thm}\label{thm:no-abelians}
Let \bA\ be a finite \cib\ and suppose that $\sansH\sansS(\bA)$ contains no nontrivial abelian algebras. Then $\var{V}(\bA)$ is \sdm. Consequently, \bA\ is a tractable algebra.
\end{thm}

\begin{proof}
Since every \cib\ has a Taylor term, and since $\sansH\sansS(\bA)$ has no abelian algebras, $\sansS(\bA)$ omits types $\{1,2\}$. Then by~\cite[Cor~2.2]{Freese:2009}, $\var{V}(\bA)$ omits types $\{1,2\}$. But then, by~\cite[Thm~9.10]{HM:1988}, $\var{V}(\bA)$ is \sdm.
\end{proof}

In a congruence-permutable variety, an algebra is abelian if and only if it is polynomially equivalent to a faithful, unital module over a ring. See~\cite[Thm~7.35]{MR2839398} for a full discussion. For example, consider the ring $\mathbb Z_3$ as a module over itself. Define the binary polynomial $x\cdot y = 2x+2y$ on this module. The table for this operation is given at the left of Figure~\ref{fig:cib3}. Conversely, we can retrieve the module operations from `$\cdot$' by $x+y = 0\cdot(x\cdot y)$ and $2x=0\cdot x$. Consequently, the commutative idempotent binar $\Sq 3$ is abelian.

In light of Theorem~\ref{thm:type2cp}, every finite, abelian \cib\ will be of this form. We make the following observation.

\begin{prop}\label{prop:cib_ab_odd}
A finite, abelian \cib\ has odd order.
\end{prop}

\begin{proof}
Let \bA\ be a finite, abelian \cib. By Theorem~\ref{thm:type2cp} (since every \cib is Taylor), \bA\ is  polynomially equivalent to a faithful, unital module $M$ over a ring $R$.  The basic operation of $\bA$ must be a polynomial of $M$. Thus there are $r,s\in R$ and $b\in M$ such that $x\cdot y = rx + sy +b$. Let 0 denote the zero element of $M$. Then $0=0\cdot 0 = r0 +s0 +b =b$. The condition $x\cdot y = y \cdot x$ implies  $rx+sy = ry+sx$, so (by taking $y=0$ and since $M$ is faithful), $r=s$. Finally by idempotence, $2r=1$.

Now, if $A$, hence $M$, has even cardinality, there is an element $a\in A$, $a\neq 0$, of additive order~2. Thus $a= a\cdot a = 2ra = 0$
which is a contradiction.
\end{proof}

To demonstrate the utility of the techniques developed in this paper, we now show that every \cib of cardinality at most~4 is tractable. It is already known that the algebraic dichotomy conjecture holds for every idempotent algebra of cardinality at most~3~\cite{MR2212000, MR521057}. However, our arguments are relatively short and help us better understand {\cib}s and their distinctive properties.

For the remainder of this section, let \bA\ be a \cib with universe $\nn$. We consider the various possibilities for $n$ and \bA\ (up to isomorphism), and in each case show that the associated \csp is tractable.  In order to make it clear that our inventory is complete, a few of the arguments and computations are postponed until the end.

\casespec{$\bm{n=2}$} There is a unique 2-element \cib, namely \slt. Idempotence determines the diagonal entries in the table, and by commutativity, the remaining two entries must be equal. The choices $0$ and $1$ for the off-diagonal yield a meet-semilattice and a join-semilattice respectively. As we remarked above, every finite semilattice is tractable, and, in fact, $\Sl$, the variety of semilattices, is \sdm.

\casespec{$\bm{n=3}$, \bA\ not simple} Let $\theta \in \Con \bA$ be nontrivial (i.e., $0_A < \theta < 1_A$). By idempotence, every congruence class is a subalgebra. Evidently, $\bA/\theta$ has cardinality~2, one $\theta$ class has 2 elements, and the other has~1. From the uniqueness of the 2-element algebra, $\bA/\theta$ and each of the $\theta$-classes are semilattices. Consequently, $\bA \in \Sl \circ \Sl$. By Theorem~\ref{thm:robust}, $\Sl\circ \Sl$ is \sdm, so by Theorem~\ref{thm:sdm-tractable}, \bA\ is tractable.

\begin{figure}
\centering
\begin{tabular}{ccc}
  \begin{tabular}{c|ccc}
   $\cdot$&0&1&2\\
  \hline
  0&0&2&1\\
  1&2&1&0\\
  2&1&0&2
  \end{tabular} \hspace*{2em} &
  \begin{tabular}{c|ccc}
  $\cdot$&0&1&2\\
  \hline
  0&0&0&1\\
  1&0&1&2\\
  2&1&2&2
  \end{tabular} \hspace*{2em}&
  \begin{tabular}{c|ccc}
  $\cdot$&0&1&2\\
  \hline
  0&0&0&2\\
  1&0&1&1\\
  2&2&1&2
  \end{tabular} \\
  \vrule height 12pt width 0pt $\Sq 3$ & $\bT_1$ & $\bT_2$
 \end{tabular}
 \caption{The simple {\cib}s of cardinality 3}\label{fig:cib3}
 \end{figure}

\casespec{$\bm{n=3}$, \bA\ simple} If \bA\ has no proper nontrivial subalgebras, then in \bA, $x\neq y \implies x\cdot y \notin \{x,y\}$. It follows that \bA\ must be the 3-element Steiner quasigroup, $\Sq 3$, of Figure~\ref{fig:cib3}. By Corollary~\ref{cor:edge-prod}, \bA\ is tractable.

Finally, if \bA\ has a proper nontrivial subalgebra, that subalgebra must be isomorphic to \slt. Thus no nontrivial subalgebra of \bA\ is abelian. By Theorem~\ref{thm:no-abelians}, \bA\ generates an \sdm variety, so is tractable by Theorem~\ref{thm:sdm-tractable}. For future reference, we remark that there are two non-isomorphic algebras meeting the description in this paragraph, $\bT_1$ and $\bT_2$.  Their tables are given in Figure~\ref{fig:cib3}.

\casespec{$\bm{n=4}$, \bA\ not simple} Let $\theta$ be a maximal congruence of \bA. Then $\bA/\theta$ is simple, so $\bA/\theta$ is isomorphic to one of $\bT_1$, $\bT_2$, $\Sq 3$, or $\slt$. If $\bA/\theta \cong \bT_i$, then the $\theta$-classes must have sizes $1$, $1$, and $2$. Consequently, each $\theta$-class is a semilattice, so $\bA\in \Sl \circ \var{V}(\bT_i)$ which is \sdm by the computations above and Theorem~\ref{thm:robust}. Thus $\bA$ is tractable.

The next case to consider is $\bA/\theta \cong \Sq 3$. Without loss of generality, assume that~$\theta$ partitions the universe as $|01|2|3|$, and further that $0\cdot 1=0$. From these data we deduce that the operation table for \bA\ must be
\begin{equation*}
\begin{tabular}{c|cccc}
      $\cdot $ & 0 & 1 & 2 & 3\\
      \hline
      0 & 0 & 0 & 3& 2\\
      1 & 0 & 1 & 3& 2\\
      2 & 3 & 3 & 2 & $a$\\
      3 & 2 & 2 & $a$ & 3
 \end{tabular}\qquad \text{with $a\in \{0,1\}$.}
\end{equation*}
If $a=0$, then the only way to obtain 1 as a product is $1\cdot1$. In that case there is a homomorphism onto \slt with kernel $\psi=|023|1|$. We have $\theta \cap \psi = 0_A$, so $\bA$ is a subdirect product of $\Sq 3 \times \slt$. Therefore, by Corollary~\ref{cor:edge-prod}, \bA\ is tractable.  More problematic is the case $a=1$, but we already established tractability of that algebra in Example~\ref{ex:quot-inst}.

Finally, suppose that $\bA/\theta \cong \slt$. The possible sizes of the $\theta$-classes are $3$, $1$ or $2$, $2$.  If neither class is isomorphic to \Sq 3 then the classes both lie in an \sdm variety, so by Theorem~\ref{thm:robust}, \bA\ too lies in an \sdm variety and is thus tractable. If one of the classes is isomorphic to \Sq 3 (there are 7 such algebras), then \bA\ is an \ec algebra (the term $t(x,y)=y\cdot(x\cdot y)$ is Maltsev, hence is an edge term) and we can apply Theorem~\ref{thm:least-block} to establish that \bA\ is tractable.

\casespec{$\bm{n=4}$, \bA\ simple} By Theorem~\ref{thm:no-abelians}, if $\sansH\sansS(\bA)$ contains no nontrivial abelian algebra, then \bA\ is tractable. Thus, we may as well assume that $\sansH\sansS(\bA)$ contains an abelian algebra. By Proposition~\ref{prop:cib_ab_odd}, \bA\ itself is nonabelian. Consequently (since \bA\ is simple) \bA\ must have a subalgebra isomorphic to \Sq 3. Without loss of generality, let us assume that $\{1,2,3\}$ forms this subalgebra. Similarly, by Corollary~\ref{cor:edge-prod}, we can assume that $\sansH\sansS(\bA)$ contains a copy of \slt. Examining the 2- and 3-element \cibs, it is easy to see that if $\slt\in \sansH\sansS(\bA)$, then already $\slt\in\sansS(\bA)$. By the symmetry of \Sq 3, and the fact that it contains no semilattice, we can assume that $\{0,1\}$ forms the semilattice. Thus, the table for \bA\ must have one of the following two forms.
\begin{equation*}
\begin{tabular}{c|cccc}
$\cdot$&0&1&2&3\\
\hline
0&0&0&$u_2$&$u_3$\\
1&0&1&3&2\\
2&$u_2$&3&2&1\\
3&$u_3$&2&1&3
\end{tabular}
\qquad
\begin{tabular}{c|cccc}
$\cdot$&0&1&2&3\\
\hline
0&0&1&$v_2$&$v_3$\\
1&1&1&3&2\\
2&$v_2$&3&2&1\\
3&$v_3$&2&1&3
\end{tabular}
\end{equation*}
By checking the 32 possible tables, either using the Universal Algebra
Calculator~\cite{UAcalc} (with algebra file~\cite[\href{https://github.com/UACalc/AlgebraFiles/tree/master/Bergman}{\small
  \texttt{CIB4-SL-nonSDmeet.ua}}]{william_demeo_2016_53936}),
or directly using Freese's algorithm~\cite{Freese2008}, one determines that
there are 7 pairwise nonisomorphic algebras of one of these two forms.
Interestingly, every simple algebra of the second form is isomorphic to one of
the first form. Thus we will use the first form for all the candidates. These 7
algebras are indicated in Figure~\ref{fig:simple7}.

\begin{figure}
\centering
\begin{tabular}{l|cc|l}
&$u_2$&$u_3$&Proper Nontrivial Subalgebras\\
\hline
$\bA_0$&0&1&$\{0,1\},\, \{0,2\},\, \{1,2,3\}$\\
$\bA_1$&1&1&$\{0,1\},\, \{1,2,3\}$\\
$\bA_2$&1&2&$\{0,1\},\, \{1,2,3\}$\\
$\bA_3$&0&3&$\{0,1\},\, \{0,2\},\, \{0,3\},\, \{1,2,3\}$\\
$\bA_4$&1&3&$\{0,1\},\, \{0,3\},\, \{1,2,3\}$\\
$\bA_5$&2&2&$\{0,1\},\, \{0,2\},\, \{0,3\}\, \{1,2,3\}$\\
$\bA_6$&2&3&$\{0,1\},\, \{0,2\},\, \{0,3\}\, \{1,2,3\}$
\end{tabular}
 \caption{The 7 simple \cibs of size 4.}\label{fig:simple7}
 \end{figure}

Our intent is to apply the rectangularity theorem to establish the tractability of these~7 algebras. For this, we need to determine the minimal absorbing subuniverses of each subalgebra of \bA. First observe that if $\{a,b\}$ forms a copy of \slt\ with $a\cdot b =a$, then $\{a\} \absorbing_t \{a,b\}$, with $t(x,y)=x\cdot y$, while $\{b\}$ is not absorbing since $a$ is a sink (Lemma~\ref{lem:sink}). Second, by Lemma~\ref{lem:abelian-AF}, \Sq 3 is absorption-free, which is to say, it is its own minimal absorbing subalgebra.

 Now let $B$ be a proper minimal absorbing subalgebras of $\bA_i$, for $i<7$. Then either $B\cap
 \{1,2,3\}$ is absorbing in $\{1,2,3\}$ or it is empty (Lemma~\ref{lem:restriction}). As $\{1,2,3\}\cong
 \Sq 3$ is absorption free and $B$ is a proper subalgebra of \bA, we must have
 $B=\{1,2,3\}$ or $B=\{0\}$. However the first of these is impossible because $\{1,2,3\}
 \cap \{0,1\} = \{1\}$, which is not absorbing in $\{0,1\}$ since $0$ is a sink. Thus the
 only possible proper absorbing subalgebra of \bA\ is $\{0\}$.

 For $i\leq 2$, $\{0\}$ is indeed absorbing in $\bA_i$. We believe that the shortest
 binary absorbing term in each case is as follows.
 \begin{equation}\label{eq:abs-term}
   \begin{aligned}
     \bA_0 \quad& \bigl(x(xy)\bigr)\bigl(y(xy)\bigr)\\
     \bA_1 \quad& \bigl(x(x(xy))\bigr)\bigl(y(y(xy))\bigr)\\
     \bA_2 \quad& ((x(x(xy)))((y(xy))(x(x(xy)))))((y(x(xy)))(x(x(x(xy))))).
   \end{aligned}
 \end{equation}
 However, observe that if $i>2$ then
 according to Figure~\ref{fig:simple7} either $u_2=2$ or $u_3=3$. If $u_2=2$ then $2$ is a
 sink for the subuniverse $\{0,2\}$. In that case $\{0\}= B\cap \{0,2\}$ is not absorbing
 in $\{0,2\}$ contradicting the fact that $B$ is absorbing in \bA. A similar argument
 works for $u_3=3$ and $\{0,3\}$. Thus $\bA_i$ is absorption free for $i>2$.

 We summarize these computations in the following table.
\begin{center}
  \begin{tabular}{lc}
  Algebra&Minimal Absorbing Subalgebra\\\toprule
  $\slt=\{0,1\}$ & $\{0\}$ \\
  $\{1,2,3\}$ & $\{1,2,3\}$ \\
   $\bA_0, \bA_1, \bA_2$ & $\{0\}$\\
  $\bA_3,\dots,\bA_6$ & $A$
 \end{tabular}
\end{center}
The crucial point here is that every member of $\sansS(\bA_i)$ has a unique minimal
absorbing subalgebra.

We make one more observation. Let $i<7$ and let $h$ be an automorphism of $\bA_i$. Then
$h(0)=0$. This holds because $\{1,2,3\}$ must be preserved by $h$, as it is the unique
3-element subalgebra. Of course since \slt has no nontrivial automorphisms, $h(0)=0$ also
holds in that algebra.

We now proceed to argue that each algebra from Figure~\ref{fig:simple7} is
tractable. Following Definition~\ref{def:csp}, let $\sI=\langle \sV, \sA, \sS, \sR\rangle
$ be an instance of $\CSP(\bA_i)$, for some $i<7$. We shall simply write $\bA$ instead of
$\bA_i$.  We have $\sA=(\bC_0, \bC_1,\dots,\bC_{n-1})$, in which each $\bC_i$ is a subalgebra of $\bA$, and
$\sR=(R_0,\dots,R_{J-1})$ is a list of subdirect products of the various $\bC_k$'s. We
must show that there is an algorithm that determines, in polynomial time, whether this
instance has a solution. For this we will use Corollary~\ref{cor:RT-cor-gen}.

For every $k<n$ let $\bB_k$ be the minimal absorbing subuniverse of $\bC_k$. As mentioned, $\bB_k$ is unique. Let $\alpha=\left\{\,k<n : \text{$\bC_k$ is abelian}\,\right\}$ and set $\alpha'=\nn-\alpha$. By reordering the factors, we may assume that $\alpha=\{0,1,\dots,p-1\}$ for some $p<n$.

Fix $j<J$. There is an equivalence relation on $\alpha'$ defined by $k\sim_j \ell$ iff $\etaR^j_k = \etaR^j_\ell$. Call the least index in each equivalence class the \emph{leader.}  The remaining elements are the \emph{followers.} When $k\sim_j \ell$, subdirectness of $R_j$ implies that $\bC_k \cong \bR_j/\etaR^j_k = \bR_j/\etaR^j_\ell \cong \bC_\ell$. Thus, there is an isomorphism $h^j_{k,\ell}\colon \bC_k \to \bC_\ell$ such that $\Proj_{k,\ell}(R_j) =   \bigl\{\,(x,h^j_{k,\ell}(x)) : x\in C_k\,\bigr\}$.
These maps are consistent in the sense that if $k\sim \ell \sim m$ then $h_{k,m} = h_{\ell,m}\circ h_{k,\ell}$. Observe also that $h_{k,\ell}(B_k) = B_\ell$, since the minimal absorbing subalgebra is unique in each factor. Moreover, from our comments about automorphisms, and
since $k,\ell \in \alpha'$, we have $h_{k,\ell}(0)=0$. Finally, for $k\sim_j \ell$, define $H^j_{k,\ell} = \Proj_{k,\ell}(R_j)$, so $(0,0)\in H^j_{k,\ell}$ whenever $k\sim_j \ell$.

For each $j<J$, let the relation $\bar R_j$ be obtained from $R_j$ as follows. First, for
each follower $\ell$, remove the $\ell$-th variable from the scope of $R_j$. Then, define
$\bar R_j = R_j \times \prod_{\ell \notin \im\sigma_j} C_\ell$. Let
$\bar\sR=(\bar R_j : j<J)$ and $\bar\sH=(H^j_{k,\ell} : j<J, k\sim_j \ell)$.  Lastly, let
$\bar \sI$ denote the modification of $\sI$ with $\bar\sR \cup \bar \sH$ replacing $\sR$
and every scope equal to $\sV$. Then the solution set of $\bar \sI$ is identical to the
solution set of the original instance $\sI$. (To be more precise, the relations in
$\bar\sR \cup \bar\sH$ are pp-definable (see\cite[Sec.~3.1]{BartoKrokhinWillard2017})
from $\sR$ and conversely.)

Consider the partial instance $\bar\sI_{\pp}$. This is an instance of the problem
$\CSP(\Sq3)$, which is tractable. If the partial instance has no solution, then neither
does $\sI$, and we halt. So suppose that $\bar\sI_{\pp}$ has a solution, $\bx$. (We do not
actually need to know $\bx$, just that it exists.) Since $\Sq3$ is absorption-free, $\bx
\in \prod_\alpha C_k = \prod_\alpha B_k$.

We now wish to apply Corollary~\ref{cor:RT-cor-gen} to show that
\begin{equation}
\{\bx\}\times \prod_{k\in \alpha'} B_k \subseteq \bigcap \bar\sR.\label{eq:3}
\end{equation}
The first 3 conditions follow from the
definition of $\alpha$ and the construction of the $\bar R_j$'s. The fifth condition holds
because of the solution $\bx$ to the partial instance. We must address the fourth requirement.

Fix $j<J$. We shall show that $\bar R_j$ intersects $\prod B_k$. We have two cases to
consider. First, suppose that $i<3$. In that case the minimal absorbing subalgebra of
$\bA$ is $\{0\}$, with corresponding absorbing term given in~\eqref{eq:abs-term}. Think of
that term as a nonassociative product. For every $\ell \in \alpha$,
$B_\ell = \{1,2,3\}= C_\ell$, while for $\ell \in \alpha'$, $B_\ell = \{0\}$. Since
$\bar R_j$ is subdirect, for every $\ell \in \alpha'$ there is a tuple $\br^\ell \in R_j$
with $r^\ell_\ell = 0$. Let $\br$ be the ``product'' (under the absorbing term) of all the
$\br^\ell$'s. (For definiteness, associate to the right.) Then, since $\{0\}$ is
absorbing, $\br_\ell = 0\in B_\ell$ for every $\ell \in \alpha'$, and for
$\ell \in \alpha$, $\br_\ell \in C_\ell = B_\ell$. Thus the condition is satisfied.

The argument when $i\geq 3$ is not very different. In this case, let
$\beta = \{\ell <n : \bC_\ell = \slt \}$. Note that if $\ell \notin \beta$ then
$B_\ell = C_\ell$.  As before, for $\ell \in \beta$ let $\br^\ell$ be a tuple with
$\br^\ell_\ell = 0$, and let $\br$ be the product. (This time simply use the basic binary
operation to compute the product.  It is, after all, an absorbing term for \slt.) Then for
$\ell \in \beta$, $r_\ell = 0 \in B_\ell$, while, for $\ell \notin \beta$,
$r_\ell \in C_\ell = B_\ell$.

Thus inclusion~\eqref{eq:3} holds. Let $\mathbf 0$ denote a tuple of $0$'s indexed by
$\alpha'$, and set $\vy$ to be the concatenation of $\bx$ with $\mathbf 0$. Then $\vy
\in \bigcap \bar\sR$ by~\eqref{eq:3}. On the other hand, for any $j<J$ and $k\sim_j \ell$,
we have $(y_k, y_\ell) = (0,0) \in H^j_{k,\ell}$. Thus $\vy$ is a solution to $\bar\sI$.

\section{Acknowledgements}
The authors wish to thank the anonymous referee for identifying errors, suggesting corrections, and providing many other valuable suggestions for improvements.  The authors also gratefully acknowledge the support of National Science Foundation grant no.~1500218.

\ifthenelse{\boolean{draft}}{\newpage}{}


\appendix

\ifthenelse{\boolean{arxiv}}{   
\section{Nonrectangularity of abelian algebras}
The example in this section reveals why the rectangularity theorem cannot be generalized to products of abelian algebras. Before examining the example, we state a lemma that will be useful when discussing the example. The proof is straightforward.
\begin{lem}\label{lem:abs2}
Let $\bA_1, \dots, \bA_n$ be finite simple algebras in a Taylor variety and suppose
\begin{itemize}
\item each $\bA_i$ is absorption-free,
\item $\bR \sdp \bA_1 \times \cdots \times \bA_n$,
\item $\etaR_i \neq \etaR_j$ for all $i\neq j$, and
\item $\mu \subseteq \nn $ is minimal among the sets in $\{\sigma \subseteq \nn  \mid \bigwedge_{\sigma} \etaR_i = 0_R\}$.
\end{itemize}
Then, $|\mu|>1$ and $\Proj_{\tau} \bR = \Pi_\tau \bA_i$ for every set $\tau \subseteq \nn $ with $1< |\tau|\leq |\mu|$.
\end{lem}

\begin{exa}
  Let $\bA$ be the algebra $\<\{0,1\}, \{f\}\>$, with universe $\{0,1\}$, and a single basic operation given by the ternary function $f(x,y,z) = x+y+z$ where addition is modulo 2.  This algebra is clearly simple and has two proper subuniverses, $\{0\}$ and $\{1\}$, neither of which is absorbing, so $\bA$ is absorption-free.  Let $\bR \sdp \bA \times \bA \times \bA$ be the subdirect power of $\bA$ with universe
  \[ R = \{(x,y,z)\in A^3\mid x+y+z=0 \text{ (mod $2$)}\} = \{(0,0,0), (1,1,0), (1,0,1), (0,1,1) \}.\]
  It's convenient to give short names to the elements of $R$.  Let us use the integers that they (as binary tuples) represent.
  That is, $0 = (0,0,0)$, $3 = (1,1,0)$,  $5 = (1,0,1)$,  and $6 = (0,1,1)$,  so $R = \{0, 3, 5, 6\}$. Let $\etaR_i = \ker(\bR \onto \bA_i)$, and identify each congruence with the associated partition of the set $R$. Then, $\etaR_1 = |0,6|3,5|; \,  \etaR_2 = |0,5|3,6|; \,  \etaR_3 = |0,3|5,6|$,
  so $\etaR_i \meet \etaR_j = 0_R$ and $\etaR_i \join \etaR_j = 1_R$. In fact, the three projection kernels are the only nontrivial congruence relations of $\bR$.

  \tikzstyle{lat} = [circle,draw,inner sep=0.8pt]
  \begin{center}
  \begin{tikzpicture}[scale=1.2]
    \node[lat] (bot) at (0,0) {};
    \node[lat] (top) at (0,2) {};
    \node[lat] (a) at (-1,1) {};
    \node[lat] (b) at (0,1) {};
    \node[lat] (c) at (1,1) {};
    \draw[semithick] (bot) -- (a) -- (top) -- (b) -- (bot) -- (c) -- (top);
    \draw (top) node [right]{$1_R$};
    \draw (bot) node [right]{$0_R$};
    \draw (a) node [left]{$\etaR_1$};
    \draw (b) node [right]{$\etaR_2$};
    \draw (c) node [right]{$\etaR_3$};
    \draw (-2.5,1) node {$\Con(\bR) = $};
  \end{tikzpicture}
  \end{center}

  Each projection of $\bR$ onto 2 coordinates of $\bA^3$ is ``linked;''  these binary projections are all readily seen to be
  $\{(0,0), (0,1), (1,0), (1,1)\} =  A \times A$ in this case.  Lemma~\ref{lem:abs2} tells us that this must be so.  For the set
  $\{S \subseteq \nn  \mid \bigwedge_{i\in S} \etaR_i = 0_R\}$ that appears in the lemma is, in this example,
  $\sS = \{\{1,2\}, \{1,3\}, \{2,3\}, \{1,2,3\}\}$, and the projection of $R$ onto the coordinates in each minimal  set in $\sS$ must
  equal $A \times A$ by the lemma.  Now let $\bS \sdp \bA \times \bA \times \bA$ be the subdirect power of $\bA$
  with universe
  \begin{align*}
    S &= \{(x,y,z)\in A^3\mid x+y+z=1 \text{ (mod $2$)}\}\\
    &= \{(1,0,0), (0,1,0), (0,0,1), (1,1,1) \}\\
    &= \{1, 2, 4, 7\}.
  \end{align*}
  All of the facts observed above about $\bR$ are also true of $\bS$.  In particular, $\Proj_{\{i,j\}} S = A_i \times A_j$ for
  each pair $i\neq j$ in $\{1,2,3\}$. If both $\bR$ and $\bS$ belong to the set of relations (or ``constraints'') of a single $\CSP(\sansS(\bA))$ instance, then the instance has no solution since $R \cap S = \emptyset$, despite the fact that $\Proj_{\{i,j\}}\bR = \Proj_{\{i,j\}}\bS$  for each pair $i\neq j$ in $\{1,2,3\}$. To put it another way, $\bR$ and $\bS$ witness a failure of the \emph{2-intersection property}. The ``potato diagram'' of $\bA^3$ below depicts the elements of $R$ and $S$ as colored lines (with colors chosen to emphasize equality of the projections onto $\{1,2\}$).

  \begin{center}
  \begin{figure}
    \caption{potatoes}
  \begin{tikzpicture}[scale=1]
    \draw (0,.5) ellipse (4mm and 12mm);
    \draw (2,.5) ellipse (4mm and 12mm);
    \draw (4,.5) ellipse (4mm and 12mm);
    \node[lat] (00) at (0,0) {};
    \node[lat] (01) at (0,1) {};
    \node[lat] (10) at (2,0) {};
    \node[lat] (11) at (2,1) {};
    \node[lat] (20) at (4,0) {};
    \node[lat] (21) at (4,1) {};
    \draw (00) node [below]{$0$};
    \draw (01) node [above]{$1$};
    \draw (10) node [below]{$0$};
    \draw (11) node [above]{$1$};
    \draw (20) node [below]{$0$};
    \draw (21) node [above]{$1$};
    \draw[red,thick] (00) -- (10) -- (20);
    \draw[blue,thick] (00) -- (11) -- (21);
    \draw[green,thick] (01) -- (10) -- (21);
    \draw[yellow,thick] (01) -- (11) -- (20);
    \draw (7,0.5) node {(lines are elements of $R$)};
  \end{tikzpicture}
  \vskip2mm
  \begin{tikzpicture}[scale=1]
    \draw (0,.5) ellipse (4mm and 12mm);
    \draw (2,.5) ellipse (4mm and 12mm);
    \draw (4,.5) ellipse (4mm and 12mm);
    \node[lat] (00) at (0,0) {};
    \node[lat] (01) at (0,1) {};
    \node[lat] (10) at (2,0) {};
    \node[lat] (11) at (2,1) {};
    \node[lat] (20) at (4,0) {};
    \node[lat] (21) at (4,1) {};
    \draw (00) node [below]{$0$};
    \draw (01) node [above]{$1$};
    \draw (10) node [below]{$0$};
    \draw (11) node [above]{$1$};
    \draw (20) node [below]{$0$};
    \draw (21) node [above]{$1$};
    \draw[yellow,thick] (01) -- (11) -- (21);
    \draw[red,thick] (00) -- (10) -- (21);
    \draw[green,thick] (01) -- (10) -- (20);
    \draw[blue,thick] (00) -- (11) -- (20);
    \draw (7,0.5) node {(lines are elements of $S$)};
  \end{tikzpicture}
  \end{figure}
  \end{center}

\end{exa}
}{}

\section{Miscellaneous Proofs}\label{sec:proofs-elem-facts}
\subsection{Proof of Lemma~\ref{lem:fact1}}\label{sec:fact1}
To simplify notation, let $ABB\cdots B := A\times B \times B \times \cdots \times B$ (the usual Cartesian product) and let $t=f\star g$. We  handle the $j=1$ case of the definition of absorption. (The general case is no harder, but the notation becomes tedious.)
That is, we prove
\[
t[ABB\cdots B] = f\bigl[g[ABB\cdots B] \times g[BB\cdots B] \times \cdots \times g[BB \cdots B]\bigr]\sseq B.
\]
If $f$ is the absorbing term, then since $B$ is a subalgebra, we have $g[BB\cdots B] \subseteq B$. Hence, $t[ABB\cdots B] \subseteq f[ABB\cdots B] \sseq B$. On the other hand, if $g$ is the absorbing term, then $t[ABB \cdots B] \sseq f[BB \cdots B] \sseq B$. \hfill \qedsymbol%

\subsection{Proof of Corollary~\ref{cor:fact2}}\label{sec:fact2}
Since $\bB_0 \absorbing_f \bA_0$, it follows that $\bB_0\times \bA_1 \absorbing_f \bA_0\times \bA_1$. Similarly $\bA_0 \times \bB_1 \absorbing_g \bA_0\times \bA_1$. Hence, $\bB_0\times \bB_1=(\bB_0\times \bA_1) \cap (\bA_0\times \bB_1)$ is absorbing in $\bA_0\times \bA_1$ with respect to $f\star g$, by Lemma~\ref{lem:bk-prop-2-4}. \hfill \qedsymbol%

\subsection{Proof of Lemma~\ref{lem:min-abs-prod}}\label{sec:proof-cor-min-abs-prod}
In case $n=2$, the fact that $\bB \absorbing_s \bA$ follows directly from Corollary~\ref{cor:fact2}. We first extend this result to \emph{minimal} absorbing subalgebras, still for $n=2$, and then an easy induction argument will complete the proof for arbitrary finite $n$.

  Assume $\bB_0 \minabsorbing_{t_0} \bA_0$ and $\bB_1 \minabsorbing_{t_1} \bA_1$. Then $\bB \absorbing_s \bA$ (Cor.~\ref{cor:fact2}), where $\bB := \bB_0\times \bB_1$, $\bA := \bA_0\times \bA_1$, and $s = t_0\star t_1$. To show $\bB_0\times \bB_1$ is minimal absorbing, let $\bS$ be a proper subalgebra of $\bB_0\times \bB_1$. By transitivity of absorption, it suffices to prove that $\bS$ is not absorbing in $\bB_0\times \bB_1$.

  Let $t$ be a term of arity $q$. For $b \in B_1$, the set $S^{-1}b := \{b_0 \in B_0 \mid (b_0, b) \in S\}$ is a subuniverse of $\bB_0$.
Since $\bS$ is a proper subalgebra, there exists $b^* \in B_1$ such that $S^{-1}b^*$ is not all of $B_0$. Therefore, $S^{-1}b^*$ is a proper subuniverse of $\bB_0$, so $S^{-1}b^*$ is not absorbing in $\bA_0$ (by minimality of $\bB_0$). Consequently,
  \begin{align*}
    \exists x_i \in S^{-1}b^*, &\quad \exists b\in B_0,  \quad \exists j< q, \quad \exists  b'\notin S^{-1}b^* \quad \text{ such that }  \\
    \;& t^{\bA_0}(x_0, x_1, \dots, x_{j-1}, b, x_{j+1}, \dots, x_{q-1}) = b'.
  \end{align*}
  Thus, $t^{\bA_0\times \bA_1}((x_0,b^*), \dots, (x_{j-1}, b^*), (b, b^*), (x_{j+1}, b^*), \dots, (x_{q-1}, b^*)) = (b', b^*) \notin S$ (since $b' \notin S^{-1}b^*$). Finally, because $(x_i, b^*) \in S$ for all $i$, and since $t$ was an arbitrary term, $\bS$ is not absorbing in $\bB_0\times \bB_1$.

  Now fix $n>2$ and assume the result holds when there are at most $n-1$ factors. Let $\bB' := \bB_0 \times \cdots \times \bB_{n-2}$ and
  $\bA' := \bA_0\times \cdots \times \bA_{n-2}$. By the induction hypothesis, $\bB' \minabsorbing_{s'} \bA'$, and since $\bB_{n-1} \minabsorbing_{t_{n-1}} \bA_{n-1}$ we have (by the $n=2$ case) $\bB' \times \bB_{n-1} \minabsorbing_s \bA'\times \bA_{n-1}$, where $s = s' \star t_{n-1}$.
\hfill \qedsymbol

\subsection{Proof of Corollary~\ref{cor:gen-abs1}}\label{sec:gen-abs1}
Let $\bA = \myprod_i \bA_i$ and $\bB = \myprod_i \bB_i$. By Lemma~\ref{lem:min-abs-prod}, $\bB \absorbing_t \bA$, so the result follows from Lemma~\ref{lem:restriction} if we put $C = R$. \hfill \qedsymbol%

This is proved Section~\ref{sec:proof-lemma-sdp-general}.

\subsection{Proof of Lemma~\ref{lem:sdp-general}}\label{sec:proof-lemma-sdp-general}

\ifthenelse{\boolean{extralong}}{
  The argument used to prove the next lemma also works in the more general case of $n$-fold products. Although the arguments are almost identical, and although we present the general proof below, we begin with a proof of the $2$-fold case since it is so much easier to read.

  \noindent {\bf Lemma.} Let $\bA_1$ and $\bA_2$ be finite algebras in an idempotent variety, and suppose $\bB_i \minabsorbing \bA_i$ for $i=1,2$. Let $\bR \sdp \bA_1 \times \bA_2$, and let $R' = R \cap (B_1 \times B_2)$.  If $R'\neq \emptyset$, then $\bR' \sdp \bB_1 \times \bB_2$.

  \begin{proof}
  If $R'\neq \emptyset$, then $S_1 := \Proj_1 R'$ and $S_2 := \Proj_2 R'$ are also nonempty.   We want to show $S_i = B_i$ for $i=1,2$. By minimality of $\bB_1 \minabsorbing \bA_1$ and by transitivity of  absorption, it suffices to prove $\bS_1 \absorbing \bB_1$. Assume $\bB_1 \minabsorbing \bA_1$ with respect to $t$, say, $k = \ar(t)$. Fix $s_1, \dots, s_k \in S_1$, $b \in B_1$, and $j\leq k$.  Then $\tilde{b} := t(s_1, \dots, s_{j-1}, b, s_{j+1}, \dots, s_k) \in B_1$, and we must show $\tilde{b} \in S_1$.

  Since $\bR$ is subdirect, there exists $a\in A_2$ with $(b,a) \in R$.  Also, there exist $s_1', \dots, s_k'$ in $S_2$ such that for all $i=1, \dots, n$ we have $(s_i, s_i')\in R'$.  Since all the pairs belong to $R$, the following expression is also in $R$:
  \begin{align}  
    &t^{\bA_1\times \bA_2}((s_1,s_1'), \dots, (s_{j-1}, s_{j-1}'),  (b, a),  (s_{j+1}, s_{j+1}'), \dots, (s_k, s_k')) \nonumber \\
    &=(t^{\bA_1}(s_1, \dots, s_{j-1}, b, s_{j+1}, \dots, s_k), t^{\bA_2}(s_1', \dots, s_{j-1}', a, s_{j+1}', \dots, s_k')) = (\tilde{b}, \tilde{a}).
  \end{align}
  Since $\bB_2$ is absorbing in $\bA_2$, we see that $\tilde{a}$ belongs to $\bB_2$. Therefore, $(\tilde{b}, \tilde{a}) \in R\cap (B_1\times B_2) = R'$,  which means $\tilde{b}\in S_1$, as desired. Of course, the same argument works to prove $S_2 = B_2$.
  \end{proof}
  As mentioned, the result generalizes to $n$-fold products and the proof is nearly identical to the one above.
}{}


  If $R'\neq \emptyset$, then for $1\leq i\leq n$ the projection $S_i := \Proj_i R'$ is also nonempty.  We want to show $S_i = B_i$.
  By minimality of $\bB_i \minabsorbing \bA_i$ and by transitivity of absorption, it suffices to prove $\bS_i \absorbing \bB_i$.
  Assume $\bB_i \minabsorbing \bA_i$ with respect to $t$, say, $k = \ar(t)$. Fix $s_1, \dots, s_k \in S_i$, $b \in B_1$, and $j\leq k$.  Then
  $\tilde{b_i} := t^{\bA_i}(s_1, \dots, s_{j-1}, b, s_{j+1}, \dots, s_k) \in B_i$, and we must show $\tilde{b_i} \in S_i$.

  Since $\bR$ is subdirect, there exist $a_i \in A_i$ (for all $i\neq j$) such that
  \[\ba^\ast:= (a_1, \dots, a_{j-1}, b, a_{j+1}, \dots, a_n)\in R.\]
  Also, for each $1\leq j\leq n$ there exist $s_1{(j)}, \dots, s_k{(j)}$ in $S_j$ such that for all $1\leq \ell \leq k$ we have
  \[
  \bs_\ell := (s_\ell{(1)},\dots, s_\ell{(i-1)}, s_\ell, s_\ell{(i+1)}, \dots, s_\ell{(n)})\in R'.
  \]
  Since all these $n$-tuples belong to $R$, so does
    $t^{\bA_1\times \cdots \times \bA_n}(\bs_1, \dots, \bs_{j-1},\ba^\ast, \bs_{j+1}, \dots, \bs_{k})$, and the latter
    %
reduces to $(\tilde{b_1},\dots, \tilde{b_{n}})$.
  Since $\bB_i \absorbing \bA_i$, we have
  $(\tilde{b_1},\dots, \tilde{b_{n}})\in (B_1 \times \cdots \times B_n)$.
  Thus,
  $(\tilde{b_1},\dots, \tilde{b_{n}})\in R\cap B_1 \times \cdots \times B_n$,
  so $\tilde{b_i}\in S_i$, as desired. Of course, the same argument
  works for all $1\leq i\leq n$.
\hfill \qedsymbol

\ifthenelse{\boolean{arxiv}}{
  \subsection{Direct Proof of Corollary~\ref{cor:fry-pan}}%
  \label{sec:proof-fry-pan-cor}

  We prove the following:  Let $\bA_0, \dots, \bA_{n-1}$ be finite
  idempotent algebras in a Taylor variety and suppose $\bR \sdp \prod \bA_i$.
  For some $0< k < n-1$, assume the following:
  \begin{itemize}
  \item if $0\leq i < k$ then $\bA_i$ is abelian;
  \item if $k\leq i < n$ then $\bA_i$ has a sink $s_i \in A_i$.
  \end{itemize}
  Then
   $Z := R_{\kk}
  \times \{s_k\} \times \{s_{k+1}\} \times \cdots \times \{s_{n-1}\}  \subseteq R$,
  where $R_{\kk} = \Proj_{\kk}R$.

\begin{proof}
  Since $\bA_{\kk} := \prod_{i<k} \bA_i$ is abelian and lives in a Taylor variety,
  there exists a term $m$ such that $m^{\bA_{\kk}}$ is a \malcev term operation on
  $\bA_{\kk}$ (Theorem~\ref{thm:type2cp}).  Since we are working with idempotent terms, we can be sure
  that for each $i\in \nn$ the term operation $m^{\bA_i}$ is not
  constant (so depends on at least one of its arguments).

  Fix $\bz :=(r_0, r_1, \dots, r_{k-1}, s_k, s_{k+1}, \dots, s_{n-1}) \in Z$.
  We will show that $\bz \in R$.
  Since $\bz_{\kk}\in R_{\kk}$, there exists $\br \in R$
  whose first $k$ elements agree with those of $\bz$.
  That is, $\br_{\kk} =  (r_{0}, r_{1}, \dots, r_{k-1}) = \bz_{\kk}$.

  Now, since $\bR$ is subdirect, there exists $\bx^{(0)} \in R$ such that
  $\bx^{(0)}(k) = s_k$, the sink in $\bA_k$.
  If the term operation $m^{\bA_k}$ depends on its second or third argument,
  consider $\vy^{(0)}= m(\br, \bx^{(0)}, \bx^{(0)}) \in R$.
  (Otherwise, $m^{\bA_k}$ depends on its
  first argument, so consider $\vy^{(0)}= m(\bx^{(0)}, \bx^{(0)}, \br)$.)
  For each $0\leq i < k$ we have
  $\vy^{(0)}(i) = m^{\bA_i}(r_i, \bx^{(0)}(i), \bx^{(0)}(i)) = r_i$,
  since $m^{\bA_i}$ is \malcev.   Thus, $\vy^{(0)}_{\kk} = \bz_{\kk}$.
  At index $i = k$,  we have
  $\vy^{(0)}(k) = m^{\bA_k}(r_k, s_k, s_k) = s_k$, since $s_k$ is a sink in $\bA_k$.
  By the same argument, but starting with $\bx^{(1)} \in R$ such that
  $\bx^{(1)}(k+1) = s_{k+1}$, there exists $\vy^{(1)} \in R$ such that
  $\vy^{(1)}_{\kk} = \bz_{\kk}$ and $\vy^{(1)}(k+1) = s_{k+1}$.

  Let $t$ be any term of arity $\ell\geq 2$ that depends on at least two of its
  arguments, say, arguments $p$ and $q$, and
  consider $t(\vy^{(0)}, \dots, \vy^{(0)},\vy^{(1)}, \vy^{(0)}, \dots, \vy^{(0)})$, where
  $\vy^{(1)}$ appears as argument $p$ (or $q$) and $\vy^{(0)}$ appears elsewhere.
  By idempotence, and by the fact that $s_k$ and $s_{k+1}$ are sinks, we have
  \[
  t\left( \begin{array}{cccc}
    (r_0,\dots,r_{k-1},&s_k,&\ast,&\ast, \dots,\ast)\\
        \vdots & &&\vdots\\
    (r_0,\dots,r_{k-1},&s_k,&\ast,&\ast, \dots,\ast)\\
    (r_0, \dots, r_{k-1},&\ast,&s_{k+1},&\ast, \dots, \ast)\\
    (r_0,\dots,r_{k-1},&s_k,&\ast,&\ast, \dots,\ast)\\
        \vdots  && &\vdots\\
    (r_0,\dots,r_{k-1},&s_k,&\ast,&\ast, \dots,\ast)\\
  \end{array}\right) = (r_0, \dots, r_{k-1}, s_k, s_{k+1}, \ast, \dots, \ast),
  \]
  where the wildcard $\ast$ represents unknown elements.
  Denote this element of $R$ by
  $\br^{(1)} = (r_0, \dots, r_{k-1}, s_k, s_{k+1}, \ast, \dots, \ast)$.
  Continuing as above, we find
$\vy^{(2)} = (r_0, \dots, r_{k-1}, \ast, \ast, s_{k+2}, \ast, \dots, \ast) \in R$,
and compute
\begin{align*}
\br^{(2)}:= t(\br^{(1)}, \dots, \br^{(1)}, & \, \vy^{(2)}, \br^{(1)}, \dots, \br^{(1)}) =
(r_0, \dots, r_{k-1}, s_k, s_{k+1}, s_{k+2}, \ast, \dots, \ast),\\
&\; ^{\widehat{\lfloor}} \, \text{$p$-th argument}
\end{align*}
which also belongs to $R$.  In general, once we have
\begin{align*}
\br^{(j)}&:= (r_0,  \dots, r_{k-1}, s_k, \dots, s_{k+j}, \ast, \dots, \ast) \in R, \text{ and }\\
\vy^{(j+1)} &:= (r_0,  \dots, r_{k-1}, \ast, \dots, \ast, s_{k+j+1}, \ast, \dots, \ast) \in R,
\end{align*}
we compute
\begin{align*}
\br^{(j+1)} &= t(\br^{(j)},  \dots, \br^{(j)}, \vy^{(j+1)}, \br^{(j)}, \dots, \br^{(j)}) \\
              & = (r_0,  \dots, r_{k-1}, s_k, \dots, s_{k+j+1}, \ast, \dots, \ast) \in R.
\end{align*}
Proceeding inductively in this way yields $\bz =
(r_0, \dots, r_{k-1}, s_k, \dots, s_{n-1}) \in R$, as desired.
\end{proof}

}{}

\ifthenelse{\boolean{arxiv}}{

\subsection{Other elementary facts}
The remainder of this section collects some observations that can be useful when
trying to prove that an algebra is abelian.  We have moved the statements and
proofs of these facts to the appendix since we didn't end up using any of them
in the paper.

Denote the diagonal of $A$ by $D(A) := \{(a,a)\mid a \in A\}$.

\begin{lem}%
\label{lem:diagonal}
An algebra $\bA$ is abelian if and only if there is some $\theta \in \Con (\bA^2)$ that has
the diagonal set $D(A)$ as a congruence class.
\end{lem}
\begin{proof}
  ($\Leftarrow$) Assume $\Theta$ is such a congruence.  Fix
  $k<\omega$,
  $t^{\bA}\in \sansClo_{k+1}(\bA)$,
  $u, v \in A$, and
  $\bx, \vy \in A^k$.
  We will prove the implication~(\ref{eq:22}), which in the present context is
  \begin{equation*}
    t^\bA(\bx,u) = t^\bA(\vy,u) \quad \Longrightarrow \quad
    t^{\bA}(\bx,v) = t^{\bA}(\vy,v).
  \end{equation*}
  Since $D(A)$ is a class of $\Theta$, we have
  $(u,u) \mathrel{\Theta} (v,v)$, and since $\Theta$ is a reflexive relation, we have
  $(x_i,y_i)  \mathrel{\Theta} (x_i,y_i)$ for all $i$.  Therefore,
  \begin{equation}%
    \label{eq:9}
    t^{\bA\times \bA}((x_1,y_1), \dots, (x_k,y_k), (u,u))
    \mathrel{\Theta}
    t^{\bA\times \bA}((x_1,y_1), \dots, (x_k,y_k), (v,v)).
  \end{equation}
  since $t^{\bA \times \bA}$ is a term operation of $\bA\times \bA$.
  Note that~(\ref{eq:9}) is equivalent to
  \begin{equation}%
    \label{eq:13}
    (t^{\bA}(\bx, u), t^{\bA}(\vy,u))
    \mathrel{\Theta}
    (t^{\bA}(\bx, v), t^{\bA}(\vy, v)).
  \end{equation}
  If $t^{\bA}(\bx, u)= t^{\bA}(\vy, u)$ then
  the first pair in~(\ref{eq:13}) belongs to the $\Theta$-class
  $D(A)$, so the second pair must also belong this $\Theta$-class.
  That is, $t^{\bA}(\bx, v)= t^{\bA}(\vy, v)$, as desired.

  \vskip2mm

  \noindent ($\Rightarrow$) Assume $\bA$ is abelian. We show
  $\Cg^{\bA^2}(D(A)^2)$ has $D(A)$ as a block.  Assume
  \begin{equation}%
    \label{eq:16}
  ((x,x), (c,c')) \in \Cg^{\bA^2}(D(A)^2).
  \end{equation}
  It suffices to prove that $c=c'$.  Recall, \malcev's congruence generation
  theorem states that (\ref{eq:16}) holds iff
  \begin{align*}
  \exists \,& (z_0,z_0'), (z_1,z_1'), \dots, (z_n,z_n') \in A^2\\
    \exists \,& ((x_0,x_0'), (y_0,y_0')), ((x_1,x_1'), (y_1,y_1')), \dots,
    ((x_{n-1},x_{n-1}'), (y_{n-1},y_{n-1}')) \in D(A)^2\\
    \exists \, & f_0, f_1, \dots, f_{n-1}\in F^*_{\bA^2}
  \end{align*}
  such that
  \begin{align}%
    \label{eq:7}
    \{(x, x),(z_1,z_1')\} &= \{f_0(x_0,x_0'), f_0(y_0,y_0')\}\\
    \nonumber
     \{(z_1,z_1'),(z_2,z_2')\} &= \{f_1(x_1,x_1'), f_1(y_1,y_1')\}\\
    \nonumber
     & \vdots\\%
    \label{eq:8}
     \{(z_{n-1},z_{n-1}'),(c, c')\} &= \{f_{n-1}(x_{n-1},x_{n-1}'), f_{n-1}(y_{n-1},y_{n-1}')\}
  \end{align}
  The notation $f_i\in F^*_{\bA^2}$ means
  \begin{align*}
    f_i(x, x') &= g_i^{\bA^2}((a_1, a_1'), (a_2, a_2'), \dots, (a_k, a_k'), (x, x'))\\
               &= (g_i^{\bA}(a_1, a_2, \dots, a_k, x), g_i^{\bA}(a_1', a_2', \dots, a_k', x')),
  \end{align*}
  for some $g_i^{\bA} \in \sansClo_{k+1}(\bA)$ and some constants
  $\ba = (a_1, \dots, a_k)$ and $\ba' = (a_1', \dots, a_k')$ in $A^k$.
  Now, $((x_i,x_i'), (y_i,y_i'))\in D(A)^2$ implies
  $x_i=x_i'$, and $y_i=y_i'$, so in fact we have
  \[
    \{(z_i,z_i'),(z_{i+1},z_{i+1}')\} = \{f_i(x_i,x_i), f_i(y_i,y_i)\} \quad (0\leq i < n).
  \]
  Therefore, by Equation~(\ref{eq:7}) we have either
  \[
    (x,x)= (g_i^{\bA}(\ba, x_0), g_i^{\bA}(\ba', x_0)) \quad \text{ or } \quad
    (x,x)= (g_i^{\bA}(\ba, y_0), g_i^{\bA}(\ba', y_0)).
  \]
  Thus, either $g_i^{\bA}(\ba, x_0) =  g_i^{\bA}(\ba', x_0)$ 
  or $g_i^{\bA}(\ba, y_0) =  g_i^{\bA}(\ba', y_0)$.
  By the abelian assumption, if one of these equations holds, then so does the
  other. This and and Equation (\ref{eq:7}) imply $z_1 = z_1'$.  Applying the same
  argument inductively, we find that $z_i = z_i'$ for all $1\leq i < n$ and so, by
  (\ref{eq:8}) and the abelian property, we have $c= c'$.
\end{proof}

Lemma~\ref{lem:diagonal} can be used to prove the next result
which states that if there is a congruence of $\bA_1 \times \bA_2$ that has the
graph of a bijection between $A_1$ and $A_2$ as a block, then both $\bA_1$ and
$\bA_2$ are abelian algebras.

\begin{lem}%
  \label{lem:bijection_abelian}
  Suppose $\rho \colon A_0 \to A_1$ is a bijection and suppose the graph
  $\{(x, \rho x) \mid x \in A_0\}$ is a block of some congruence
  $\beta \in \Con (A_0 \times A_1)$.  Then both $\bA_0$ and $\bA_1$ are abelian.
\end{lem}
\begin{proof}
  Define the relation $\alpha\subseteq (A_1\times A_1)^2$ as follows: for
  $((a,a'), (b,b')) \in (A_1\times A_1)^2$,
  \[
  (a,a')\mathrel{\alpha} (b,b')
  \quad \iff \quad
  (a, \rho a') \mathrel{\beta} (b, \rho b')
  \]
  We prove that the diagonal $D(A_1)$ is a block of $\alpha$ by showing that
  $(a, a) \mathrel{\alpha} (b,b')$ implies $b = b'$.
  Indeed, if $(a, a) \mathrel{\alpha} (b,b')$, then
  $(a, \rho a) \mathrel{\beta} (b, \rho b')$, which means that
  $(b, \rho b')$ belongs to the block and
  $(a, \rho a)/\beta = \{(x, \rho x)\mid x\in A_1\}$.  Therefore,
  $\rho b  = \rho b'$, so $b = b'$ since $\rho$ is injective.
  This proves that $\bA_1$ is abelian.

  To prove $\bA_2$ is abelian, we reverse the roles of $A_1$ and $A_2$ in the
  foregoing argument.
  If $\{(x, \rho x) \mid x \in A_1\}$ is a block of $\beta$,
  then
  $\{(\rho^{-1}(\rho x), \rho x) \mid \rho x \in A_2\}$ is a block of $\beta$; that
  is, $\{(\rho^{-1} y, y) \mid y \in A_2\}$ is a block of $\beta$.  Define
  the relation $\alpha\subseteq (A_2\times A_2)^2$ as follows: for
  $((a,a'), (b,b')) \in (A_2\times A_2)^2$,
  \[
  (a,a')\mathrel{\alpha} (b,b')
  \quad \iff \quad
  (\rho^{-1}a, \rho a') \mathrel{\beta} (\rho^{-1}b, \rho b').
  \]
  As above, we can prove that the diagonal $D(A_2)$ is a block of $\alpha$
  by using the injectivity of $\rho^{-1}$ to show that $(a, a) \mathrel{\alpha}
  (b,b')$
  implies $b = b'$.
\end{proof}

\begin{lem}%
\label{lem:triv-clone-implies-abelian}
If $\sansClo(\bA)$ is trivial (i.e., generated by the projections),
then $\bA$ is abelian.
\end{lem}
(In fact, it can be shown that $\bA$ is \emph{strongly abelian} in this case.
We won't need this stronger result, and the proof that $\bA$ is abelian is elementary.)
\begin{proof}
We want to show $\C{1_A}{1_A}$.  Equivalently, we must show
that for all $t\in \sansClo(\bA)$ (say, $(\ell+m)$-ary)
and all $a, b \in A^\ell$, we have $\ker t(a,\cdot)=\ker t(b,\cdot)$.
We prove this by induction on the height of the term $t$.  Height-one terms are
projections and the result is obvious for these.  Let $n>1$ and assume the result
holds for all terms  of height less than
$n$.  Let $t$ be a term of height $n$, say, $k$-ary.  Then for some terms
$g_1, \dots, q_k$ of height less than $n$ and for some $j\leq k$, we have
$t = p^k_j [q_1, q_2, \dots, q_k] = q_j$ and since $q_j$ has height less than
$n$, we have
\[
\ker t(a,\cdot)=\ker g_j(a,\cdot) = \ker g_j(b,\cdot)=\ker t(b,\cdot).
\]\end{proof}
}{}

\ifthenelse{\boolean{extralong}}{

\section{Strongly Tractable Divisors}
Suppose $\bA$ is a finite idempotent algebra with congruence relation $\theta \in \Con(\bA)$.
Let $\sI$ be an $n$-variable instance of $\CSP(\bA)$ and suppose
$\sC = \{(\sigma_j, R_j) \mid 0\leq j < p\}$ is the (finite) set of constraints of $\sI$.
Recall, 
the quotient instance $\sI/\theta$ is the $n$-variable instance of $\CSP(\bA/\theta)$ with constraint set given in (\ref{eq:100}) above.
We say that $\bA/\theta$ is a \emph{strongly tractable divisor} of $\bA$ if there exist
constants $c$ and $d$ and an algorithm that, given an instance $\sI$ of $\CSP(\bA)$,
\begin{itemize}
\item determines the full set $\mathcal S$ of solutions to the quotient instance $\sI/\theta$, and
\item takes at most $c |\sI|^d$ steps to complete.
\end{itemize}
Here $|\sI|$ denotes the size of the instance $\sI$, which we take to
be the length of a string encoding all the scopes and tuples of constraints
of $\sI$.
We use the adjective ``strongly'' in the definition above because the algorithm must determine
the full set of solutions to $\sI/\theta$, as opposed to just deciding whether or not
there is a solution.
On the other hand, the bound on the running time of the
algorithm is $c |\sI|^d$ and not $c|\sI/\theta|^d$, so this notion of tractability is maybe not
quite as strong as it first appears.

\subsection{Application: strongly tractable atop tractable}
Let $\bA$ be a finite idempotent algebra with a congruence $\theta$ whose
blocks belong to a tractable variety $\var{V}$, and
suppose that $\bA/\theta$ is a strongly tractable divisor of $\bA$.
Our goal is to prove $\CSP(\bA)$ is tractable.

As a guide to intuition, we imagine applying the algorithm described in this section to situations
in which the number of congruence classes of $\theta$---and therefore the number of
solutions to quotient instances---is relatively small.  To this end, we
call $\theta\in \Con(\bA)$ a \emph{coarse congruence}
if there are constants $c$ and $d$ such that, for every instance $\sI$ of $\CSP(\bA)$,
the number of solutions to the quotient instance $\sI/\theta$ is bounded above by $c|\sI|^d$.

Recall Fact~\ref{fact:soln-quotient} above which says that if the quotient instance
$\sI/\theta$ has no solution, then $\sI$ has no solution.
This and the notion of a ``coarse congruence'' suggest the following algorithm for solving $\CSP(\bA)$.
Given an instance $\sI$ of $\CSP(\bA)$,
\begin{enumerate}
\item compute the solution set $\mathcal S$ of the quotient instance $\sI/\theta$ of $\CSP(\bA/\theta)$;
\item for each solution $\bx/\theta\in \mathcal S$ and corresponding block instance $\sI_{\bx}$,
check 
whether there is a solution to $\sI_{\bx}$.
(If $\sI_{\bx}$ has a solution, then so does $\sI$---in fact, a solution to $\sI_{\bx}$ \emph{is} a solution to $\sI$.)
\end{enumerate}
If for all $\bx/\theta \in S$ the block restricted instance $\sI_{\bx}$ has no solution, then $\sI$ has no solution.

The number of steps in stage (1) of the algorithm is bounded by a polynomial in $|\sI|$
since we assumed $\bA/\theta$ is a strongly tractable divisor of $\bA$.
In stage (2), for each solution $\bx/\theta \in \mathcal S$,
determining whether there is a solution to $\sI_{\bx}$ has polynomial time complexity since
$\CSP(\bB_{\chi(0)}, \bB_{\chi(1)}, \dots, \bB_{\chi(n-1)})$ is tractable.  Therefore,
if we can find some constants $c$ and $d$ such that the number $|\mathcal S|$ of solutions to the quotient instance
$\sI/\theta$ is bounded above by $c|\sI|^{d}$, then stage (2) can be completed in
polynomial
  \ifthenelse{\boolean{footnotes}}{%
    time.\footnote{To be a little
      more precise, the time complexity of stage (2) of the algorithm is bounded
      by $c|\sI|^{d} \cdot c' |\sI|^{d'} \approx C|\sI|^{D}$, where $c'$ and $d'$ are found
      by bounding the complexity of
      $\CSP(\bB_{\chi(0)}, \bB_{\chi(1)}, \dots, \bB_{\chi(n-1)})$ for a
      ``worst-case product'' of $\theta$ classes.}
  }{time.}
The next proposition summarizes what we have shown.
\begin{prop}%
  \label{prop:chain-atop-tractable}
  Suppose $\bA$ is a finite idempotent algebra with $\theta \in \Con(\bA)$ satisfying the following conditions:
  \begin{itemize}
  \item  $\bA/\theta$ is strongly tractable;
  \item  $\theta$ blocks belong to a tractable variety;
  \item  $\theta$ is coarse.
  \end{itemize}
  Then $\CSP(\bA)$ is tractable.
\end{prop}
}{}

\bibliographystyle{alphaurl}
\bibliography{\jobname}

\end{document}